\documentclass[prd,twocolumn,showpacs,superscriptaddress,nofootinbib]{revtex4-1}

\usepackage{amsfonts}
\usepackage{amsmath}
\usepackage{amssymb}
\usepackage{bm}
\usepackage{dcolumn}
\usepackage{graphicx}
\usepackage{graphics}
\usepackage[utf8]{inputenc}
\usepackage{latexsym}
\usepackage{rotating}
\usepackage[colorlinks=true]{hyperref}
\usepackage[all]{hypcap}
\usepackage{xspace} 
\usepackage[usenames]{color}
\usepackage{mathrsfs}
\usepackage{amsthm}
\usepackage{multirow}
\usepackage{pifont}

\graphicspath{{Plots/}}

\usepackage{tikz}
\usetikzlibrary{positioning,shapes,shadows,arrows}

\usepackage{ulem}
\definecolor {darkgreen}{rgb}{0.2,0.7,0.2}

\newcommand\be{\begin{equation}}
\newcommand\ba{\begin{eqnarray}}
\newcommand\ee{\end{equation}}
\newcommand\ea{\end{eqnarray}}

\newcommand\bw{\begin{widetext}}
\newcommand\ew{\end{widetext}}

\newcommand{\nn}{\nonumber}

\newcommand{\ppE}{{\mbox{\tiny ppE}}}

\newcommand{\GR}{{\mbox{\tiny GR}}}

\newcommand{\MAT}{{\mbox{\tiny mat}}}

\newcommand{\CS}{{\mbox{\tiny CS}}}

\newcommand{\pont}{{\,^\ast\!}R\,R}
\newcommand{\sRsR}{{\,^\ast\!}R{\,^\ast\!}R}

\newcommand{\pd}{\partial}
\newcommand{\cd}{\nabla}

\newcommand{\BH}{{\mbox{\tiny BH}}}
\newcommand{\NS}{{\mbox{\tiny NS}}}

\newcommand{\BL}{{\mbox{\tiny BL}}}
\newcommand{\HH}{{\mbox{\tiny H}}}

\newcommand{\ext}{\mathrm{ext}}

\newcommand{\mrm}{\mathrm}

\newcommand{\GB}{{\mbox{\tiny GB}}}
\newcommand{\EH}{{\mbox{\tiny EH}}}

\newcommand{\TEDGB}{{\mbox{\tiny TEdGB}}}

\newcommand{\theGB}{\ensuremath{\mathcal{G}}}

\newcommand{\DDGB}{{\mbox{\tiny D$^{2}$GB}}}
\newcommand{\DDCS}{{\mbox{\tiny D$^{2}$CS}}}

\newcommand{\DDGBText}{{\mbox{D$^{2}$GB}}}
\newcommand{\DDCSText}{{\mbox{D$^{2}$CS}}}

\newcommand{\p}{{\mbox{\tiny P}}}
\newcommand{\K}{{\mbox{\tiny K}}}

\newcommand{\specialEq}{\ensuremath{(\star)}}

\newtheorem{theorem}{Theorem}

\begin{document}
\title{Challenging the Presence of Scalar Charge and Dipolar Radiation
  in Binary Pulsars}

\author{Kent Yagi}
\affiliation{Department of Physics, Princeton University, Princeton, NJ 08544, USA.}
\affiliation{Department of Physics, Montana State University, Bozeman, MT 59717, USA.}

\author{Leo C. Stein}
\affiliation{TAPIR, Walter Burke Institute for Theoretical Physics, MC 350-17, California Institute of Technology, Pasadena, CA 91125, USA}
\affiliation{Cornell Center for Astrophysics and Planetary Science (CCAPS), Cornell University, Ithaca, NY 14853, USA}

\author{Nicol\'as Yunes}
\affiliation{Department of Physics, Montana State University, Bozeman, MT 59717, USA.}

\date{\today}

\begin{abstract}
Corrections to general relativity that introduce long-ranged scalar
fields which are non-minimally coupled to curvature typically predict
that neutron stars possess a non-trivial scalar field profile anchored
to the star.
An observer far from a star is most sensitive to the
spherically-symmetric piece of this profile that decays linearly with
the inverse of the distance to the source, the so-called \emph{scalar
  monopole charge}, which is related to the emission of dipolar
radiation from compact binary systems.  The presence of dipolar
radiation has the potential to rule out or very strongly constrain
extended theories of gravity.
These facts may lead people to believe
that gravitational theories that introduce long-ranged scalar fields
have already been constrained strongly from
binary pulsar observations.
Here we challenge this ``lore'' by investigating the decoupling limit of
Gauss-Bonnet gravity as an example, in which the scalar field couples 
linearly to the Gauss-Bonnet density in the action.  
We prove a theorem that neutron
stars in this theory can not possess a scalar charge, due to the
topological nature of the Gauss-Bonnet density.  Thus Gauss-Bonnet
gravity evades the strong binary pulsar constraints on dipole
radiation.
We discuss the astrophysical systems which will yield the best
constraints on Gauss-Bonnet gravity and related quadratic gravity
theories.
To achieve this we compute the scalar charge in quadratic gravity
theories by performing explicit analytic and numerical matching
calculations for slowly-rotating neutron stars.
In generic quadratic gravity theories, either neutron~star-binary or
neutron~star-black~hole systems can be used to constrain the theory,
but because of the vanishing charge, Gauss-Bonnet gravity evades the
neutron star-binary constraints.
However, in contrast to neutron~stars, black
holes in Gauss-Bonnet gravity do anchor scalar charge, because of the
difference in topology.
The best constraints on Gauss-Bonnet gravity will thus come from
accurate black hole observations, for example through gravitational
waves from inspiraling binaries or the timing of pulsar-black hole
binaries with radio telescopes.  We estimate these constraints to be a
factor of ten better than the current estimated bound, and also include
estimated constraints on generic quadratic gravity theories from
pulsar timing.
\end{abstract}



\maketitle

\section{Introduction}

Almost a century since Einstein's discovery of his general theory
of relativity (GR), we continue to test it and wonder whether it is right
in unexplored regimes. Perhaps the most famous of these tests
are those carried out in the solar system~\cite{will-living}, which
not only confirmed the theory initially, but also served to rule out
a plethora of modified models in the 1970s. The solar system, 
however, is an unsuitable place to test the strong, dynamical,
and non-linear features of the gravitational interaction~\cite{Yunes:2013dva}. Gravity
is simply too weak and the velocities of planets are simply
too small relative to the speed of light within the solar system.

On top of that, astrophysical observations and theoretical studies
have sparked a renewed interest in corrections to GR. The observations
of galactic rotation curves have been used to pose certain modifications
to Newtonian dynamics~\cite{Milgrom:1983ca,Sanders:2002pf,Bekenstein:2004ne,
Famaey:2011kh}.
Corrections to GR on cosmological length scales have been invoked to
attempt to explain the observation of late-time acceleration of the
universe~\cite{2010deto.book.....A,DeFelice:2010aj,deRham:2014zqa}.
Quantum gravitational theories, like string theory, will induce
corrections to GR in their low-energy 
effective theories~\cite{polchinski1,polchinski2, Fujii:2003pa}.
The low-energy effective theories of loop quantum gravity and
inflation, for example, predict corrections to GR in the form of
curvature-squared and higher operators~\cite{Jackiw:2003pm, CSreview,
  alexandergates, Taveras:2008yf, Mercuri:2009zt, Weinberg:2008hq}.

Enter binary pulsars. With some of the strongest gravitational fields in the universe,
neutron stars (NSs) are spectacular laboratories to test strong gravity. Neutron
stars can rotate at fantastic speeds, sometimes with millisecond spin periods, and
their emission can be detected as pulses in the radio band. When in a binary
system, the pulses' arrival times are modulated, encoding rich information about the properties
of the orbit, and thus, of the nature of gravity. Since the discovery
of the Hulse-Taylor binary pulsar, many others have been used to stringently constrain 
any deviations from the predictions of Einstein~\cite{Damour:1996ke,Stairs:2003yja,kramer-double-pulsar,Freire:2012mg,Wex:2014nva}. Some of these consist of 
two NSs, like the Hulse-Taylor binary, while others consist
of a white dwarf (WD) and a NS.  Typical orbital periods
are on the order of several hours to days.

Over the past 50 years, binary pulsars have been used to invalidate many a modified gravity
theory, and some people
in the community may think that
modified gravity theories with
long-ranged scalar fields can be very well-constrained with binary pulsar observations because they
lead to the emission of dipolar radiation in binary systems. Such radiation is absent in GR because
of the conservation of certain Noether charges (a curved-space version of mass and linear momentum).
Long-ranged scalar fields in modified gravity do not typically
have such conservation laws, and thus, they may carry energy and
momentum away from the binary system in a monopolar or dipolar fashion
if excited.  If so, the orbital period decay would be much faster than
that predicted in GR.  An observation consistent with GR would thus
lead to a stringent constraint on such a theory.

In this paper, we show explicitly and in great detail that this lore is not always correct.
The presence or absence of dipole
radiation in modified gravity when modeling binary systems depends
sensitively on the structure of the scalar
field that is excited by the compact objects that form the binary. Far away from the compact objects, 
the scalar field can be decomposed in spherical harmonics, and the spherically symmetric ($\ell=0$),
$1/r$ fall-off shall be called the \emph{scalar monopole charge} or \emph{scalar charge} for short.
When a member of a binary system possesses such a scalar charge,
dipole radiation is typically excited;
when scalar charge is absent in both binary constituents, 
dipole radiation is heavily
suppressed~\cite{Yagi:2011xp, Stein:2013wza}.\footnote{%
  The dipole radiation is suppressed in the post-Newtonian (PN) sense: it may be
  present, but will appear at a higher order in powers of $v/c$ than
  one would naively expect, which would be $-1$PN relative to the GR
  quadrupolar radiation.  This is distinct from the suppression of
  dipole radiation in a NS/NS system in Brans-Dicke type theories,
  which comes about because the dipolar radiation depends on the
  difference in sensitivities of the two bodies.
}
For example, such a
suppression is present in a particular quadratic gravity theory,
dynamical Chern-Simons (dCS) gravity~\cite{Jackiw:2003pm, CSreview}.
In dCS, the scalar charge of an isolated body is suppressed by parity
considerations, and therefore dipolar radiation from a binary is also
heavily suppressed~\cite{kent-CSGW,Yagi:2013mbt}.
A similar absence of the stellar scalar charge and the suppression of
dipolar radiation of stellar binaries are found in shift-symmetric
Horndeski theories~\cite{Barausse:2015wia}, a certain class of generic
scalar-tensor theories with up to second derivatives in the field
equations. Shift-symmetric theories
are those whose field equations remain unchanged upon a constant
shift of the scalar field

We formalize the above through a \emph{miracle hair loss} conjecture
and theorem.  This conjecture is aimed at establishing that
gravity theories with a shift-symmetric, long-ranged scalar
sourced by a linear (non-derivative) coupling to a topological density do not
activate a scalar charge in NSs. A topological density is one which
when integrated over the manifold yields a topological invariant.  We
rigorously prove such a theorem for a particular modified gravity
theory that has the above properties: dynamical Gauss-Bonnet (DGB)
gravity in the so-called decoupling limit~\cite{Yunes:2011we} (here
abbreviated \DDGBText{}).  This theory modifies
the action by adding a term that is the product of a dynamical scalar field with the Gauss-Bonnet topological density. 
The integral of the latter is a topological invariant, i.e.~the Euler characteristic of the manifold, and the theory
is manifestly shift-symmetric.  Such a theorem has long-ranging
implications, since the absence of a scalar charge in isolated neutron
stars automatically implies that dipolar radiation is highly
suppressed in binary systems.  This renders pulsar binaries
ineffective at constraining such theories.

We should point out that this theorem only applies to ordinary stars, such as
NSs, and not to black holes (BHs).
For example, the scalar charge of a BH in \DDGBText{} gravity
  has been explicitly calculated and found to be non-vanishing 
in~\cite{1992PhLB..285..199C,Mignemi:1992nt,
Yunes:2011we,Yagi:2011xp,Sotiriou:2014pfa}.
This situation
is in direct contrast to the no-hair theorems of scalar-tensor
theories~\cite{hawking-no-hair,Sotiriou:2011dz,Graham:2014ina} of the
Jordan-Brans-Dicke variety (or, more generally, the Bergmann-Wagoner
type~\cite{Bergmann:1968ve, Wagoner:1970vr}; we will continue to refer
to this class of theories as simply ``scalar-tensor'' theories).
These theorems prove that in scalar-tensor theories, stationary,
isolated BHs in vacuum have no scalar hair,
and thus, in particular no scalar charge. Possible ways 
 to grow BH scalar hair are to introduce a potential for the scalar
field~\cite{Healy:2011ef} or to impose certain cosmological boundary conditions~\cite{Jacobson:1999vr,Horbatsch:2011ye}. 
The latter has been dubbed Jacobson's \emph{miracle hair growth} formula for BHs in scalar-tensor
theories~\cite{Jacobson:1999vr,Horbatsch:2011ye}. On the other hand, NSs in scalar-tensor theories do generically 
possess a scalar charge, as this is sourced by the matter stress-energy tensor~\cite{Will:1989sk,Will:1994fb,Damour:1998jk,Alsing:2011er,Sampson:2014qqa}. 
It is this fact that makes scalar-tensor theories easy to constrain with binary pulsar observations. 
A similar absence of black hole scalar charges has recently been shown
in shift symmetric Horndeski theories except for \DDGBText{} gravity in~\cite{Hui:2012qt,Maselli:2015yva}.

{
\newcommand{\yes}{Yes}
\newcommand{\no}{\bf{No}}
\renewcommand{\yes}{\checkmark}
\renewcommand{\no}{\ding{55}}

\newcommand{\minitab}[2][l]{\begin{tabular}{#1}#2\end{tabular}}

\renewcommand{\arraystretch}{1.2}
\begin{table*}[tb]
\begin{centering}
\begin{tabular}{r|c|c|c|c|c|c|c}
\hline
\hline
\noalign{\smallskip}
\multirow{2}{*}{Theory}  & \multicolumn{2}{c|}{Scalar charge} &
 \multicolumn{4}{c|}{Estimated upper bound on $\sqrt{|\alpha|}$ [km]} &
 \multirow{2}{*}{\minitab[c]{Current estm. \\ bound [km]}} \\
\cline{2-7}
& NS & BH & NS/WD & NS/NS & NS/BH & BH/BH &  \\
\hline
\DDGBText{} & \no & \yes & \no  & \no & (0.12) & (3.4) & 1.9~\cite{Yagi:2012gp} \\
TEdGB & \yes & \yes & 1--2 & 1.5--3.5 & (0.12) & (3.4)  & 1.4~\cite{Kanti:1995vq,Pani:2009wy} \\
Kretsch. & \yes & \yes & 0.06--0.1 & 0.15--0.45 & (0.03--0.07) & (3.4) &  1.9~\cite{Yagi:2012gp} \\
\noalign{\smallskip}
\hline
\hline
\end{tabular}
\end{centering}
\caption{%
Summary of whether certain modified theories activate a scalar charge
in NSs and BHs, and the estimated bounds on the
coupling parameter $\sqrt{|\alpha|}$ from a variety of systems.  We
consider truncated Einstein-dilaton-Gauss-Bonnet gravity (TEdGB) and the
decoupling limit of any Gauss-Bonnet theory (\DDGBText{}), as well as a
certain quadratic gravity theory in which the scalar field does not
couple to a topological density (Kretschmann).
Quantities in parentheses represent projected bounds using
systems that have not yet been observed, but which may be observed in
the near future with radio/GW observations.
Note that a NS/BH system could produce the best constraint within each
theory.
}
\label{tab:summary}
\end{table*}
}

We demonstrate this theorem at work by computing the scalar charge explicitly in \DDGBText{},
as well as in theories that violate the conditions of the theorem, wherein the scalar charge is non-vanishing.
We first compute the scalar charge analytically for slowly-rotating NSs in a post-Minkowskian expansion, 
i.e.~an expansion in the ratio of the stellar mass to its radius (the stellar compactness) and with simple equations of state. 
We then compute the scalar charge numerically without a post-Minkowskian expansion and with more realistic equations
of state.  We verify that in the limit where quadratic gravity approaches
\DDGBText{}, the scalar charge vanishes linearly with the coupling
constants.
Theories which don't satisfy the conditions of the theorem need not
have a particularly large or small scalar charge.  In fact in our
explicit calculations, we show examples of strong dependence on
coupling parameters and the compactness of a body.  This may lead to
suppression of scalar charge in some theories, though it is still
present.

The miracle hair loss theorem/conjecture is so powerful that it allows us to easily predict which binary systems will be best
to test which theory. 
Generically, if scalar charge is activated, then the observation of
WD/NS or NS/NS pulsar binaries
with radio telescopes is already sufficient to strongly constrain the
particular modified theory.
On the other hand, if the conjecture is applicable, and scalar charge
is not activated in NSs but is activated in BHs, one
requires BH binaries or mixed BH/NS binaries
to place a constraint.
The former could be detected through their gravitational waves (GWs) by ground-based interferometers, such as LIGO and 
Virgo~\cite{ligo,Abramovici:1992ah,Abbott:2007kv,virgo,Giazotto:1988gw,TheVirgo:2014hva},
while the latter may be detected in future radio telescopes, such as the Square
Kilometer Array~\cite{Carilli:2004nx}. The conjecture,
for example, can easily be applied to quadratic gravity theories
which \emph{naturally} avoid solar system constraints (since quadratic
curvature densities are
small in the solar system)~\cite{EspositoFarese:2003ze,
Amendola:2007ni,Sotiriou:2006pq,Alexander:2007zg,Alexander:2007vt,Smith:2007jm},
yet predict strong modifications to GR when the curvature is
large, such as around NSs and BHs~\cite{Yagi:2011xp,kent-LMXB,kent-CSGW,I-Love-Q-Science,
  I-Love-Q-PRD,sopuerta-yunes-DCS-EMRI,pani-DCS-EMRI,
  canizares,harko,amarilla,chen}. Table~\ref{tab:summary} shows a few examples of such theories, 
whether they activate a scalar charge in NSs and BHs, 
which systems are best to constrain them and an estimate of how well they can be constrained.

The remainder of this paper deals with the details of the results
explained above.
Section~\ref{sec:ABC-MQG} reviews the basics of quadratic gravity.
Section~\ref{sec:Hair-Loss-Theorem} states and proves the miracle
hair loss theorem.
Section~\ref{sec:analytics} presents an analytical and a numerical demonstration
of the theorem.
Section~\ref{sec:bin-puls-cons} estimates binary pulsar and
GW constraints due to dipolar radiation.
Section~\ref{sec:conclusions} concludes and points to future research.
All throughout, we use geometric units in which $G = 1 = c$.

\section{The ABC of Quadratic Gravity}
\label{sec:ABC-MQG}
Here we review and classify quadratic gravity theories, mainly
following the presentation of~\cite{CSreview,Yunes:2011we,
Stein:2013wza}.  We begin with a description of the motivation
for studying such theories. This will set up the stage for the
introduction of our classification. We then conclude with a discussion
of well-posedness in these theories and current constraints. As we
classify theories, we will come across a few theories that we will
investigate in detail in this paper; the actions that define these
theories will have their equations marked with $\specialEq$.

\subsection{Motivation}

Our focus is theories that correct GR through
higher-curvature terms in the action and include long-ranged scalar fields.
The scalar field of interest must be of ``gravitational strength,''
i.e.~should couple to matter weakly, although it is allowed to couple
directly to curvature. As such, these theories will be \emph{metric}, with
matter coupling directly only to the metric tensor, and the scalar field coupling
to the metric through curvature, and thus indirectly to matter.

These theories are motivated from at least two places.  First, from
the modern bottom-up, effective field theory (EFT) standpoint, we should
expect GR to acquire corrections at some length scale. 
At low energies, these potential corrections can be described at
the level of the action via an expansion in powers of the Riemann
curvature tensor.  If truncating at first order in curvature, we
recover GR, with the Einstein-Hilbert action
\begin{align}
S_{\EH} &= \kappa \int d^4x \sqrt{-g}  R\,,
\label{eq:EH-action}
\end{align}
where $g$ is the determinant of $g_{ab}$,
$R=g^{ab}R_{ab}=g^{ab}R_{acb}{}^{c}$ is the Ricci scalar, and
$\kappa=(16\pi G)^{-1}$.

If we include scalar fields, then at this order we for example arrive at
Jordan-Brans-Dicke theory~\cite{Jordan:1959eg,Brans:1961sx},
where the scalar field is linearly coupled to the Ricci scalar in the Jordan frame. 
Jordan-Brans-Dicke theory allows for
gross modifications from GR in the form of long-ranged scalar charges
and scalar dipole radiation, which are strongly constrained from solar
system and binary pulsar observations~\cite{will-living,Berti:2015itd}.
Furthermore, it is always possible, through a conformal
transformation~\cite{Bloomfield:2011np}, to go to the Einstein
frame of the theory, where the linear-in-curvature term in the
action is simply the Einstein-Hilbert term of Eq.~\eqref{eq:EH-action}. At
next order in curvature, we arrive at quadratic gravity theories,
which are the topic of this paper.

The second motivation for such theories is from the top-down,
high-energy theory viewpoint.  Fundamental theories of
quantum gravity (such as string theory and loop quantum gravity) will
induce both higher-corrections to GR and a number of
scalar fields, which may be
long-ranged~\cite{Gross:1986mw,Metsaev:1987zx,Mignemi:1992nt,
Mignemi:1993ce,Kleihaus:2011tg,Taveras:2008yf,Alexander:2008wi}. For
example, in heterotic string theory and in the string frame in $D$ dimensions, 
the next-to-leading order
correction to GR in a low-curvature expansion (first derived
in~\cite{Metsaev:1987zx}) is given (in the notation
of~\cite{Maeda:2009uy}) by
\begin{align}
\label{eq:EDGB-theory}
S = \frac{1}{2\kappa_{D}^{2}} \int &d^{D}\hat{x} \sqrt{|\hat g|} e^{-2\hat{\phi}} \left[
\hat{R} + 4(\hat{\pd}\hat{\phi})^{2}  \right. \\
&\left.
+{\textstyle\frac{\alpha'}{8}}\left(
\hat{R}^{2}-4 \hat{R}_{ab}\hat{R}^{ab} + \hat{R}_{abcd}\hat{R}^{abcd} + \ldots
\right)
\right]\,, \nn
\end{align}
where $\ldots$ stand for higher-order in curvature terms. 
Here, $\phi$ is the dilaton, $\alpha'$ is the Regge slope parameter,
$\kappa_D^2$ is the $D$-dimensional 
gravitational strength, and $\hat{}$ represents a quantity in the string frame.  The theory described
by the action in Eq.~\eqref{eq:EDGB-theory} is referred to as \emph{Einstein-dilaton-Gauss-Bonnet} (EdGB).  
Through a conformal transformation and a field redefinition, this can be cast
in the Einstein frame (in the notation of~\cite{Maeda:2009uy}) as
\begin{align}
S &= \frac{1}{2\kappa_{D}^{2}} \int d^{D}x \sqrt{|g|} \left[
R - \frac{1}{2}(\pd\phi)^{2}  \right. \\
&\left.
\qquad + {\textstyle\frac{\alpha'}{8}} e^{-\gamma\phi} \left(
R^{2}-4 R_{ab}R^{ab} + R_{abcd}R^{abcd} + \ldots
\right)
\right]\,, \nn
\end{align}
where $\gamma = \sqrt{2/(D-2)}$ and again the $\ldots$ stand for higher order terms, but in both
curvature and the scalar field.  If such higher order terms
are dropped from the action, the resulting theory is
called \emph{truncated Einstein-dilaton-Gauss-Bonnet}
(TEdGB)~\cite{Maeda:2009uy}.

In string theory, one can work in
dimensions higher than 4, but henceforth in this paper, we will 
focus only on theories that have already been compactified to 
4 dimensions. 
This compactification introduces a large number of dynamical degrees
of freedom (moduli fields), such as the dilaton and axion(s).  The
resulting low-energy effective action may be truncated to a specific
operator order, and this truncation may affect the field content of
the effective action.  Performing this truncation consistently is not
trivial~\cite{Burgess:2007pt}.  Here we focus only on the metric and
long-ranged scalar sector of such a theory.

\subsection{Classification}

Let us define \emph{quadratic gravity} theories through the action
\begin{equation}
  \label{eq:S-full-action}
  S = S_{\EH} + S_{\MAT} + S_{\vartheta} + S_{q}\,,
\end{equation}
where the Einstein-Hilbert term $S_{\EH}$ was given in
Eq.~\eqref{eq:EH-action}, $S_{\MAT}$ is the action of any matter
fields that do not depend on the scalar field, and
$S_{\vartheta}$ is the action for a canonical scalar field with
potential $U$,
\begin{equation}
  \label{eq:S-vartheta}
  S_{\vartheta} = -\frac{1}{2} \int d^{4}x\sqrt{-g}
  \left[(\cd_{a}\vartheta)(\cd^{a}\vartheta) + 2U(\vartheta)\right]\,.
\end{equation}
This form for the action is always possible in an appropriate conformal
frame and through field redefinitions~\cite{Bloomfield:2011np}.  The quadratic
part of quadratic gravity comes from
\begin{equation}
  S_{q} = \int d^{4}x\sqrt{-g} \, F[R_{abcd}, \vartheta, \pd R_{abcd}, \pd\vartheta, \ldots] \,,
\end{equation}
where $F[\cdot]$ is the \emph{interaction density} between the scalar
field and the curvature, which must be homogeneous of degree 2 in the
curvature tensor (and its derivatives).  This means that for any real
constant $\lambda$,
\begin{equation}
F[\lambda R_{abcd}, \ldots] =
\lambda^{2} F[R_{abcd}, \ldots ]\,.
\end{equation}
This property does not allow the interaction density to depend on 
terms independent of or linear in the Riemann tensor.

The field equations of quadratic gravity can be obtained by varying the full action with 
respect to the metric tensor and the scalar field. The latter leads to the scalar evolution equation
\begin{equation}
 \square \vartheta - U'(\vartheta) = - \frac{\pd F}{\pd \vartheta}
 + \cd_{a} \frac{\pd F}{\pd (\cd_{a}\vartheta)} - \ldots\,,
\label{eq:evol-scal-eq}
\end{equation}
where the right-hand side is minus the variational derivative of
$S_{q}$ with respect to $\vartheta$.
This is a particularly simple wave equation when $U=0$ and when $F$ is linear in
the scalar field (and its derivatives).

The space of quadratic gravity theories is spanned by the functional degree
of freedom in the interaction density, which makes this space extremely large.
There are several non-exclusive ways to classify the types of
interaction densities that may appear within $S_{q}$ in ways that are
relevant to the phenomenology of the theories. We now provide 
a partial classification on the basis of three properties:
\begin{itemize}
\item having derivative or non-derivative interactions,
\item coupling to a topological density or not, and
\item possessing a shift symmetry or not.
\end{itemize}
This classification is summarized in Fig.~\ref{fig:quad-theory-classif}, a figure
we will return to in the following subsections when we define and discuss each of 
the above items in detail.

\tikzstyle{boxstyle}=[rounded corners=4mm, draw=black, rectangle, thick]
\tikzstyle{nodestyle}=[align=center]
\tikzstyle{linestyle}=[dashed,->,>=stealth',semithick]

\newcommand{\boxunit}{3.6cm}
\newcommand{\boxsep}{0.4cm}
\newcommand{\boxh}{2*(\boxunit+\boxsep)}

\begin{figure}[tb]
  \centering
  \begin{tikzpicture}

    \draw [draw=white] (-\boxunit-\boxsep,-\boxunit-\boxsep) rectangle
    (\boxunit+\boxsep,\boxunit+\boxsep);

    \begin{scope}[transform canvas={xshift=-\boxunit/2-\boxsep/2}]
      \node [style=boxstyle,minimum width=\boxunit,
      minimum height=\boxh,label={Topological density}] (topol) {};
    \end{scope}
    \begin{scope}[transform canvas={xshift=\boxunit/2+\boxsep/2}]
      \node [style=boxstyle,minimum
      width=\boxunit,minimum height=\boxh,label={Non-topological}] (non-topol) {};
    \end{scope}
    \begin{scope}[transform canvas={rotate=90,
        xshift=\boxunit/2+\boxsep/2,yshift=0}]
      \node [style=boxstyle,minimum width=\boxunit,
      minimum height=\boxh,label={Derivative interaction}] (deriv) {};
    \end{scope}
    \begin{scope}[transform canvas={rotate=90,
        xshift=-\boxunit/2-\boxsep/2,yshift=0}]
      \node [style=boxstyle,minimum width=\boxunit,
      minimum height=\boxh,label={Non-derivative}] (non-deriv) {};
    \end{scope}
    \begin{scope}[transform canvas={}]
      \node [style=boxstyle,minimum width=1.25*\boxunit,
      minimum height=1.25*\boxunit,
      label={[yshift=1em,fill=white]-90:Shift symmetric}] (shift) {};
    \end{scope}


    \node [style=nodestyle] (topol-non-deriv-shift)
    at (-\boxsep-0.25*\boxunit,-\boxsep-0.25*\boxunit) 
    {$\specialEq$ $\alpha\vartheta \theGB$ \\ $\alpha \vartheta \pont$};

    \node [style=nodestyle] (topol-deriv-shift)
    at (-\boxsep-0.225*\boxunit,\boxsep+0.25*\boxunit)
    {$\alpha(\square\vartheta) \pont$};

    \node [style=nodestyle] (non-topol-non-deriv-shift)
    at (\boxsep+0.25*\boxunit,-\boxsep-0.25*\boxunit)
    {$\alpha\sin(\omega \vartheta) K$};

    \node [style=nodestyle] (non-topol-deriv-shift)
    at (\boxsep+0.225*\boxunit,\boxsep+0.25*\boxunit)
    {$\alpha(\square\vartheta) K$};

    \node [style=nodestyle] (topol-non-deriv)
    at (-\boxsep-0.45*\boxunit,-\boxsep-0.75*\boxunit)
    {$\specialEq$ $\alpha e^{-\gamma\vartheta} \theGB$\\
     $\alpha\vartheta \pont + \frac{1}{2}m^{2}\vartheta^{2}$};

    \node [style=nodestyle] (topol-deriv)
    at (-\boxsep-0.45*\boxunit,\boxsep+0.75*\boxunit)
    {$\alpha(\square\vartheta) \pont + \frac{1}{2}m^{2}\vartheta^{2}$};

    \node [style=nodestyle] (non-topol-non-deriv)
    at (\boxsep+0.45*\boxunit,-\boxsep-0.75*\boxunit)
    {$\specialEq$ $\alpha \vartheta K$};

    \node [style=nodestyle] (non-topol-deriv)
    at (\boxsep+0.45*\boxunit,\boxsep+0.75*\boxunit)
    {$\alpha (\square\vartheta) K + \frac{1}{2}m^{2}\vartheta^{2}$};

    \draw [style=linestyle] (topol-non-deriv) edge [bend left=10]
    (topol-non-deriv-shift);

    \draw [style=linestyle] (topol-deriv) edge [bend right=10]
    (topol-deriv-shift);

    \draw [style=linestyle] (non-topol-deriv) edge [bend left=10]
    (non-topol-deriv-shift);

  \end{tikzpicture}
  \caption{
    A classification of corrections to GR which are at most quadratic
    in curvature and couple to a single dynamical scalar field $\vartheta$.
    The interaction density may explicitly include
    derivatives (top half), or may be a non-derivative
    interaction (bottom half).  The combination of curvature
    tensors to which the scalar is coupled may be of a topological
    nature (like the Pontryagin or Gauss-Bonnet scalars; left half),
    or may be unrelated to topological invariants (right half).  In
    all of these sectors, an interaction may enjoy
    a continuous or discrete shift symmetry (inner region), or it may
    not (outer region).  Some of these interactions are limiting cases
    of others, e.g.~when an explicit mass vanishes $m\to 0$, a shift
    symmetry can be acquired. The theories we investigate in detail 
    in this paper are marked with $\specialEq$.
    }
  \label{fig:quad-theory-classif}
\end{figure}

\subsubsection{Derivative and Non-derivative Interactions}

For the purposes of this paper, a \emph{derivative interaction}
is one that depends on at least one derivative of the scalar field 
or of the Riemann tensor.  From the EFT viewpoint, derivative interactions 
are higher operator order and should be
suppressed relative to non-derivative (algebraic)
interactions---unless for some reason only derivative interactions
appear (for example, to enforce a shift symmetry).

Given this, let us focus on non-derivative interactions.  The interaction density
must then be a sum of terms of the form
\begin{equation}
  F[R_{abcd},\vartheta] = \sum_{i} f_{i}(\vartheta) A_{i}[R_{abcd}]\,,
  \label{eq:non-deriv}
\end{equation}
where $f_{i}(\vartheta)$ are ordinary functions of $\vartheta$, and every component 
of $A_{i}[R_{abcd}]$ is a scalar function that is homogeneous of
degree 2, and now depends only on the
Riemann tensor but not its derivatives.
In the units we are using, the $f_{i}(\vartheta)$ have dimensions of
length squared.  If so desired, they can be made dimensionless by
pulling out some coefficients $\alpha_{i}$ with dimensions of
length squared, i.e.
\begin{equation}
\label{eq:alpha-def}
f_{i}(\vartheta) =\alpha_{i} \bar{f}_{i}(\vartheta)\,,
\end{equation}
with $\bar{f}_{i}$ dimensionless.

At first degree in curvature, there is only one scalar curvature invariant, the Ricci scalar $R$.
At quadratic degree there are only four independent scalar curvature
invariants,
\begin{align}
\label{eq:quad-scalar-invar}
  R^{2}, && R_{ab}R^{ab}, && K, && \pont \,,
\end{align}
where the Kretschmann scalar is $K\equiv R_{abcd}R^{abcd}$,
the Pontryagin density is $\pont\equiv {}^{*}\!R_{abcd}R^{abcd}$,
with the (left) dual of the Riemann tensor defined as
\begin{equation}
  \label{eq:dual-riem-def}
  {}^{*}\! R^{ab}{}_{cd} = \frac{1}{2}\epsilon^{abef}R_{efcd}\,,
\end{equation}
and where $\epsilon_{abcd}$ is the Levi-Civita tensor.
All other quadratic curvature invariants are algebraically dependent
on the four in \eqref{eq:quad-scalar-invar}.  For example, using
the Weyl tensor $C_{abcd}$, we have that ${}^{*}\!C_{abcd}C^{abcd}=
\pont$; similarly, an appropriate contraction of two copies of the
left-dual Riemann tensor $\sRsR$ is proportional to both the Euler density and
what is typically referred to as the four-dimensional Gauss-Bonnet
density $\theGB$,
\begin{align}
  \label{eq:4euler-dens}
  \sRsR &\equiv {}^{*}\!R_{abcd} {}^{*}\!R^{cdab} = -\theGB\,, \\
  \theGB &\equiv R^{2}-4R_{ab}R^{ab}+R_{abcd}R^{abcd}\,.
\end{align}

Given all of this, when we restrict ourselves to quadratic gravity theories with
non-derivative interactions, the most general form of $S_{q}$ is given
by
\begin{align}
\label{eq:q-action}
 S_{q} ={} \int d^4x& \sqrt{-g} \; \left[f_1(\vartheta) R^2
+ f_2(\vartheta) R_{ab} R^{ab}  \right. \\
& \left.  {}+ f_3(\vartheta) R_{abcd} R^{abcd}
+ f_4(\vartheta) R_{abcd} {}^{*}\!R^{abcd}   \right]\,.\nn
\end{align}
Some examples of these are presented in the bottom rectangular box of Fig.~\ref{fig:quad-theory-classif}. 

If the coupling between the scalar and curvature is large, the terms
in Eq.~\eqref{eq:q-action} can drastically affect the theory.  For
example, at strong coupling, the Pontryagin coupling can make the
graviton kinetic term flip sign at high $k$-number, becoming a ghost
field~\cite{Dyda:2012rj}.  However, in the EFT context, this only
occurs outside the regime of validity of the
theory~\cite{Delsate:2014hba}.  We discuss this further in
Sec.~\ref{sec:EFT}.

When considering derivative interactions without any further restrictions, 
a plethora of terms could be written for the interaction density.
Some examples are presented in the top rectangular box of Fig.~\ref{fig:quad-theory-classif}. 
Note, of course, that many other terms could also be 
written down. For example, derivative interactions could also include terms 
proportional to products of derivatives of the scalar field with the Ricci tensor
or scalar. We mention derivative interactions for pedagogical reasons
and for completeness of the classification, but we will not study them in 
this paper in detail.  

\subsubsection{Topological/non-topological density}

Let us define a \emph{topological interaction} as one that is
the product of a function of the scalar field and a \emph{topological density} $T_{i}$:
\begin{equation}
F[R_{abcd},\vartheta,\ldots] = \sum_{i} \bar{F}_i[\vartheta,\pd \vartheta,\ldots] \, T_i\,,
\end{equation}
where each functional $\bar{F}_{i}$ is now independent of the curvature tensor.
A topological density is defined as a quantity whose volume integral over the 
four-dimensional manifold is a topological invariant. 
For example, the Euler density, which is proportional to $\theGB$
(see Sec.~\ref{sec:Hair-Loss-Theorem} for the exact relationship), 
is a topological density because its volume integral is proportional
to the Euler characteristic $\chi$.
Similarly, the Pontryagin density $\pont$ is also a topological density
because its volume integral is proportional to the first Pontryagin
number.
 
An important property of topological densities is that in a
simply-connected neighborhood, each may be written as the divergence
of a 4-current. For example, in the Pontryagin case,
\begin{align}
  \pont &= \cd_{a} J^{a}\,, & \text{(locally)}
\end{align}
and similarly for $\theGB$. This allows for the local integration by parts
of such interactions, which explains why, in the absence of a scalar field, 
they do not lead to modifications to the classical field equations.

Thus, when we restrict ourselves to quadratic gravity theories with non-derivative, 
topological interactions, the most general form of $S_{q}$ is 
\begin{align}
  \label{eq:S-q-topol}
  S_{q} = \sum_{i} \int d^{4}x \; \sqrt{-g} \; f_{i}(\vartheta) \; T_{i}\,,
\end{align}
which we define as \emph{topological quadratic gravity} (though this is
unrelated to topological quantum field theory or to topological
massive gravity~\cite{Deser:1981wh,Deser:1982vy}). 
A few examples of topological quadratic gravity theories are presented in 
the intersection of the left and bottom rectangles of Fig.~\ref{fig:quad-theory-classif}.
For simplicity, let us consider the case in which the scalar field couples only to a single
topological density. Then, when $T=\pont$, $S_{q}$
defines \emph{dynamical Chern-Simons} (DCS) gravity, and
when $T=\theGB$ it defines \emph{dynamical Gauss-Bonnet} (DGB) gravity.
The specific choice
\begin{flalign}
  \label{eq:S-q-TEdGB}
\specialEq &&
S_{\TEDGB} &= \int d^{4}x\sqrt{-g} \ \alpha_{\TEDGB} \; e^{-\gamma\vartheta} \; \theGB
&&
\end{flalign}
recovers \emph{truncated Einstein-dilaton-Gauss-Bonnet}
(TEdGB) theory in 4 dimensions~\cite{Maeda:2009uy}, where $\vartheta$ plays the role of the
dilaton and $(\alpha_{\TEDGB},\gamma)$ are constant coupling strengths. Later in this paper, 
we will study TEdGB in more detail, which is why we have marked it with $\specialEq$.
The full action for TEdGB is given by the sum of the Einstein-Hilbert
term \eqref{eq:EH-action}, the kinetic term action for $\vartheta$
with vanishing potential \eqref{eq:S-vartheta}, and the interaction
term above, viz.~$S=S_{\EH}+S_{\vartheta}+S_{\TEDGB}$.

An important property of these theories is that any constant shift or offset,
$f(\vartheta)\to f(\vartheta) + c$, in Eq.~\eqref{eq:S-q-topol} does not affect the classical
equations of motion (EOMs), since it only contributes a constant multiple of
a topological number to the action.  Equivalently, since $T$ can
locally be written as a divergence, upon variation of the action, this shift only
contributes a boundary term.  Thus, without loss of generality, we may shift
$f(\vartheta)$ such that $f(0)=0$, so that if $f(\vartheta)$ is expanded as a
Taylor series, the expansion starts at first order in
$\vartheta$.

Expanding $f(\vartheta)$ in a Taylor series is appropriate in the
so-called \emph{decoupling limit} of a theory, where the modifications to GR
are sufficiently small. This can be enforced, for example, by requiring that $f(\vartheta)$ 
satisfy 
\begin{align}
  f(\vartheta) \; T \ll \kappa R\,,
\end{align}
where we recall that, for example, $T=\pont$ or $T=\theGB$.  In this
case, we say that the
scalar $\vartheta$ ``decouples and interacts weakly.''
In fact, if the theory is treated as an EFT, we expect the action to
contain terms at higher order in $\alpha$.  In this case, the
solution $\vartheta$ should also be expanded in a power series in
$\alpha$; sub-leading terms in this expansion could be corrected by
higher-$\alpha$ terms in the action, so they may not be trusted.
Thus, to be consistent, if one ignores $\mathcal{O}(\alpha^{2})$ terms
in the action, one must also ignore $\mathcal{O}(\alpha^{2})$ terms in
the solution for $\vartheta$.

So long as
$f'(0)\neq 0$, all choices of coupling functions yield the
same two theories, \emph{decoupled dynamical Gauss-Bonnet} (\DDGBText{}) and
\emph{decoupled dynamical Chern-Simons} (\DDCSText), with actions\footnote{%
These theories have sometimes been referred to as ``Einstein-dilaton-Gauss-Bonnet'' gravity 
and ``dynamical Chern-Simons'' gravity elsewhere in the literature~\cite{CSreview,yunespretorius,Yunes:2011we}; 
we have here changed the terminology to distinguish between other similar theories.}
\begin{flalign}
\label{eq:S-DDGB}	
\specialEq &&
  S_{\DDGB} &= \int d^{4}x\sqrt{-g} \ \alpha_{\GB} \, \vartheta \; \theGB\,,
&&
\\
\label{eq:S-DDCS}
&&
S_{\DDCS} &= - \frac{1}{4} \int d^{4}x\sqrt{-g} \ \alpha_{\CS}  \, \vartheta \, \pont\,,
&&
\end{flalign}
where each $\alpha_{X}$ is constant (and the factor of $-1/4$ in $S_{\DDCS}$ is conventional).
\DDGBText{} is another theory that we will study in more detail later in this paper,
which is why we have marked it with $\specialEq$.

When considering quadratic gravity theories with non-topological interactions, many other
terms may be written down. Some examples are provided in the right rectangle of Fig.~\ref{fig:quad-theory-classif}. 
Those examples consist of an interaction that is the product of a function of the scalar field 
and the Kretschmann scalar $K$, which is not a topological invariant. For future convenience,
let us define the theory with $S_{q}$ given by 
\begin{flalign}
\specialEq &&
S_{\K} &= \int d^{4}x \sqrt{-g} \; \alpha_{K} \; \vartheta \; K \,
&&
\end{flalign}
as \emph{Kretschmann gravity}, where again $\alpha_{K}$ is a constant.
This is another theory we will investigate in some detail later on, which
is why we have marked it with $\specialEq$.
Of course, the function of the scalar field could depend on its derivatives. 
This paper will not focus on such theories any further, but we include them in 
Fig.~\ref{fig:quad-theory-classif} for completeness of the classification.  

\subsubsection{Shift symmetry}
Let us define \emph{shift-symmetric} theories as those whose equations
of motion are invariant
under the (discrete or continuous) shift $\vartheta\to\vartheta+c$, for a constant $c$. 
Quadratic gravity theories with non-derivative, topological interactions that depend on a linear 
coupling function, i.e.~$f(\vartheta) = \alpha \, \vartheta$, and contain a flat potential, i.e.~$U''(\vartheta)=0$, 
are shift symmetric (\DDGBText{} and \DDCSText{} are both special cases of such theories).
We can see this by locally integrating Eq.~\eqref{eq:S-q-topol} by parts, 
\begin{align}
  \label{eq:S-q-topol-int-by-parts}
  S_{q} = - \int d^{4}x \sqrt{-g} f'(\vartheta) (\cd_{a}\vartheta) J^{a}
  + \textrm{bndry.}\,,
\end{align}
and noting that for the case of a linear coupling function,
$f'(\vartheta)=\alpha$ is independent of $\vartheta$.  Thus 
a global constant shift $\vartheta\to\vartheta + c$ only changes
$S_{q}$ by a boundary term, so the EOMs are invariant.
A shift symmetry is a natural outcome of the decoupling limit of either DCS or
DGB.  It may also be a desirable property built into a theory, since
it protects the scalar potential from acquiring a mass via quantum
corrections, since all corrections must abide by the shift symmetry.

Shift symmetry is then another feature that may be used to classify theories,
which we denote as a central square in Fig.~\ref{fig:quad-theory-classif}.  
\DDGBText{} and \DDCSText{} are not the only quadratic gravity theories
that enjoy a shift symmetry.  Another example is actions that
only involve derivatives of $\vartheta$, such as
$\mathcal{L}\supset \alpha(\square \vartheta)K$.  Alternatively, a
theory may exhibit a discrete shift symmetry if it is periodic in
$\vartheta$, e.g.~$\mathcal{L}\supset \alpha\sin(\omega \;\! \vartheta)K$, 
for some constant $\omega$. 

Let us comment here that TEdGB, while not shift-symmetric, is
invariant under the simultaneous field redefinition and parameter
scaling $\vartheta \to \vartheta + c$,
$\alpha_{\TEDGB}\to \alpha_{\TEDGB} e^{+\gamma c}$, for some additive
constant $c$.  Therefore we can only discuss bounds on
$\alpha_{\TEDGB}$ if we have some way of specifying the constant $c$.
We will fix this freedom by identifying the asymptotic value
$\vartheta_{\infty}$.

\subsection{Well-posedness and EFT}
\label{sec:EFT}
As presented and classified here, quadratic gravity theories may not be
well-posed when treated as exact theories.  The EOMs may have
higher than second-order derivatives, and they may suffer from the
Ostrogradski instability~\cite{Woodard:2006nt}.  Indeed,~\cite{Delsate:2014hba}
analyzed \DDCSText{} as an exact theory and found a problematic
initial value formulation. However, not all quadratic gravity theories
contain higher than second-order derivatives: DGB is the special case
with only second-order EOMs.

Quadratic gravity theories should in fact be treated as \emph{effective theories}
with a limited regime of validity, rather than exact theories.  
This requirement comes from the expectation that GR fails as
a description at some short length scale, and we may model corrections
to GR in an extended regime of validity through an EFT
approach~\cite{Burgess:2003jk}. Within the regime of validity of
the EFT, the corrections to GR must be controllably small, so we find
ourselves in the decoupling limit [e.g.~Eqs.~\eqref{eq:S-DDGB} and
\eqref{eq:S-DDCS}].  In fact, it was shown in~\cite{Delsate:2014hba}
that in the decoupling limit, \DDCSText \, is well-posed around appropriate
background solutions.  This same argument should hold for other
higher-curvature theories when treated through order-reduction
 in the decoupling limit.

A common criticism here is that from dimensional analysis, the
dimensional coupling coefficients\footnote{Note that one could make
the replacement $\alpha_{X} \to \ell^{2} \bar{\alpha}_{X}$, such that the new
coupling strength is dimensionless, and all units are carried by the length scale $\ell$. 
We will not make that choice here.} $\alpha_{X}$ are expected to be
Planck scale and thus irrelevant at astrophysical length scales.
However, naive dimensional analysis seems to fail in certain sectors
(most obviously in the scaling of the cosmological constant), so we
will remain agnostic here.  Instead, we simply parametrize our ignorance
of the length scale at which GR requires corrections, and allow
observations to guide theory-building.

The validity of the EFT description requires that $S_{q} \ll S_{\EH}$
and that $S_{\vartheta} \ll S_{\EH}$.  In the geometric units used in
this paper ($G=1=c$), the conditions for validity of the EFT in
\DDGBText{} and \DDCSText{} become $\sqrt{\alpha_{X}} \ll
(\kappa/{\cal{C}}^{3})^{1/4} {\cal{R}}$ for $\alpha_{X} =
\alpha_{\GB}$ or $\alpha_{\CS}$, where ${\cal{R}}$ is the radius of
the smallest object in the system, and $\mathcal{C}\equiv
GM/\mathcal{R}$ is its gravitational compactness (recall that in
geometric units, the action has dimensions of $[S]=L^{2}$,
thus $[\vartheta]=L^{0}$ is dimensionless while the coupling constant
$[\alpha]=L^{2}$ is dimensional).  The preceding scaling estimates
made use of the estimate $\vartheta = {\cal{O}}[\alpha_{X}
({\cal{C}}/{\cal{R}})^{2}]$, obtained from the EOM of
the scalar field.
Clearly, as the dimensional coupling strength goes to zero, one
recovers GR, while for sufficiently small couplings,
quadratic gravity is a small deformation of Einstein's theory.

We use the fact that the theory is a small deformation of GR to
establish a perturbative scheme for finding solutions in the
decoupling limit~\cite{Campanelli:1994sj, Woodard:2006nt,
  Cooney:2009rr}.  For a more extensive discussion of this
approach, see~\cite{Yagi:2011xp}.  The dynamical fields---the metric,
scalar field, and any other fields present---are expanded in a Taylor
series.  Explicitly, we have
\begin{align}
  \vartheta &= \vartheta^{[0]} + \zeta^{1/2}\vartheta^{[1/2]} + \mathcal{O}(\zeta^{1})\,,\\
  g_{ab} &= g_{ab}^{[0]}  + \zeta^{1/2}g_{ab}^{[1/2]} + \zeta^{1} g_{ab}^{[1]} + \mathcal{O}(\zeta^{3/2})\,,
\end{align}
where $\zeta$ is a dimensionless parameter which is proportional to
$\alpha_{i}^{2}$.  When $\zeta\to 0$, the EOM for $\vartheta$ becomes
simply $\square^{[0]}\vartheta^{[0]}=0$ (since we are interested in
long-ranged scalar fields, we have a vanishing potential).  In order
to satisfy asymptotic flatness, the asymptotic solution for
$\vartheta^{[0]}$ must be $\lim_{r\to\infty} \vartheta^{[0]} \to \text{const}.$
For the special case of a shift-symmetric theory, this latter constant
can be set to zero.  Then by examining the perturbed equations of
motion we will find
\begin{align}
  \label{eq:scalar-decomp}
  \vartheta &= 0 + \zeta^{1/2}\vartheta^{[1/2]} + \mathcal{O}(\zeta^{3/2})\,,\\
  \label{eq:metric-decomp}
  g_{ab} &= g_{ab}^{[0]} + \zeta^{1} g_{ab}^{[1]} + \mathcal{O}(\zeta^{2})\,.
\end{align}
The powers of $\zeta$ appearing above follow from setting
$\vartheta^{[0]}=0$, and the presence of the
explicit $\alpha_{i}$ in the interaction term of the Lagrangian.  The
background that we expand about is a GR solution,
$g_{ab}^{[0]}=g_{ab}^{\GR}, \vartheta^{[0]}=0$.  The most important
feature of this order-reduction scheme is that order-by-order, the
principle part of the differential operator acting on each
$\vartheta^{[k]}$ and $g^{[k]}$ is respectively
$\square^{[0]} \vartheta^{[k]}$ (the background d'Alembertian) and
$G^{[1]}[g^{[k]}]$ (the linearized Einstein tensor operator).  Because
of this, the order-reduced EOMs are always well-posed.

\subsection{Constraints on Quadratic Gravity Theories}

Not many quadratic gravity theories have been studied in sufficient
detail to be tested against observations of, for example, solar system
phenomena or binary pulsars. Nevertheless, one might think that
such theories have already been constrained strongly from binary pulsar 
observations because 
a theory that contains long-ranged scalar fields will predict the
excitation of dipole radiation in binary systems. Such radiation would
produce a much faster decay of the orbital period of binary pulsars,
and since this has not been observed, such theories must already be
well-constrained or ruled out.

Certain quadratic gravity theories that are shift-symmetric, however, evade
this problem completely, as we will show in this paper, and thus, such
theories are much less-well-constrained. The most well-studied
shift-symmetric theories in the context of experimental relativity are
those within the non-derivative and topological interaction class.  Recall that
these reduce to \DDCSText{} and \DDGBText{} in the decoupling limit.
The best estimated constraint on \DDGBText{} comes from observations of low-mass
X-ray binaries (LMXBs), which imply that $\sqrt{|\alpha_{\GB}|} \lesssim
{\cal{O}}(2 \; {\rm{km}})$~\cite{Yagi:2012gp}.\footnote{This bound also
applies to Kretschmann gravity. A similar bound is obtained on
TEdGB gravity via the existence of stellar-mass BHs~\cite{Kanti:1995vq,Pani:2009wy}.}
The best constraint on \DDCSText{} comes from observations of
frame-dragging with Gravity Probe~B and the LAGEOS satellites~\cite{alihaimoud-chen}
and from table-top experiments~\cite{kent-CSBH}, which
imply that $\sqrt{|\alpha_{\CS}|} \lesssim {\cal{O}}(10^{8} \;
{\rm{km}})$.

The \DDGBText{} estimate is much stronger than the \DDCSText{}
constraint for two reasons.  First, the \DDCSText{} correction is only
sourced by the parity-odd part of the background.  Thus spherically
symmetric configurations, like the exterior gravitational field of a
non-rotating star, are not modified in \DDCSText{} at all, because the
Pontryagin density vanishes.  Secondly, there is no estimated
constraint on \DDCSText{} from a stellar-mass BH system, only from the
solar system.  The curvature in the solar system, however, is
extremely weak relative to that of BHs in LMXBs, and thus the
estimated constraint on \DDGBText{} is much stronger.

\section{Miracle Hair Loss Theorem}
\label{sec:Hair-Loss-Theorem}

In this section, we present a proof that asymptotically $1/r$,
spherically symmetric scalar hair 
(which recall, we refer to in this paper as the \emph{scalar charge})
cannot be supported by objects with
no event horizon, like NSs and ordinary stars, in the
decoupling limit of Gauss-Bonnet gravity.
This proof both generalizes and makes more
rigorous the ``physicist's proof'' presented in~\cite{Yagi:2011xp}, 
but it follows the same spirit.
We first state the theorem and then give a sketch of the proof, which
we hope is convincing to most readers. We then present a complete proof for
those readers desiring more mathematical rigor.

\begin{theorem}
Consider a 4-dimensional manifold $M$ which is homeomorphic to
Minkowski (thus excluding black hole spacetimes, which are
`punctured'). Let $M$ be endowed with a metric $g$ with Lorentz
signature, which is stationary and asymptotically flat~\cite{Wald:1984cw,Stewart:1990}.
We require that the Riemann curvature tensor is continuous
\emph{almost everywhere},\footnote{%
  In the sense of measure theory, a property holds \emph{almost
    everywhere} if the set of points where it fails to hold has
  measure zero.
  Curvature tensors may be discontinuous for a body with a solid
  surface where the density goes to zero discontinuously; e.g.~in GR,
  the Ricci tensor is nonzero inside the body but it vanishes outside,
  with a discontinuity at the surface.
}
with any discontinuities in a spatially compact set of
measure zero. We also require that, in an asymptotically Cartesian coordinate
system, the components of the Riemann tensor decay at least as
$\mathcal{O}(r^{-2})$.\footnote{%
  To establish the Riemann-integrability of the Gauss-Bonnet density,
  we need the discontinuities to have measure-zero, and for Riemann to
  have sufficiently fast asymptotic fall-off.
}
Further, consider a real scalar field $\vartheta$,
stationary under the same isometry as the metric, whose
dynamics are governed by a linear coupling in the action to the
Gauss-Bonnet density, thus satisfying an EOM
\begin{equation}
\label{eq:vartheta-prop-GB}
  \square\vartheta = c_{1} \left( {}^{*}R_{abcd}{}^{*}R^{cdab} \right)
\end{equation}
for some constant real number $c_{1}$.
Then the asymptotically $1/r$, spherically-symmetric scalar
hair (the \emph{scalar charge}) vanishes.
\end{theorem}

\begin{proof}[Sketch of the proof]
  Our proof begins by integrating the EOMs
  [Eq.~\eqref{eq:vartheta-prop-GB}] over a suitably chosen spacetime region
  $C$,
   \begin{equation}
    \int_{C} \square\vartheta \, \sqrt{-g} \; d^{4}x = c_{1}\int_{C} \sRsR \, \sqrt{-g} \; d^{4}x\,.
    \label{eq:int-eom}
  \end{equation}
  As depicted in Fig.~\ref{fig:proof-geom}, the region $C$ is a
  spatial 3-ball crossed with a segment of time, $t\in[0,T]$. For
  technical reasons, we compactify the time direction $t\sim t+T$,
  which for a stationary situation does not change the
  physics. 
   
  We now manipulate the right-hand side of Eq.~\eqref{eq:int-eom}.
  First, we use the generalized Gauss-Bonnet-Chern
  (GBC) theorem for (pseudo-)Riemannian
  manifolds~\cite{avez1962geometrie,MR0155261} with
  boundary~\cite{Alty:1994xj,Gilkey:2014wca}. This converts the integral 
  into a sum of (i) a topological number that vanishes for our
  topology, and (ii) a boundary integral. Thus we have
  \begin{equation}
    \int_{C} \square\vartheta \sqrt{-g} \; d^{4}x = -c_{2} \oint_{\pd C}
    \mathbf{\Theta}(\mathbf{n})\,,
  \end{equation}
  where $\mathbf{\Theta}(\mathbf{n})$ is a 3-form which depends on the
  outward spatial unit normal vector $\mathbf{n}$, and $c_{2}$ is
  another constant multiple of $c_{1}$. The integral of the
  right-hand side at large radius $r$ exists and decays at least as $r^{-1}$.

\begin{figure}[tb]
  \centering
  \includegraphics[width=5cm]{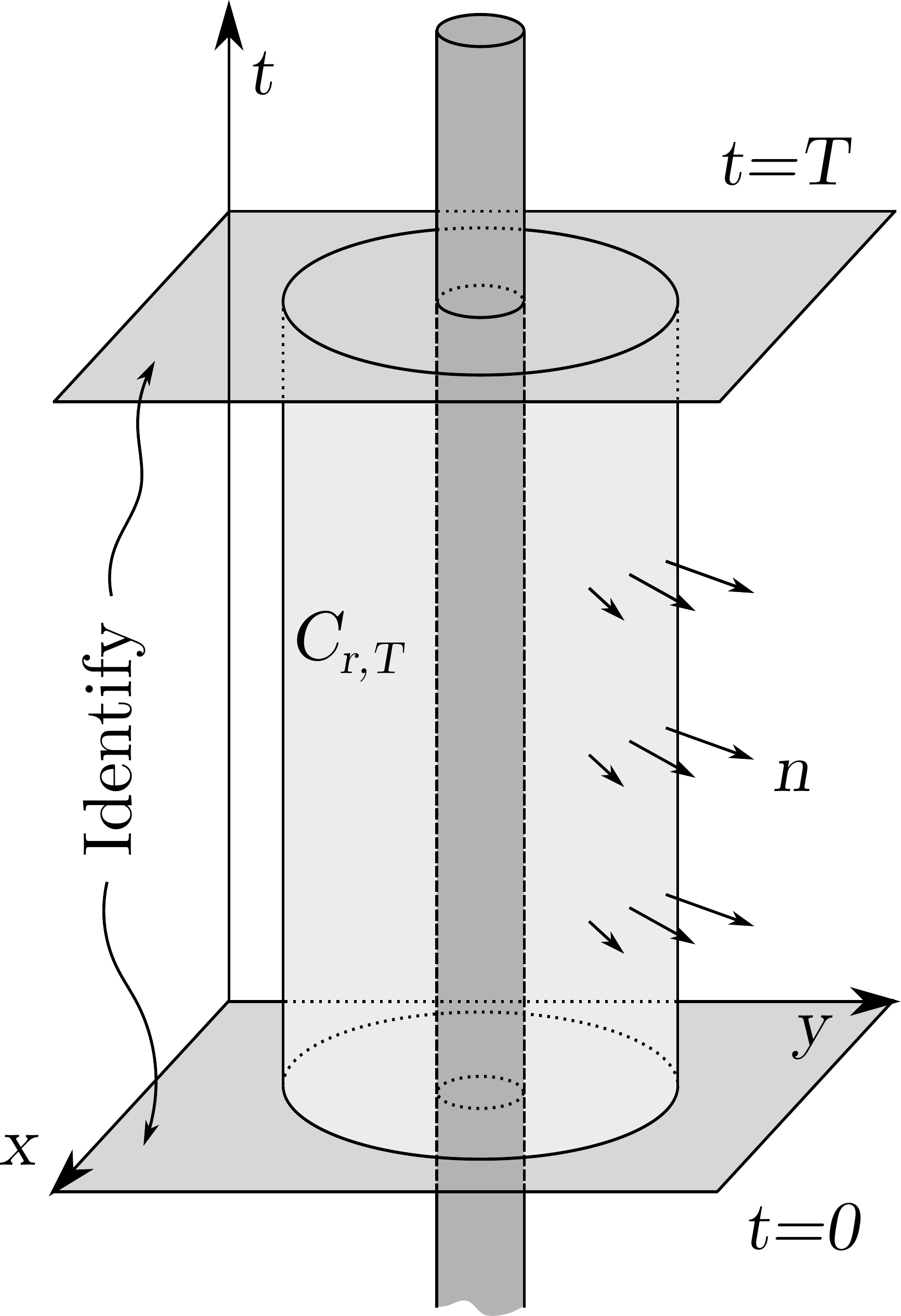}
  \caption{
    Spacetime geometry for our proof. Time $t$ is up, two spatial
    dimensions are presented as $x,y$ and one spatial dimension is
    suppressed. The dark shaded cylinder represents the world tube of
    the compact matter sources, e.g.~a stationary isolated neutron
    star. The integration region is the lighter shaded cylinder
    $C_{r,T}$, a spatial 3-ball of radius $r$ over the time interval
    $t\in[0,T]$. We compactify the time axis by the identification
    $t\sim t+T$, giving our time axis the topology of a circle
    $S^{1}$. The region $C_{r,T}$ has unit outward normal vector
    $\mathbf{n}$.
  }
  \label{fig:proof-geom}
\end{figure}

  Next we manipulate the left-hand side of Eq.~\eqref{eq:int-eom}.
  First, we use the generalized Stokes theorem to
  convert the integral to another boundary integral, and consider the limit as
  $r\to\infty$:
   \begin{equation}
    \lim_{r\to\infty} T
    \oint_{\pd C} (\pd_{r} \vartheta) r^{2} d^{2}\Omega = 0\,,
  \end{equation}
  where $d^{2}\Omega = \sin\theta d\theta d\phi$ is the area element
  on the unit 2-sphere. In the limit as $r\to\infty$, this integral
  exists and depends on only the scalar charge $\mu$, which recall
  we define as
  \begin{equation}
  \vartheta = \frac{\mu}{r} + {\cal{O}}\left(r^{-2}\right)\,.
  \label{eq:scalar-charge-def}
  \end{equation}
  Thus, we find
  \begin{equation*}
    4\pi T \mu = 0\,,
  \end{equation*}
  which implies that $\mu = 0$. \qedhere
\end{proof}

The sketch presented above summarizes the proof that 
scalar charge cannot be supported by objects without
event horizons in a quadratic gravity theory in which 
the scalar field satisfies Eq.~\eqref{eq:vartheta-prop-GB}.
In particular, this is the case in the decoupling limit of 
dynamical Gauss-Bonnet gravity, i.e.~in \DDGBText{}
where $c_{1} = \alpha_{\GB}$. Those readers who are satisfied
with this level of detail may proceed to the next Section. 
For those desiring mathematical rigor, we now present the 
complete proof.

\begin{proof}
Our proof makes use of a special case of the
generalized Gauss-Bonnet-Chern theorem for manifolds with boundary and
indefinite metric, presented by Alty~\cite{Alty:1994xj} and later by
Gilkey and Park~\cite{Gilkey:2014wca}. Let us briefly re-state this theorem
for the special case of a 4-dimensional manifold $\mathcal{M}$ with metric 
$g$ of signature $(-{+}{+}+)$ or $(+{-}{-}-)$. Let $\mathcal{M}$ 
have a (potentially empty) boundary $\pd\mathcal{M}$ with an induced metric
on each of its connected components whose signature never changes sign, i.e.~each
component has a normal $\mathbf{n}$ which is either everywhere
spacelike, everywhere timelike, or everywhere null.
Then, the 4-dimensional generalized Gauss-Bonnet-Chern
theorem says that\footnote{%
  The careful reader will note that Alty's theorem was more general,
  allowing for vector fields other than $\mathbf{n}$. The price for
  this generalization is to also include the topological \emph{kink}
  number~\cite{Alty:1994xj}, which vanishes when only considering
  $\mathbf{n}$; hence we omit it. Gilkey and
  Park~\cite{Gilkey:2014wca} make this same simplification, also
  omitting the kink number.
}
\begin{equation}
  \label{eq:GBC4}
  \chi(\mathcal{M}) =
  (-1)^{[p/2]}
  \left[
    \int_{\mathcal{M}} \mathbf{\Delta}
    +\int_{\pd\mathcal{M}} \mathbf{\Theta}(\mathbf{n})
  \right]\,,
\end{equation}
where $\chi(\mathcal{M})$ is the Euler characteristic of
$\mathcal{M}$ and $[p/2]$ is the largest integer $\le p/2$. The 4-form
$\boldsymbol{\Delta}$ is given by~\cite{Alty:1994xj}
\begin{equation}
  \label{eq:Delta}
  \boldsymbol{\Delta}
  =\frac{1}{128\pi^{2}}\epsilon_{abcd}\epsilon^{a'b'c'd'}
  R^{ab}{}_{a'b'} R^{cd}{}_{c'd'}
  \sqrt{-g} d^{4}x \,,
\end{equation}
and the 3-form $\mathbf{\Theta}(\mathbf{n})$ is given
by~\cite{Alty:1994xj}
\begin{align}
  \label{eq:Theta}
  \mathbf{\Theta}(\mathbf{n}) &=
  \left[\Theta_{1}(\mathbf{n})-\Theta_{2}(\mathbf{n})\right] \ 
  {}^{(3)}\boldsymbol{\epsilon}(\mathbf{n})\,, \\
  \Theta_{1}(\mathbf{n}) &= \frac{1}{16\pi^{2}}
   {}^{(3)}\epsilon_{abc} {}^{(3)}\epsilon^{a'b'c'} R^{ab}{}_{a'b'} n^{c}{}_{;c'}\,, \\
  \Theta_{2}(\mathbf{n}) &= \frac{1}{12\pi^{2}}
   {}^{(3)}\epsilon_{abc} {}^{(3)}\epsilon^{a'b'c'} n^{a}{}_{;a'} n^{b}{}_{;b'} n^{c}{}_{;c'} \,,
\end{align}
with ${}^{(3)}\epsilon_{abc}=n^{d}\epsilon_{dabc}$ the induced volume
3-form in the tangent subspace orthogonal to $\mathbf{n}$.

To apply the 4-dimensional generalized Gauss-Bonnet-Chern
theorem, we will take the manifold $\mathcal{M}\equiv
C'_{r,T}$ to be a submanifold $C'_{r,T}\subset M$ of the whole
spacetime. We start with a submanifold $C_{r,T}$ which is a spatial
3-ball $B_{r}$ of radius $r$ crossed with a time interval
$t\in[0,T]$. This submanifold does not satisfy the conditions of
Alty's proof because (i) the boundary is not smooth, having
``corners'' at the ends of the 4-cylinder (see
Fig.~\ref{fig:proof-geom}), and (ii) the boundary $\pd C_{r,T}$ has a
normal which is timelike in some regions (the top/bottom of the
4-cylinder) and spacelike in others (the sides of the
4-cylinder). However, because the spacetime is stationary, the physics
is not affected by compactifying the time direction. Thus,  we use the
identification $t\sim t+T$, which turns the time axis from
$\mathbb{R}$ into $S^{1}$. This glues the top and bottom of the
4-cylinder together, giving it the topology of $B_{r}\times S^{1}$. We
call this glued manifold $C'_{r,T}$. The boundary $\pd C'_{r,T}$ has
topology $S^{2}\times S^{1}$, and the normal is everywhere spacelike,
satisfying the conditions to apply Alty's proof.

We now proceed by integrating the EOMs
[Eq.~\eqref{eq:vartheta-prop-GB}] over the region $C'_{r,T}$,
\begin{equation}
  \label{eq:EOM-integrated}
  \int_{C'_{r,T}}\square\vartheta \sqrt{-g} d^{4}x =
  c_{1} \int_{C'_{r,T}} \left( {}^{*}R_{abcd}{}^{*}R^{cdab} \right)
  \sqrt{-g} d^{4}x \,.
\end{equation}
First we investigate the right-hand side, and apply
Eq.~\eqref{eq:GBC4}, which gives
\begin{equation}
\label{eq:EOM-int-apply-GB}
\begin{aligned}
  \int_{C'_{r,T}}\square\vartheta \sqrt{-g} d^{4}x =
  c_{2} \Bigg[ &
    \chi(C'_{r,T}) -\int_{\pd C'_{r,T}} \mathbf{\Theta}(\mathbf{n})
  \Bigg]\,,
\end{aligned}
\end{equation}
in the case of $p=1$ (the final result is the same for $p=3$), where
$c_{2}=32\pi^{2} c_{1}$.
The Euler characteristic of a product manifold
satisfies $\chi(M\times N)=\chi(M)\cdot\chi(N)$, and $\chi(S^{1})=0$,
thus $\chi(C'_{r,T})=0$.  Thus, we have
\begin{equation}
\label{eq:EOM-int-apply-GB-chi-kink}
  \int_{C'_{r,T}}\square\vartheta \sqrt{-g} d^{4}x =
  -c_{2} \int_{\pd C'_{r,T}} \mathbf{\Theta}(\mathbf{n})\,.
\end{equation}

We now must prove that in the limit $r\to\infty$, each side of
Eq.~\eqref{eq:EOM-int-apply-GB-chi-kink} exists (i.e.~both integrals
converge).
We start with the asymptotic behavior of the integrand
$\mathbf{\Theta}(\mathbf{n})$. In a stationary, asymptotically flat
spacetime~\cite{Wald:1984cw,Stewart:1990}, in asymptotically Cartesian
coordinates $(t,x,y,z)$, the metric has asymptotic fall-off
\begin{equation}
  \label{eq:g-asymp}
  g_{ab} = \eta_{ab} + \mathcal{O}(r^{-1})\,,
\end{equation}
with $r$ defined in the ordinary Cartesian fashion,
$r^{2}=x^{2}+y^{2}+z^{2}$. It is this function that defines the
region $C'_{r,T}$, and thus, $n_{a}=\cd_{a}r$.  In these same coordinates,
the behavior of $n^{a}{}_{;b}$ is
\begin{equation}
  n^{a}{}_{;b} = r^{-1}(\delta^a_b - n^a n_b) + \mathcal{O}(r^{-2})\,.
\end{equation}
By assumption, we also have the asymptotic fall-off for the components
of the Riemann tensor, in an asymptotically Cartesian coordinate
system, $R_{abcd}\sim \mathcal{O}(r^{-2})$ [true for all index
positions because of Eq.~\eqref{eq:g-asymp}]. In fact,
Eq.~\eqref{eq:g-asymp} implies $R_{abcd}\sim \mathcal{O}(r^{-3})$, but
only the weaker condition $\mathcal{O}(r^{-2})$ is required for our
proof.

Now we can see the leading asymptotic behavior of the integrands on
the right-hand side of Eq.~\eqref{eq:EOM-int-apply-GB-chi-kink}:
\begin{align}
  \label{eq:theta-asympt}
  \Theta_{1}(\mathbf{n}) &\sim \mathcal{O}(r^{-3})\,,\\
  \Theta_{2}(\mathbf{n}) &\sim \mathcal{O}(r^{-3})\,.
\end{align}
Actually, $\Theta_{1}(\mathbf{n})$ decays as $\mathcal{O}(r^{-4})$
following the fall-off of Riemann determined by
Eq.~\eqref{eq:g-asymp}, but again we only need the weaker decay. When
integrated over $\pd C'_{r,T}$, we find that the integral
exists and converges at least as
\begin{align}
  \label{eq:theta-convergence}
  \int_{\pd C'_{r,T}} \mathbf{\Theta}(\mathbf{n}) &\sim \mathcal{O}(r^{-1})\,,
\end{align}
and in the limit as $r \to \infty$, 
\begin{align}
  \label{eq:theta-convergence-lim}
  \lim_{r\to\infty}
  \int_{\pd C'_{r,T}} \mathbf{\Theta}(\mathbf{n}) &= 0\,.
\end{align}

We now turn to the left-hand side of
Eq.~\eqref{eq:EOM-int-apply-GB-chi-kink}, where we can apply the
generalized Stokes theorem to turn the volume integral into a boundary
integral:
\begin{equation}
  \label{eq:stokes-LHS}
  \int_{C'_{r,T}} \square \vartheta \sqrt{-g} d^{4}x =
  \int_{\pd C'_{r,T}} d\Sigma_{a} \, \cd^{a}\vartheta \,,
\end{equation}
where $d\Sigma_{a}$ is the area element on the boundary $\pd
C'_{r,T}$. In asymptotically spherical coordinates, this is given by
\begin{equation}
\label{eq:bndry-vol-el}
d\Sigma_{r} = r^{2} d^{2}\Omega dt\left[1 +
  \mathcal{O}(r^{-1})\right]\,,
\end{equation}
where the standard unit 2-sphere area element is
$d^{2}\Omega=\sin\theta{}d\theta{}d\phi$, and other components of
$d\Sigma_{a}$ are subdominant and vanish in the limit $r\to\infty$.
To show convergence we must study the behavior of asymptotic solutions
to Eq.~\eqref{eq:vartheta-prop-GB}.

The solution for $\vartheta$ will be a combination of homogeneous and
particular solutions,
$\vartheta = \vartheta_{\text{hom}} + \vartheta_{\text{part}}$,
subject to the condition of asymptotic flatness. Let us first consider
the homogeneous solution. In the limit $r\to\infty$,
Eq.~\eqref{eq:vartheta-prop-GB} reduces to the flat-space Laplacian
(from stationarity and asymptotic flatness); thus we know that
\begin{equation}
  \label{eq:vartheta-hom}
  \vartheta_{\text{hom}} \sim \sum_{lm}
  Y_{lm}(\theta,\phi) \left[\frac{a_{lm}}{r^{l+1}}+b_{lm}r^{l}\right]\,,
\end{equation}
for coefficients $a_{lm}$, $b_{lm}$. To satisfy asymptotic flatness as
$r\to\infty$, we must have $b_{lm}=0$ except for $b_{00}$.  The
coefficient $b_{00}$ is determined by boundary conditions (or, in the
case of a shift-symmetric theory, it can be set to any value).  The
integral in Eq.~\eqref{eq:stokes-LHS} is insensitive to $b_{00}$ since
only the derivative $\cd_{a}\vartheta$ enters the integrand.

We now consider the particular solution.
The non-compact source term, $\sRsR$, decays at least as
$\mathcal{O}(r^{-4})$ [in fact as $\mathcal{O}(r^{-6})$ following
Eq.~\eqref{eq:g-asymp}, but again we only need the weaker decay].
Therefore, the slowest-decaying contribution from the particular
solution is at worst
\begin{equation}
  \label{eq:vartheta-part}
  \vartheta_{\text{part}} \sim r^{-2} \quad\text{or}\quad r^{-2}\log r\,.
\end{equation}
This log term arises if there is an $r^{-4}$ component in the $l=1$
term of the spherical harmonic decomposition of the source term,
$\sRsR$.

Now we can show that the left-hand side integral of
Eq.~\eqref{eq:EOM-int-apply-GB-chi-kink} converges.  From asymptotic
flatness and in asymptotically spherical coordinates we have
\begin{equation}
  \label{eq:surf-integ-coord}
  \lim_{r\to \infty}
  \int_{\pd C'_{r,T}} \!\!\!  d\Sigma^{a} \; \cd_{a}\vartheta
  =
  \lim_{r\to \infty}
  \int_{\pd C'_{r,T}} \!\!\! dt \; d^{2}\Omega \; r^{2} \; \pd_{r}\vartheta \,.
\end{equation}
With the far-field asymptotic behavior of $\vartheta_{\text{hom}}$ and $\vartheta_{\text{part}}$ given in
Eqs.~\eqref{eq:vartheta-hom} and \eqref{eq:vartheta-part}, the only
part of $\pd_{r}\vartheta$ that contributes is
\begin{equation}
  \pd_{r}\vartheta =\pd_{r}(\vartheta_{\text{hom}} +
  \vartheta_{\text{part}}) =
  -\frac{a_{00}}{r^{2}} \left[ 1 + \mathcal{O}(r^{-1}) \right]\,.
\end{equation}
We conventionally call $a_{00}=\mu$ the scalar charge [compare for
example with Eq.~\eqref{eq:scalar-charge-def}]. Thus, the
left-hand side converges to
\begin{equation}
  \label{eq:surf-integ-result}
    \lim_{r\to \infty}
  \int_{\pd C'_{r,T}} d\Sigma^{a} \, \cd_{a}\vartheta
  = -4\pi T \mu\,.
\end{equation}

We have now shown that the integrals on both the left- and the
right-hand sides converge in the limit $r\to\infty$.  Inserting the
limits [Eq.~\eqref{eq:theta-convergence-lim} and
Eq.~\eqref{eq:surf-integ-result}] into the volume-integrated equation
of motion [Eq.~\eqref{eq:EOM-int-apply-GB-chi-kink}, after applying the GBC
theorem] yields
\begin{align}
  4\pi T \mu &{}= 0\,.
\end{align}
Thus we have proved that the asymptotically $1/r$,
spherically-symmetric scalar charge $\mu$ must vanish.
\end{proof}

\section{Neutron Star Scalar Charge in Quadratic Gravity}
\label{sec:analytics}

In this section, we derive the NS scalar charge 
in a few quadratic gravity theories. In particular, we focus on theories with
non-derivative interactions but allow for both topological and non-topological
interaction densities (the bottom rectangle in Fig.~\ref{fig:quad-theory-classif}). 
Such theories are defined through the interaction density of Eq.~\eqref{eq:non-deriv},
which with a certain choice for $f_{i}(\vartheta)$, leads to the quadratic action
\begin{align}
S_{q} ={}& \int d^{4}x\sqrt{-g} \  g(\vartheta) \left[\alpha_{1} R^{2} + \alpha_{2} R_{ab} R^{ab}
\right. 
\nn \\
& \left.
\qquad {}+ \alpha_{3} R_{abcd} R^{abcd} + \alpha_{4} \, {}^{*}R_{abcd} R^{abcd} \right]\,,
\label{eq:quad-action-we-focus-on}
\end{align}
where $(\alpha_{1},\alpha_{2},\alpha_{3},\alpha_{4})$ are all constants. 
This particular choice of $f_{i}(\vartheta)$ allows us to recover
the examples discussed in Sec.~\ref{sec:ABC-MQG} through
appropriate choices of $g(\vartheta)$, as shown in Table~\ref{tab:theory-summary}.
Notice also that \DDGBText{} gravity can be recovered from TEdGB gravity by taking
the limit $\gamma \to 0$ while $\gamma \; \alpha_{\TEDGB} \to - \alpha_{\GB} = {\rm{const}}$.

\begin{table}[tb]
\begin{centering}
\begin{tabular}{r|c|c|c|c|c}
\hline
\hline
\noalign{\smallskip}
& $\alpha_{1}$ & $\alpha_{2}$ & $\alpha_{3}$ & $\alpha_{4}$ & $g(\vartheta)$ \\
\hline
Kretsch.~gravity & 0 & 0 & $\alpha_{\K}$ & 0 & $\vartheta$ \\
TEdGB gravity & $\alpha_{\TEDGB}$ & $-4 \, \alpha_{\TEDGB}$ & $\alpha_{\TEDGB}$ & 0 & $e^{-\gamma \vartheta}$  \\
\DDGBText{} gravity & $\alpha_{\GB}$ & $-4 \, \alpha_{\GB}$ & $\alpha_{\GB}$ & 0 & $\vartheta$  \\
\noalign{\smallskip}
\hline
\hline
\end{tabular}
\end{centering}
\caption{%
Parameter mapping for non-derivative, quadratic gravity theories that
we investigate in detail.
}
\label{tab:theory-summary}
\end{table}

The NS scalar charge is obtained by solving the EOM for the scalar field.
The latter follows from Eq.~\eqref{eq:evol-scal-eq}, which with a
vanishing potential ($U=0$) and the non-derivative, quadratic gravity
action of Eq.~\eqref{eq:quad-action-we-focus-on} reduces to $\square \vartheta = S$,
where we have defined the source function
\begin{align} 
S \equiv{}& - \left(\frac{\pd g}{\pd \vartheta}\right) \left(
 \alpha_{1} R^{2} + \alpha_{2} R_{ab} R^{ab} 
\right. 
\nn \\
& \left.
{}+ \alpha_{3} R_{abcd} R^{abcd} + \alpha_{4} \, {}^{*}R_{abcd} R^{abcd}
 \right)\,.
 \label{eq:source-function}
\end{align}
Once we solve the EOM and extract the scalar charge, we will use it in
the next Section to determine the best systems to constrain such
theories, and estimate new constraints when possible.


We first concentrate on deriving the scalar charge of a non-rotating NS,
and then extend the analysis to a rotating configuration. 
In each case, we will first calculate the scalar charge analytically 
within a weak-field approximation scheme and for certain simple equations of state.
We will then confirm our results numerically in the strong-field regime and
for more complicated equations of state. We will explicitly demonstrate the 
vanishing of the scalar charge in \DDGBText{} gravity, which was proven formally
in the previous section, and also show that the scalar charge does not vanish in
TEdGB gravity or in Kretschmann gravity.

For the purposes of comparison, we note here that the scalar charge
has also been computed for BHs in several quadratic gravity
theories~\cite{1992PhLB..285..199C, Mignemi:1992nt, Kanti:1995vq,
  Yunes:2011we, Yagi:2011xp, Sotiriou:2014pfa, Kleihaus:2014lba}.  In
the decoupling limit of any non-derivative quadratic gravity theory,
for a BH with a mass $M_\BH$ and at leading order in spin, the
dimensionless (mass-reduced) scalar charge is given by
\begin{equation}
\mu_{\BH}^{(0)} = 2 g'(0) \frac{\alpha_3}{M_{\rm BH}^2}\,.
\label{eq:BH-scalar-charge}
\end{equation}
This is non-vanishing for \DDGBText{}, to be compared with the
vanishing of scalar charge for NSs.

\subsection{Non-rotating Neutron Stars}
\label{sec:non-rot}

Let us first consider a non-rotating stellar configuration, described by a perfect
fluid matter source that generates a spherically symmetric spacetime.
Given this, the scalar field can only be a function of the radial coordinate, namely
$\vartheta =  \vartheta^{(0)}(r)$. The superscript (0) reminds us that we can think of
$\vartheta^{(0)}$ as the zeroth-order term in a small-spin expansion. The GR
metric is then simply
\begin{align}
ds^{2}_{(0)} ={}& - e^{\nu} dt^{2} + \left(1 - \frac{2 M}{r}\right)^{-1} dr^{2}
\nn \\
&{}+ r^{2} \left(d\theta^{2} + \sin^{2}{\theta} d\varphi^{2}\right)\,,
\end{align}
where $M = M(r)$ is the spherically-symmetric enclosed mass
and $\nu = \nu(r)$ is a metric function that satisfies the Einstein
equations.
The scalar field evolution equation is then
\begin{equation}
\label{eq:scalar-field-eq-int}
\frac{d^2 \vartheta^{(0)}}{d r^2} = - \frac{2 [2 \pi (p-\rho) r^3  +r- M]}{r (r-2 M)} \frac{d \vartheta^{(0)}}{d r}  + S^{(0)}\,,
\end{equation}
where $p(r)$ and $\rho(r)$ are the internal pressure and energy density, and 
$S^{(0)}$ is the function evaluated on $g_{ab}^{(0)}$ [see Eq.~\eqref{eq:source}].
The calculation of the scalar charge requires that we first solve 
Eq.~\eqref{eq:scalar-field-eq-int} inside the star and then match it to 
an exterior solution at the stellar surface to determine any constants of integrations.

Analytic solutions to Eq.~\eqref{eq:scalar-field-eq-int} for the scalar field
do not generically exist in closed-form, but they can be obtained using 
certain approximations, such as a \emph{post-Minkowskian} or \emph{weak-field} 
expansion. In a post-Minkowskian expansion, one expands and solves the equations
in powers of the compactness $C = M_{*}/R_{*} \ll 1$, where $M_{*}$ and $R_{*}$ are the
NS mass and radius.\footnote{%
  Neutron stars are objects with small compactness, typically of
  ${\cal{O}}(10^{-1})$, so a post-Minkowskian expansion is well-justified.}
But even with such an approximation, Eq.~\eqref{eq:scalar-field-eq-int} can
still only be solved for certain particular equations of 
state~\cite{1939PhRv...55..413V,Tolman:1939jz,Tsui:2005hb}. We
focus here on an $n=0$ polytropic equation of
state~\cite{1939PhRv...55..413V} ($p = K \rho^{1 + 1/n}$, where 
$K$ and $n$ are constants with the latter representing the polytropic index)
and a Tolman VII equation of
state~\cite{Tolman:1939jz,Tsui:2005hb} with  
$\rho \propto 1-r^2/R_*^2$.  These 
represent a constant density star and an approximation to more
realistic equations of state respectively. 

Given the above, we compute the scalar charge as follows. First, we substitute
the analytic, GR solutions\footnote{The effect of non-GR corrections 
to the metric on the scalar charge can be neglected to the order we work on, as
explained in App.~\ref{app:nonrot-NS}.} 
to the equations of structure for the metric
tensor and pressure at zeroth-order in rotation into
Eq.~\eqref{eq:scalar-field-eq-int}. We next substitute Eq.~\eqref{eq:scalar-decomp} 
in Eq.~\eqref{eq:scalar-field-eq-int} and expand the equation order by order in 
$\zeta$. The exterior solutions can be obtained by setting $p=0=\rho$, $M=M_*$
and solving the decomposed equations order by order, as given by 
Eqs.~\eqref{eq:0th-Rot-ext-sol} and~\eqref{eq:0th-Rot-ext-sol-2}.
Meanwhile, the interior solutions can be obtained as a further
expansion in powers of $C$. We match these exterior and interior solutions 
order by order in $C$ and $\zeta$ at the stellar
surface using the condition given by Eq.~\eqref{eq:matching}.

Once a solution to Eq.~\eqref{eq:scalar-field-eq-int} has been obtained, we can then read off the 
$1/r$ piece of the external solution and calculate 
\begin{equation}
\label{eq:scalar-asympt-0}
\vartheta^{(0)}_{\ext}(r) = \vartheta_\infty +  \mu^{(0)} \frac{M_*}{r} + \mathcal{O} \left( \frac{M_*^2}{r^2} \right)
\end{equation}
far from the source ($r \gg M_*$). Here, $\mu^{(0)}$ is the dimensionless scalar charge at
zeroth-order in spin and $\vartheta_{\infty}$ is a constant that the scalar field asymptotes 
to at spatial infinity. Recall again that this constant can be set to zero in theories that are
shift symmetric.

Let us now present the scalar charge in TEdGB, \DDGBText{} and Kretschmann gravity.
In order to reveal whether and how the scalar charge vanishes, let us consider the theory
defined by the interaction density of Eq.~\eqref{eq:quad-action-we-focus-on} with 
\begin{equation}
\label{eq:g-example}
g(\vartheta) = e^{-\gamma \vartheta}\,,
\end{equation}
and $\alpha_{4}=0$. We recover 
the theories mentioned above by taking the following limits
\begin{itemize}
\item \emph{TEdGB limit}: $(\alpha_{1},\alpha_{2},\alpha_{3}) \to \alpha_{\TEDGB} (1,-4,1)$.
\item \emph{Kretschmann limit}: $\gamma \to 0$, while $\gamma (\alpha_{1},\alpha_{2},\alpha_{3} )\to - \alpha_{\K} (0,0,1)$ for a constant $\alpha_{\K}$.
\item \emph{\DDGBText{} limit}: $\gamma \to 0$, while $\gamma (\alpha_{1},\alpha_{2},\alpha_{3}) \to - \alpha_{\GB} (1,-4,1)$ for a constant $\alpha_{\GB}$. 
\end{itemize}
With this in mind, the scalar charge for a NS with an $n=0$ polytropic equation of state is
\begin{align}
\label{eq:mu-PM-n0}
\mu^{(0)}_{n=0} ={}& - 12 \; {\gamma} \; e^{-\gamma \vartheta_{\infty}} \; \frac{C}{R_*^2}  \left( \alpha_{\GB,1} - {\textstyle \frac{3}{5}} \alpha_{\GB,2}  C  - {\textstyle \frac{18}{35}} \alpha_{\GB,2}  C^2  \right) \nn \\
& {}-{\textstyle \frac{96}{35}} \; {\gamma}^3 \; e^{-2 \gamma \vartheta_{\infty}} \; \frac{C^{3}}{R_*^4} \left[ 32 \alpha_3^2 + 63 \alpha_{\GB,2} ( \alpha_{\GB,2}-\alpha_3)  \right] \nn \\
& {}+ \mathcal{O} \left(\zeta^{3/2}, C^4 \right)\,,
\end{align}
while for a Tolman VII equation of state we find
\begin{align}
\label{eq:mu0-Tol}
\mu^{(0)}_\mrm{Tol} ={}& - {\textstyle\frac{120}{7}} \; {\gamma} \; e^{-\gamma \vartheta_{\infty}} \; \frac{C}{R_*^2} \left[ \alpha_{\GB,2}
- {\textstyle\frac{1}{22}} \left( 19 \alpha_{\GB,1}+ \alpha_{\GB,2} \right) C \right. \nn \\
&  \quad\left. {} - {\textstyle\frac{1}{2145}} \left( 1583 \alpha_{\GB,1} + 177 \alpha_{\GB,2} \right) C^2 \right] \nn \\
& {}-{\textstyle\frac{320}{1001}} \; {\gamma}^3 \;e^{-2 \gamma \vartheta_{\infty}} \; \frac{C^{3}}{R_*^4} \left[ 1008 \alpha_3^2 + 5 \alpha_{\GB,2} \left( 300 \alpha_{\GB,2} \right. \right. \nn \\
& \quad\left. \left. {}- 373 \alpha_3\right)  \right] + \mathcal{O} \left(\zeta^{3/2},  C^4 \right)\,,
\end{align}
where we have here defined
\begin{equation}
\label{eq:alpha_GB}
\alpha_{\GB,1} \equiv 3 \alpha_1 + \alpha_2 + \alpha_3\,, \quad \alpha_{\GB,2} \equiv  \alpha_1 + \alpha_2 + 3 \alpha_3\,.
\end{equation}
Both combinations $\alpha_{\GB,1},\alpha_{\GB,2}$ vanish for the
Gauss-Bonnet ratio  $(\alpha_{1},\alpha_{2},\alpha_{3})\propto (1,-4,1)$.
Note that the terms proportional to $R_*^{-2}$ and $R_*^{-4}$ 
in Eqs.~\eqref{eq:mu-PM-n0} and~\eqref{eq:mu0-Tol} are of $\mathcal{O}(\zeta^{1/2})$ and 
$\mathcal{O}(\zeta)$ respectively.

Let us now take the aforementioned limits to investigate the scalar charge
in TEdGB, \DDGBText{} and Kretschmann gravity. In the Kretschmann limit,
Eqs.~\eqref{eq:mu-PM-n0} and~\eqref{eq:mu0-Tol} reduce to
\begin{align}
\mu^{(0),\K}_{n=0} &= 12 \alpha_\K  \frac{C}{R_*^2} \left( 1- \frac{9}{5} C  - \frac{54}{35} C^2  \right) + \mathcal{O} \left(\zeta^{3/2}, C^4 \right)\,, \\
\label{eq:mu-K}
\mu^{(0),\K}_\mrm{Tol} &= \frac{360}{7} \alpha_\K  \frac{C}{R_*^2} \left( 1 - \frac{1}{3} C - \frac{2114}{6435} C^2 \right) 
+ \mathcal{O} \left(\zeta^{3/2},  C^4 \right)\,.
\end{align}
Observe that the $\mathcal{O}(\zeta)$ contribution vanishes in
Kretschmann gravity because the scalar field is linearly coupled
to the Kretschmann density in the quadratic action.
In fact, the $\mathcal{O}(\zeta)$ part of the scalar charge vanishes generically
for any quadratic gravity where the scalar field is coupled \emph{linearly}
to a quadratic curvature scalar, i.e.~for any non-derivative quadratic gravity
theory to leading-order in the decoupling limit.

Next, let us investigate the TEdGB limit. Since in this limit $\alpha_{\GB,1,2} \to 0$,
the scalar charge becomes
\begin{align}
\label{eq:charge-EDGB-n0}
\mu^{(0),\TEDGB}_{n=0} &=  -\frac{3072}{35} \; {\gamma}^{3} \; e^{-2\gamma \vartheta_{\infty}} \; \alpha_{\TEDGB}^2 \; \frac{C^3}{R_*^4}  + \mathcal{O} \left(\zeta^{3/2},  C^4 \right)\,, \nn \\
\\
\label{eq:charge-EDGB-Tol}
\mu^{(0),\TEDGB}_{\mrm{Tol}} &=  -\frac{46080}{143} {\gamma}^{3} \; e^{-2\gamma \vartheta_{\infty}} \; \alpha_{\TEDGB}^2 \; \frac{C^{3}}{R_*^4} 
+ \mathcal{O} \left(\zeta^{3/2},  C^4 \right)\,. \nn \\
\end{align}
Notice that the $\mathcal{O}(\zeta^{1/2})$ contribution vanishes in this case and the leading
order contribution is of $\mathcal{O}(\zeta)$. 
Such a modification would be of \emph{second-order} in
the coupling constants, and thus, a consistent treatment would require
inclusion of terms of the same order in the non-minimal curvature
coupling at the level of the action---we ignore such terms here,
although in principle they could cancel the above result.
Notice also that the leading-order terms are proportional to $C^3$, whereas those in 
Eqs.~\eqref{eq:mu-PM-n0} and~\eqref{eq:mu0-Tol} are proportional to $C$. 
Therefore, the scalar charges in TEdGB gravity are quadratically suppressed by
the stellar compactness. We have checked that these analytic results match 
the purely numerical calculation of~\cite{pani-NS-EDGB}.

Finally, let us now consider the \DDGBText{} limit. We can do this by simply taking 
the ${\gamma} \to 0$ limit of Eqs.~\eqref{eq:charge-EDGB-n0}
and~\eqref{eq:charge-EDGB-Tol} while $\gamma\alpha_{\TEDGB}$ remains
finite.
Doing so, one finds that the scalar charge vanishes identically to $\mathcal{O}(\zeta)$.
From the $\mathcal{O}(\zeta^{1/2})$ contribution in 
Eqs.~\eqref{eq:mu-PM-n0} and~\eqref{eq:mu0-Tol}, one sees that
the Gauss-Bonnet combination is the only one 
that can make the scalar charge vanish.
Moreover, the charge depends on different combinations of
$\alpha_{\GB,1,2}$ at different orders in compactness, but in all
cases, the charge vanishes linearly with $\alpha_{\GB,1}$ and
$\alpha_{\GB,2}$ in the Gauss-Bonnet limit.

Given the result of Sec.~\ref{sec:Hair-Loss-Theorem}, we know that
this vanishing must hold to all orders in compactness.  We can further
support this expectation with another explicit analytic example,
without imposing a post-Minkowskian expansion. Determining this in closed-form is
difficult in general, but doable for strongly \emph{anisotropic} NSs with an
$n=0$ polytropic equation of state. Let us then consider anisotropic NSs 
following~\cite{1974ApJ...188..657B} as a \emph{toy model}, which allows
for solutions to the equations of stellar structure analytically without any approximations.
We consider the strongly anisotropic limit, in which the radial pressure vanishes,
so that calculations are analytically tractable.
Working in \DDGBText{} gravity, 
we solve the scalar field equation in the interior region analytically
without a post-Minkowskian expansion, then
match the solution to the exterior one at the surface and find that 
the scalar charge vanishes exactly (see App.~\ref{app:ani} for 
more detailed calculations).

\begin{figure}[tb]
\includegraphics[width=8.5cm,clip=true]{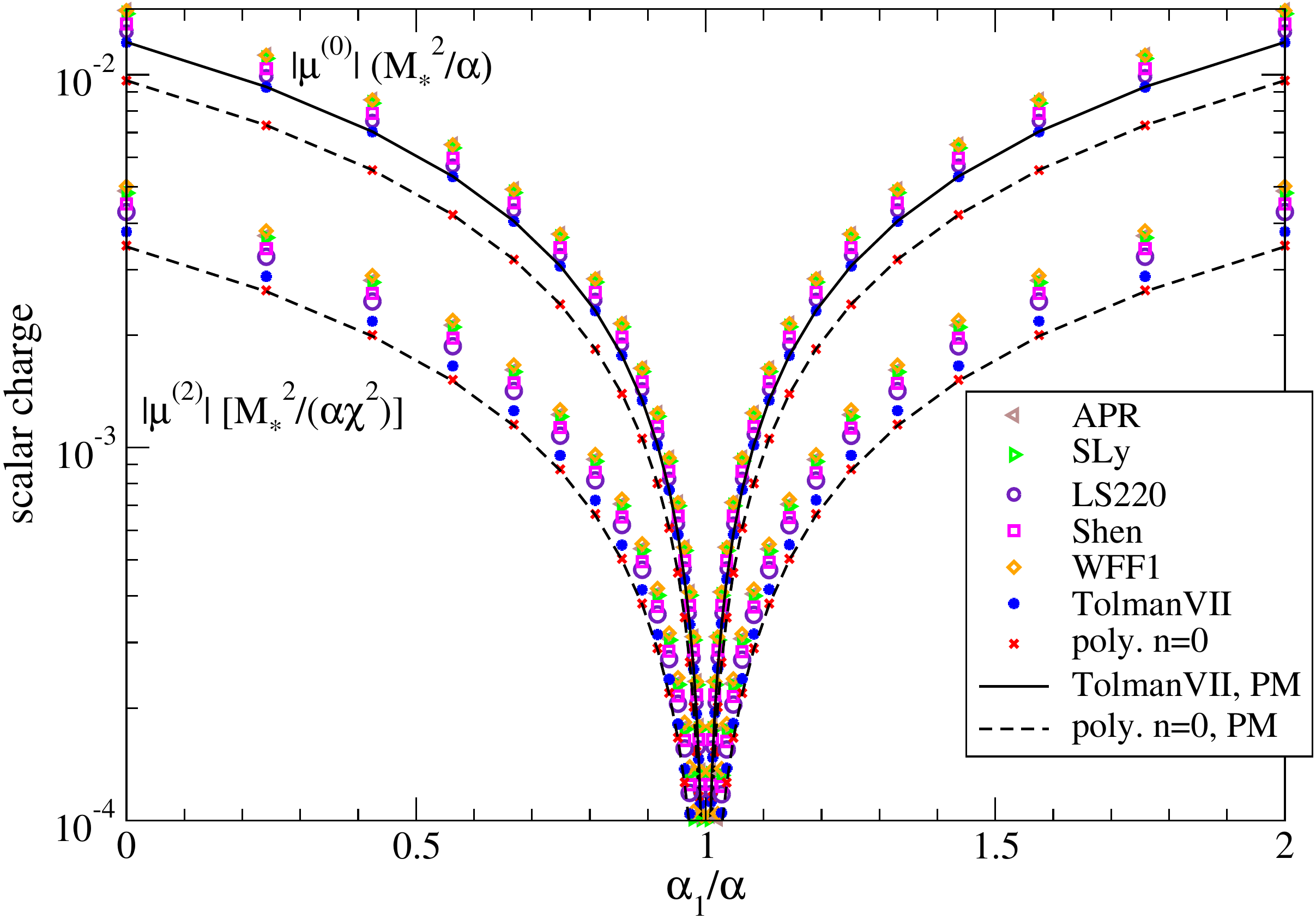}
\caption{\label{fig:scalar-charge-vanish} (Color online) Scalar
  monopole charge in non-derivative, quadratic gravity 
  in the decoupling limit for various equations of state at zeroth-
  ($\mu^{(0)}$) and second-order ($\mu^{(2)}$) in rotation as a
  function of $\alpha_1/\alpha$ with $\alpha$ an arbitrary
  constant. For this example, we set $C=0.1$, $\alpha_2 = -4 \alpha$
  and $\alpha_3 = \alpha$. For the TolmanVII stars and the $n=0$ polytropes, we
  further set $R_*=12$ km. Observe that the charges approach zero
  rapidly as one approaches the Gauss-Bonnet limit
  ($\alpha_1/\alpha \to 1$). Solid and dashed curves represent the
  analytic relation for the TolmanVII models and the $n=0$ polytropes
  within the post-Minkowskian approximation.}
\end{figure}

We can also show that the scalar charge vanishes in \DDGBText{} 
gravity for more general equations of state and isotropic matter, provided
we carry out a numerical analysis. Details of the numerical algorithm are explained
in App.~\ref{app:num}. As an example, let us consider the non-derivative, quadratic 
gravity model of Eq.~\eqref{eq:quad-action-we-focus-on} in the decoupling limit,
i.e.~with $g(\vartheta) = \vartheta$. Such a model contains both \DDGBText{} gravity 
and Kretschmann gravity, as one can see in Table~\ref{tab:theory-summary}. 
Figure~\ref{fig:scalar-charge-vanish} shows the scalar charge for
NSs with $\alpha_2 = -4 \alpha$, $\alpha_3 = \alpha$
and $C=0.1$. Observe how the charge vanishes as
the Gauss-Bonnet limit ($\alpha_{1}/\alpha \to 1$) is
approached for all equations of state considered. 
For comparison, we also include the scalar charge computed
for an $n=0$ polytrope and a Tolman VII equation of state numerically
(red crosses and blue circles) and in a post-Minkowskian expansion
(dashed and solid curves). Observe how good of an approximation the
post-Minkowskian expansion is relative to the numerical solutions.


\subsection{Slowly-rotating Neutron Stars}
\label{sec:slow-rot}

The theorem in Sec.~\ref{sec:Hair-Loss-Theorem} is not only more rigorous than that
presented in Ref.~\cite{Yagi:2011xp}, but it can also be applied to a rotating NS.
We here explicitly demonstrate that the scalar charge vanishes in \DDGBText{} 
gravity even for a slowly-rotating configuration both analytically and numerically. To do so, 
we consider a quadratic gravity theory with the quadratic action of Eq.~\eqref{eq:quad-action-we-focus-on}
but with a linear coupling function $g(\vartheta) = \vartheta$. This will allow us to investigate how the scalar 
charge vanishes in the \DDGBText{} limit. 

We work in a slow-rotation expansion of 
Hartle and Thorne~\cite{hartle1967,Hartle:1968ht}
to quadratic order in spin,
i.e.~a systematic expansion in $J/M_{*}^{2} \ll 1$,
where $J \equiv|\vec{J}|$ is the magnitude of the spin angular
momentum. Physically, we assume that $M_{*} \Omega \ll 1$, or equivalently $M_{*}/P \ll 1$,
where $\Omega$ and $P$ are the spin angular frequency and spin period
of the star respectively.  Such an assumption is well-justified for
all observed pulsars, even those with millisecond periods for which
$M_{*} \Omega = {\cal{O}}(10^{-2})$.

Because non-rotating stars are spherically symmetric,
deformations due to rotation can be modeled through a spherical
harmonic decomposition. The scalar field can then be decomposed
as
\begin{align}
\vartheta(r,\theta) \! &= \! \vartheta^{(0)}(r)  + \sum_{l = 0,2} \!\! \vartheta^{(2)}_{l} (r) P_l (\cos \theta)
+ {\cal{O}}[(M_{*} \Omega)^{4}]\,,
\end{align}
where $r$ and $\theta$ are radial\footnote{Technically, this radial coordinate has been transformed
from the standard radial coordinate of a non-rotating configuration, following the procedure laid out
by Hartle and Thorne~\cite{hartle1967,Hartle:1968ht}.} and polar coordinates
respectively, $P_{l}(\cdot)$ are Legendre polynomials, and
$\vartheta^{(2)}_{\ell} = {\cal{O}}[(M_{*} \Omega)^{2}]$.
As expected, there is no azimuthal
angle dependence, because rotating stars remain axisymmetric when in
slow rotation.

The scalar charge is the piece of the scalar field that decays as
$1/r$ at spatial infinity and is independent of $\theta$, and thus, 
we must solve for $\vartheta^{(0)}$ and $\vartheta^{(2)}_{0}$. 
The former leads to the scalar charge in spherical symmetry, 
which we already considered in the previous subsection, so we
here concentrate on the spin-dependent correction to the scalar charge
found in $\vartheta^{(2)}_{0}$.
One can define the dimensionless scalar charge $\mu^{(2)}$ 
at quadratic order in spin from the asymptotic behavior of $\vartheta^{(2)}_{0}$ 
at spatial infinity in the same way as in Eq.~\eqref{eq:scalar-asympt-0}:
\begin{equation}
\label{eq:scalar-asympt-2}
\vartheta^{(2)}_{0,\ext}(r) =  \mu^{(2)} \frac{M_*}{r} + \mathcal{O} \left( \frac{M_*^2}{r^2} \right)\,.
\end{equation}
Here, we set $\vartheta^{(2)}_{0,\ext}(\infty) = 0$ without loss of generality
by absorbing it into $\vartheta_\infty$.  As mentioned earlier, we
uniquely specify $\vartheta_{\infty}$ to be a constant
(which we will set to 0 in Sec.~\ref{sec:non-shift-sym-topo})
to all orders in rotation in
order to fix the freedom of simultaneous redefinition of
$\alpha_{\TEDGB}$ and $\vartheta$; fixing this freedom is required to
discuss limits on $\alpha_{\TEDGB}$.
The field equation for $\vartheta^{(2)}_{0}$ has the same form as 
Eq.~\eqref{eq:scalar-field-eq-int}, except that $S^{(0)}$ needs to be replaced by $S^{(2)}$,
whose explicit form with a linear coupling function is given in Eq.~\eqref{eq:source-2}.
As in the non-rotating case, we work in the small coupling approximation by 
decomposing $\vartheta^{(2)}_{l}$ in terms of $\zeta^{1/2}$ and solving the
decomposed field equation order by order. 
The exterior solution for the scalar field at second order in spin and 
to leading order in $\zeta^{1/2}$
is given by
Eq.~\eqref{eq:vartheta20-ext}.

We now derive $\mu^{(2)}$ analytically within the post-Minkowskian approximation.
As in the non-rotating case, we expand the scalar field equation
at second order in spin
about $C=0$ and solve it order by order in $C$ in the interior region. We then match this solution
to the exterior solution expanded in $C \ll 1$, using the conditions in 
Eqs.~\eqref{eq:matching201} and~\eqref{eq:matching202}.
With an $n=0$ polytropic equation of state, we find
\begin{align}
\label{eq:scalar-charge-quadratic-n0}
\mu^{(2)}_{n=0} ={}& 12 \Omega^2  \left[ \alpha_{\GB,1} - \frac{1}{20} \left( 12 \alpha_{\GB,1} + 61 \alpha_{\GB,2} \right) C \right. \nn \\
&  \left. {}+ \frac{3}{7} \left( \frac{13}{4} \alpha_{\GB,1} - \frac{9}{25} \alpha_{\GB,2} \right) C^2 \right]  + \mathcal{O} \left(\zeta^{3/2}, C^3 \right)\,.
\end{align}
Notice that $\mu^{(2)}_{n=0}$ vanishes to $\mathcal{O}(\zeta)$ in the \DDGBText{} limit, i.e.~when $\alpha_{\GB,1,2} \to 0$, 
in agreement with Sec.~\ref{sec:Hair-Loss-Theorem}.
We do not present $\mu^{(2)}$ with a TolmanVII equation of state 
because then we can only solve the equations of structure analytically at
zeroth order in rotation.

We next carry out a numerical calculation without imposing the 
post-Minkowskian approximation and for a variety of realistic equations of state.
We use the same numerical algorithm presented in App.~\ref{app:num}. The results
of this numerical investigation for non-derivative, quadratic gravity
in the decoupling limit are presented in Fig.~\ref{fig:scalar-charge-vanish}. 
Observe how $\mu^{(2)}$ approaches zero as one approaches the \DDGBText{} limit, just like 
$\mu^{(0)}$ does. Observe also that the numerical results for an $n=0$ polytrope 
agree very well with the analytic ones in Eq.~\eqref{eq:scalar-charge-quadratic-n0}.

\section{Current and Future Constraints}
\label{sec:bin-puls-cons}

In this section, we study what estimated and projected constraints can
be placed on some of the quadratic gravity theories discussed in
Sec.~\ref{sec:ABC-MQG}, using binary pulsar observations and
GW observations.
We begin in Sec.~\ref{sec:dipole-rad-bin-gw} with a discussion of how
a non-vanishing scalar charge leads to the emission
of dipolar radiation, which affects the orbital period decay of binary systems and the gravitational
waves they emit. 

Consider a binary system emitting dipolar radiation at the orbital
period.  If the system is observed via pulsar timing, the leading
deviation from GR is the correction to the post-Keplerian parameter
$\dot{P}_{b}$, the binary's period derivative.  This correction enters
at $-1$PN relative to the GR effect.
Meanwhile if the system is observed through GWs, the deformation from
the GR GW prediction can be captured via the $\beta_{\ppE}$
parameter of the parameterized post-Einstein (ppE)
framework~\cite{PPE}.  Thus the two observables we seek to compute in
this Section are the correction to the change in the binary period,
$\dot{P}_{b}$, and the ppE parameter $\beta_{\ppE}$.
These quantities can be used to project estimated constraints on the
theories of interest.  

We will consider three example theories.  In
Sec.~\ref{sec:DDGB-cons}, we focus on \DDGBText{} gravity as an
example of shift-symmetric, topological quadratic gravity.  We then
study two theories which do not satisfy the conditions of the miracle
hair loss theorem.  In Sec.~\ref{sec:non-shift-sym-topo} we focus on
non-shift-symmetric but topological theories, with TEdGB as an
example.  Finally in Sec.~\ref{sec:non-shift-sym-non-topo} we focus on
non-shift-symmetric and non-topological theories, with Kretschmann
gravity as an example.

For each theory, we will consider the two cases of binaries without
BHs, and binaries with at least one BH.  The latter is required for
\DDGBText{}, which falls under the purview of the theorem, since
otherwise dipolar radiation is strongly suppressed.
Since both TEdGB and Kretschmann predict that NSs source a
non-vanishing scalar charge, binary systems without BHs
suffice to stringently constrain these theories (though NS/BH binaries
still produce the most stringent projected constraint).
The best estimated and projected constraints for each of these theories 
are summarized in Table~\ref{tab:summary}. 
  
\subsection{Dipole Radiation, Binary Pulsars, and Gravitational Waves}
\label{sec:dipole-rad-bin-gw}

Consider a binary system composed of two compact objects (either BHs, 
NSs, or WDs) with masses $m_{1}$ and $m_{2}$ and
observed, for example, through radio pulsar timing or through future GW detectors.
A dynamical scalar field will induce a plethora of corrections to the dynamics of the binary, 
but, typically, the most important of these is due to the energy flux the field carries away as it evolves.
As calculated in e.g.~\cite{Will:1989sk,will-living,Yagi:2011xp}, this flux is
\begin{equation}
\label{eq:Edot-scalar}
\dot E^{(\vartheta)} = - \frac{4 \pi}{3}\eta^2 \left(\mu^{(0)}_{1} - \mu^{(0)}_{2} \right)^2 (v_{12})^8\,,
\end{equation}
where $\eta \equiv m_1 m_2 / m^2$ is the symmetric mass ratio, 
$m \equiv m_1 + m_2$ is the total mass, $\mu^{(0)}_{1,2}$ is the scalar charge 
of each compact object (we drop higher-spin corrections since NSs spin only
slowly as already mentioned in Sec.~\ref{sec:slow-rot}), and $v_{12}$
is the magnitude of the binary's relative orbital velocity.  Observe
that $\dot{E}^{(\vartheta)}$ is a $-1$PN order correction\footnote{%
  Henceforth, a term proportional to $(m/r_{12})^{A}$ relative to its
  leading-order expression will be said to be of $A$PN order,
  with $r_{12}$ the binary's separation. 
  By the virial theorem, $v_{12}^{2}$ is of the same order as
  $m/r_{12}$, and thus, a term proportional to $(v_{12})^{2N}$
  relative to some other term will be said to be of $N$PN
  order.} 
to the GW energy flux in GR $\dot E_\GR = -(32/5) \eta^2 (v_{12})^{10}$. 

For binary pulsars, the most important effect of the scalar energy flux  
is a modification to the rate of orbital period decay $\dot P_b$, which has already been
stringently constrained~\cite{Freire:2012mg,Wex:2014nva}. The orbital period
decay is also modified due to corrections to the binding energy $E_{b}$, but
these are subdominant in a PN sense. The rate of decay of the orbital
period $P_{b}$, at leading PN order, can be written as
\begin{equation}
\label{eq:Pbdot}
\dot{P}_{b} = \frac{3}{2} \frac{P_{b}}{E_{b}} \left( \dot{E}_{\GR} +
\dot{E}^{(\vartheta)} \right)\,.
\end{equation}
We consider a binary pulsar whose period derivative is observed to be
consistent with the prediction of GR, with an observational
uncertainty given by $\sigma_{\dot{P}_{b}}$.  From this observation,
one could infer that the fractional correction due to
$\dot{E}^{(\vartheta)}$ must be smaller than the fractional
uncertainty in the measurement and we can
estimate
\begin{equation}
\label{eq:Edot-Pdot}
\left| \frac{\dot E^{(\vartheta)}}{\dot E_\GR} \right| \lesssim
\left| \frac{\sigma_{\dot{P}_{b}}}{\dot P_{b}} \right| \,.
\end{equation}
Combining  Eqs.~\eqref{eq:Edot-scalar} and~\eqref{eq:Edot-Pdot}, one
then finds that a binary pulsar observation consistent with GR places 
a constraint on the scalar charges of roughly 
\begin{equation}
\label{eq:bound}
\left|\mu_{1}^{(0)} - \mu_{2}^{(0)}\right| \lesssim
\left| \frac{24}{5 \pi} \frac{\sigma_{\dot{P}_{b}}}{\dot P_{b}}\right|^{1/2}
|v_{12}|\,.
\end{equation}
Such a bound can be converted into a constraint on $\sqrt{|\alpha_X|}$ once we substitute
the explicit forms of the scalar charge on the left-hand side of Eq.~\eqref{eq:bound}. 
Observe that if the binary components have comparable compactnesses (and thus comparable scalar charges), as is the case for NS/NS
binary pulsars, then the constraint is weakened. Thus, when a scalar charge is present and a scalar energy
flux is sourced, the best systems to constrain such modifications are mixed binaries. 

The most important effect of the scalar energy flux in the GWs emitted
by binary systems is a modification in the Fourier phase of the waves. This correction
can be well-described in the parameterized post-Einsteinian framework of~\cite{PPE}. 
Performing a stationary phase approximation analysis,
the GW phase correction to GR in the Fourier domain is of $-1$PN relative order,
with magnitude 
\begin{equation}
\beta_{\ppE} = -\frac{5 \pi}{1792} \left(\mu_{1}^{(0)} - \mu_{2}^{(0)} \right)^{2} \eta^{2/5}\,.
\end{equation}
Consider a future GW measurement that is consistent with GR, meaning that
$\beta_{\ppE}$ is consistent with zero with an uncertainty given by
$\sigma_{\beta_{\ppE}}$.  From such a measurement, we would estimate
the projected constraint
\begin{equation}
\label{eq:bound-GW}
\left|\mu_{1}^{(0)} - \mu_{2}^{(0)}\right| \lesssim \left(\frac{1792}{5 \pi} \sigma_{\beta_{\ppE}} \right)^{1/2} \eta^{-1/5}\,.
\end{equation}
This bound can be converted into a projected constraint on $\sqrt{|\alpha_X|}$, just as for binary pulsars, 
by substituting in the explicit forms of the scalar charge on the left-hand side. 

When dipole radiation is present in a given theory, binary pulsar observations are better than GW observations
to constrain that theory~\cite{cornish-PPE}. This is because dipole radiation enters at a pre-Newtonian order 
relative to GW predictions in GR, which enter through quadrupole radiation at ``Newtonian'' order. This means 
that dipole radiation affects observables at ${\cal{O}}(c^{2}/v_{12}^{2})$ relative to the GR expectation. 
The orbital period of binary pulsars is much larger than that of GW sources of ground-based interferometers, 
and thus, $v_{12}$ is much smaller, which makes the effect of dipole radiation much larger for binary pulsars.  

\subsection{Shift-Symmetric, Non-Derivative, Topological Quadratic Gravity: \DDGBText{} Example}
\label{sec:DDGB-cons}

Let us take \DDGBText{} as an example of quadratic gravity theories with a shift symmetric, non-derivative,
topological interaction density. As discussed in Sec.~\ref{sec:Hair-Loss-Theorem}, NSs
will not source scalar charge in this theory, but BHs will. Let us then separate the discussion
of constraints into those that come from binaries where at least one of the components is a BH
and those where neither component is a BH. 

\subsubsection{Binaries without black holes}

Since scalar charge is not sourced in this case, the main modifications to the evolution
of the binary is due to higher-order multipole scalar hair and metric deformations induced by
modifications to the multipole moments. The latter dominates the former as sketched 
in App.~\ref{app:shift-sym-TQG-corr} and
arise through corrections to the moment of inertia and to the quadrupole moment 
of each individual star. These modifications 
induce corrections in the rate of decay of the orbital period and to the GW signal, 
but they are of 1.5 PN order or higher relative to the leading-order term in GR. 
Therefore, current binary pulsar observations of $\dot{P}_{b}$ cannot place meaningful 
constraints on \DDGBText{} gravity.  

One could, in principle, use other binary pulsar observables to constrain
the theory, such as the rate of change of the pericenter $\left<\dot{w}\right>$. This quantity
would be corrected at 0.5PN order relative to GR due to modifications to 
the star's moment of inertia and at 1PN order due to quadrupole moment deformations.  
The moment of inertia, however, is very hard to measure~\cite{lattimer-schutz}, 
and expected errors will be too large to allow for meaningful 
constraints~\cite{pani-NS-EDGB}. Moreover, 1PN corrections to 
perihelion precession from quadrupole moment deformations would lead to constraints
that are outside the regime of validity of the decoupling limit, as
is also the case in dynamical Chern-Simons gravity~\cite{Yagi:2013mbt}. 

One could in principle constrain these higher PN effects with GW
observations produced in the inspiral of NS binaries. We cannot construct
an estimate of these here because the precise form of such corrections is not yet
known. 

\subsubsection{Binaries with at least one black hole}
\label{sec:DDGB-BHbinaries}

Black holes source a scalar charge, so the dynamics of mixed BH/NS systems
is strongly corrected. Therefore, the best constraints on quadratic gravity with shift-symmetric, non-derivative,
and topological interactions will come from future observations of binary systems 
where at least one of the binary components is a black 
hole.\footnote{%
  In dynamical Chern-Simons gravity, BHs do not possess a
  scalar charge, but it is still true that the best constraints come
  from BH binaries.  This is because the leading correction to
  the binary evolution in this theory enters at 2PN order and is
  proportional to spin
  squared~\cite{Yagi:2011xp,kent-CSGW,Yagi:2013mbt}.  Since black
  holes spin much faster than NSs in general, BH
  binaries will place constraints that are stronger than NS
  binaries~\cite{kent-CSGW,Yagi:2013mbt}.}
This could be achieved, for example,
through future observations of yet unobserved BH/NS
pulsar binaries with radio telescopes~\cite{Liu:2011ae,Liu:2014uka} or
through future observations of the GWs emitted by
BH/BH or BH/NS
binaries~\cite{Yagi:2011xp,kent-LMXB}.

We can estimate the magnitude of the constraints one could achieve on 
\DDGBText{} gravity, given a future BH/NS pulsar observation with 
the \emph{Five-Hundred-Meter Aperture Spherical Radio Telescope} (FAST)~\cite{Nan:2011um} 
or with the \emph{Square-Kilometer Array} (SKA)~\cite{Carilli:2004nx}. 
Let us then consider a binary with masses of $m_\BH = 10M_\odot$ and $m_\NS = 1.4M_\odot$
in a circular orbit. Substituting the scalar charge of a non-rotating BH, given by
Eq.~\eqref{eq:BH-scalar-charge}~\cite{1992PhLB..285..199C,Mignemi:1992nt,
Yunes:2011we,Yagi:2011xp,Sotiriou:2014pfa} to leading order in the coupling constant,
and Kepler's law $v_{12} = (2\pi m/P_b)^{1/3}$ in Eq.~\eqref{eq:bound}, one obtains
the projected bound 
\begin{align}
\label{eq:bound-EdGB-decouple}
\left(\sqrt{|\alpha_\GB|}\right)^{\rm{Bin.Pul.}}_{\rm{BHNS}} \lesssim{} & 0.12 \, \mrm{km} \left( {\textstyle \frac{m_\BH}{10M_\odot}} \right) \left( {\textstyle \frac{m}{11.4M_\odot}} \right)^{1/6} \nn \\
& \times \left( {\textstyle \frac{\sigma_{\dot{P}_b}/\dot{P}_b}{10^{-2}}} \right)^{1/4}
\left( {\textstyle \frac{P_b}{3 \, \mrm{days}}} \right)^{-1/6}\,.
\end{align}

The contours of Fig.~\ref{fig:constraints-BH-NS} show estimates of the upper bound on 
$\sqrt{|\alpha_\GB|}$ [km] from Eq.~\eqref{eq:bound-EdGB-decouple}, as
a function of $|\sigma_{\dot P_b}/\dot P_b|$ and $P_b$.
Reference~\cite{Liu:2014uka} estimated the
accuracy to which FAST and SKA, as well as a 100-meter radio
dish for reference, would be able to measure the orbital period decay
as a function of the orbital period~\cite{Liu:2014uka}.
Figure~\ref{fig:constraints-BH-NS} shows these estimates with solid,
dashed or dotted curves, assuming 
a 3 or a 5 year
observation. 
Given
an observation of the orbital decay rate consistent with GR 
and with period $P_{b,{\rm obs}}$, one obtains a point in this
figure that must lie on one of the observation curves (either the 100
meter, FAST or SKA curves), which would then place an estimated constraint 
on  $\sqrt{|\alpha_\GB|}$ shown by the contours.

\begin{figure}[tb]
\begin{center}
\includegraphics[width=8.5cm,clip=true]{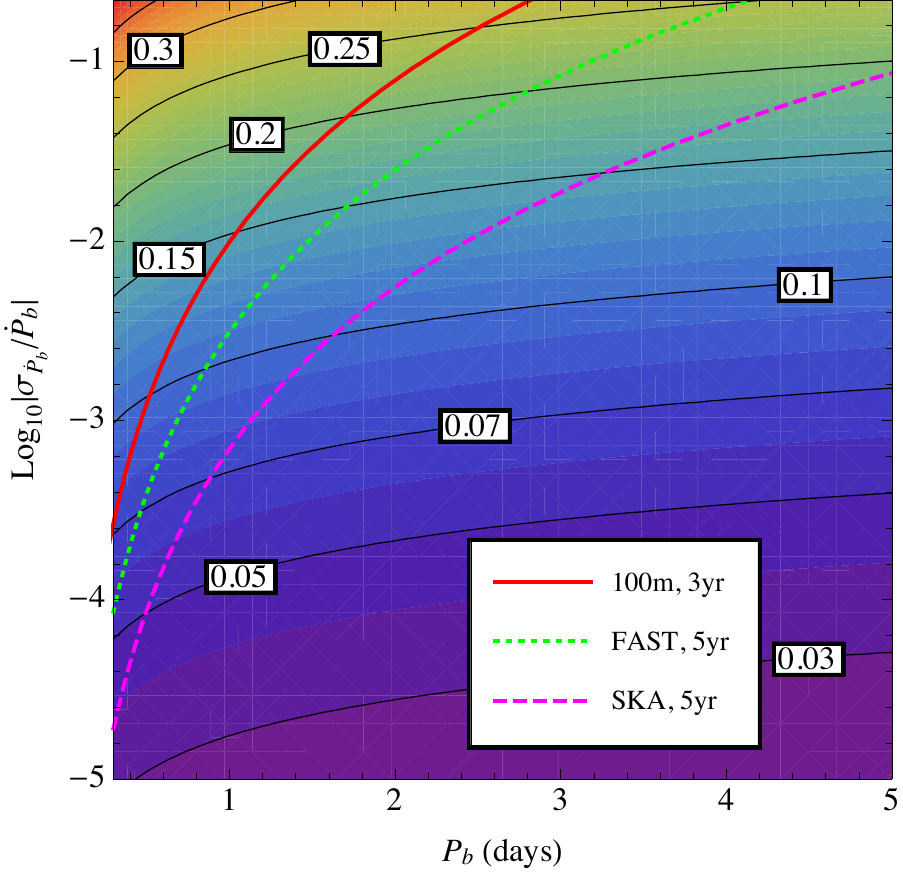}
\caption{\label{fig:constraints-BH-NS} (Color online) 
Upper bound on the theory coupling parameter
$|\alpha_{\GB}|^{1/2}$ in \DDGBText{}
gravity in units of kilometers 
as a function of orbital period and the
measurement accuracies of the orbital period decay rate of a BH/NS
binary pulsar
[Eq.~\eqref{eq:bound-EdGB-decouple}]. We also show
the projected accuracy of the
measurement of the orbital period decay rate with a 100 meter
antenna (red solid curve), FAST (green dotted curve) and SKA (magenta
dashed curve), as a function of orbital period. 
Given observations of the orbital period decay rate with this 
predicted accuracy, one would be able to constrain \DDGBText{} 
gravity one order of magnitude more 
stringently than the current estimated bound.}
\end{center}
\end{figure}

Let us now estimate the magnitude of the constraints one could achieve on 
\DDGBText{} gravity with a future GW observation with aLIGO
of the late inspiral of a compact binary. Cornish et al.~\cite{cornish-PPE} found that given an aLIGO
GW observation with signal-to-noise ratio (SNR) of 20, one could assert that $\beta_{\ppE}$ is
zero up to an uncertainty of roughly $\sigma_{\beta_{\ppE}} = 5 \times 10^{-4}$. Using this projected
uncertainty and Eq.~\eqref{eq:bound-GW}, we find the projected bound
\begin{align}
\left(\sqrt{|\alpha_{\GB}|} \right)^{\rm{GW}}_{\rm{BHNS}} \lesssim{}& 3.0 \; {\rm{km}} \left({\textstyle \frac{\sigma_{\beta_{\ppE}}}{5 \times 10^{-4}}}\right)^{1/4}
 \left({\textstyle \frac{m_{\BH}}{5 M_{\odot}}} \right)\left({\textstyle \frac{0.171}{\eta}}\right)^{1/10}
\label{eq:GB-GW-BH-NS}
\end{align}
for a $(1.4,5) M_{\odot}$ NS/BH binary, and
\begin{align}
\left(\sqrt{|\alpha_{\GB}|}\right)^{\rm{GW}}_{\rm{BHBH}} \lesssim{}& 3.4 \; {\rm{km}} \left({\textstyle \frac{\sigma_{\beta_{\ppE}}}{5 \times 10^{-4}}}\right)^{1/4}
\left({\textstyle \frac{m}{15 M_{\odot}}}\right)
\nn \\
& 
\left({\textstyle \frac{0.33}{\delta m/m}} \right)^{1/2}
\left({\textstyle \frac{\eta}{0.22}}\right)^{9/10}
\label{eq:DDGB-BHBH-Const}
\end{align}
for a $(10,5) M_{\odot}$ BH/BH binary, where $\delta m \equiv m_{1} - m_{2}$.
Eq.~\eqref{eq:DDGB-BHBH-Const} is in agreement with~\cite{Yagi:2011xp}, where the constraint was first calculated. 

How do these future constraints compare to current constraints? Recall that 
the observation of LMXBs has implied the constraint
$\sqrt{|\alpha_{\GB}|} < 1.9 \; {\rm{km}}$~\cite{kent-LMXB}. We then see that
pulsar observations would be able to improve this constraint by 1 order
of magnitude, while GW observations would lead to comparable constraints. 
One concludes that binary pulsars will be better at constraining
dipole radiation than GW observations, provided a BH/pulsar
binary is observed. These results are in agreement with those found in~\cite{cornish-PPE}.

\subsection{Non-Shift-Symmetric, Topological Quadratic Gravity: TEdGB Example}
\label{sec:non-shift-sym-topo}

Let us take TEdGB as an example of quadratic gravity theories with a non-shift-symmetric, non-derivative,
topological interaction density. This time, NSs do source a scalar charge, and thus, BHs
are not needed to activate scalar energy flux correction. We then expect that binary pulsar observations
of NS/WD and NS/NS systems will lead to strong estimated constraints. 
As before, we separate the discussion into observations that involve binaries 
with at least one BH and those without BHs. 

\subsubsection{Binaries without black holes}
\label{sec:TEDGB-no-BH}

Let us first concentrate on radio pulsar observations.  
The best constraints on TEdGB using binaries without BHs
will come from mixed
NS/WD observations, as these will have the most dissimilar
compactnesses, and thus, the difference in the scalar charges will not be inherently small. 
Using the leading-order term of Eq.~\eqref{eq:charge-EDGB-Tol} with
a Tolman VII model  in a $C \ll 1$ expansion, and choosing $\gamma=1$
and $\vartheta_\infty = 0$ to be consistent with~\cite{pani-NS-EDGB},
one finds the estimated constraint
\begin{align}
\left(\sqrt{|\alpha_{\TEDGB}|}\right)^{\rm{Bin.Pul.}}_{\rm{NSWD}} \lesssim{}& 1.4
 \; \mrm{km} \; \left( {\textstyle \frac{m_{\NS}}{1.46 M_\odot}} \right)^{-3/4}
\left( {\textstyle \frac{R_\NS}{12 \mrm{km}}} \right)^{7/4}
\nn \\
&  
{}\times
\left( {\textstyle \frac{|v_{12}/c|}{1.2 \times 10^{-3}}} \right)^{1/4}
\left( {\textstyle \frac{\sigma_{\dot P_{b}}/\dot P_{b}}{6.1 \times 10^{-2}}} \right)^{1/8}\,.
\end{align}
The values of the NS mass, the orbital velocity and the observational error 
that we chose to normalize the above constraint are those of J1738+0333~\cite{Freire:2012mg}. 

\begin{figure}[tb]
\begin{center}
\includegraphics[width=8.5cm,clip=true]{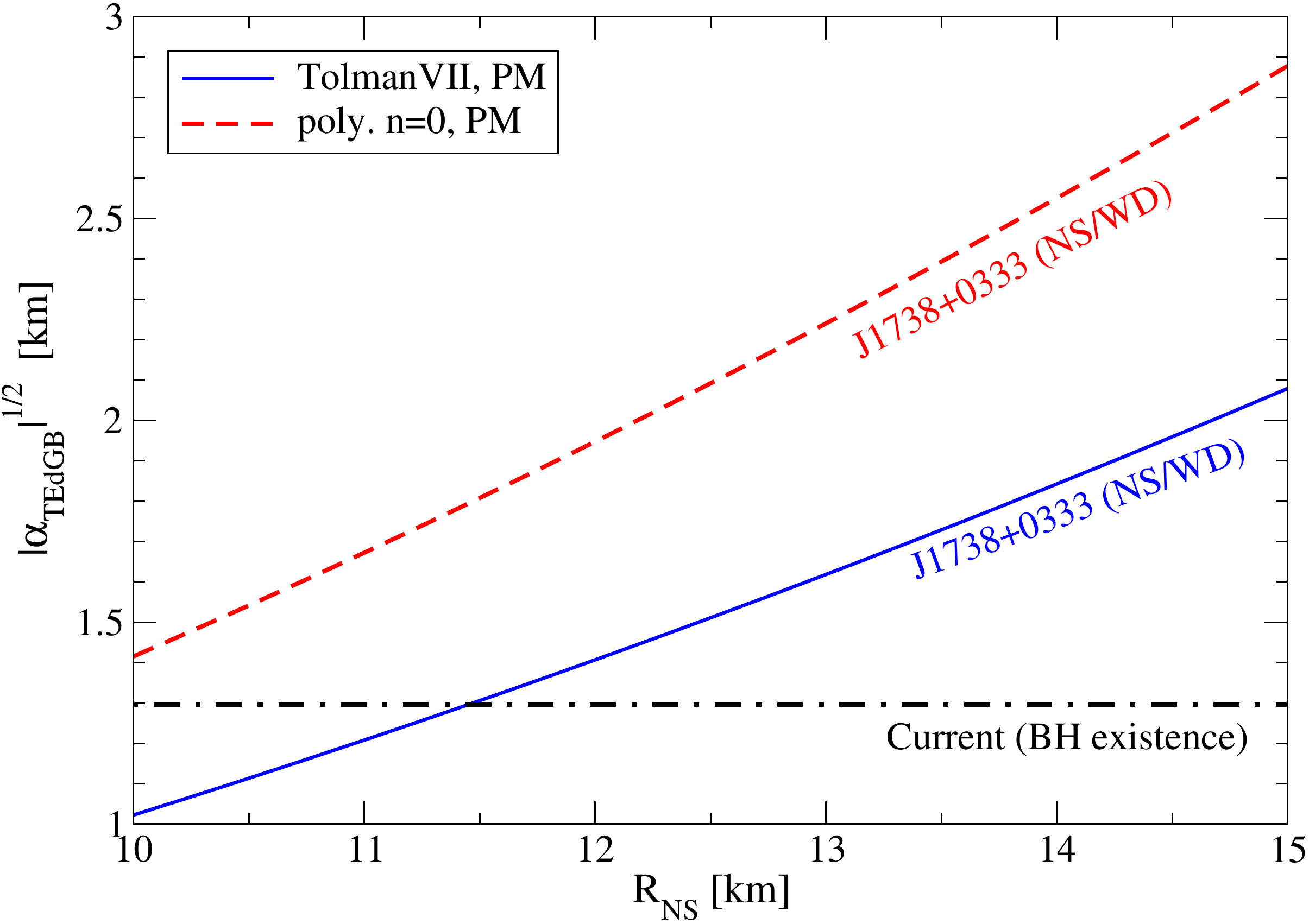}
\caption{\label{fig:constraints} (Color online) Upper bound on
  $\sqrt{|\alpha_{\TEDGB}|}$ from observations of the orbital
  decay rate of J1738+0333~\cite{Freire:2012mg} as a function of the
  NS radius with $\gamma = 1$ and $\vartheta_\infty = 0$. 
  For this constraint, we use the analytic post-Minkowskian
  calculation of the scalar charge with a Tolman VII and an
  $n=0$ polytropic equation of state (Eqs.~\eqref{eq:charge-EDGB-n0} and~\eqref{eq:charge-EDGB-Tol}). 
  For comparison, we also include
  the most stringent current constraint based on BH existence considerations~\cite{Kanti:1995vq,Pani:2009wy}.}
\end{center}
\end{figure}

The estimate found above clearly depends on the NS radius $R_{\NS}$, because
this quantity unavoidably enters the NS charge.
Figure~\ref{fig:constraints} shows this dependence for a Tolman VII model and an $n=0$ 
polytropic model using the mass, orbital velocity and measurement accuracy of 
J1738+0333~\cite{Freire:2012mg}. Observe that the estimated constraints are always between $1$
and $2$km for the Tolman VII model, which recall is a more realistic equation of state
than an $n=0$ polytrope. Observe also that the bounds on TEdGB gravity are comparable to the 
best current bounds that uses the existence of BH solutions~\cite{Kanti:1995vq,Pani:2009wy}. 
We have also studied the bound one could place on $\alpha_{\TEDGB}$ 
from observations of the WD/NS pulsar binary J0348+0432~\cite{2.01NS} 
and found it to be slightly weaker.

A similar constraint can be derived from the observation of the double NS pulsar binary 
J0737--3039~\cite{kramer-double-pulsar}. Following the same procedure as that described above,
but this time keeping both scalar charges, we find the estimated constraint
\begin{align}
\left(\sqrt{|\alpha_{\TEDGB}|}\right)^{\rm{Bin.Pul.}}_{\rm{NSNS}} \lesssim{}&
1.7 \; \mrm{km} \; 
\left( {\textstyle \frac{R_{1}}{11.5 \; {\rm{km}}}} \right)^{7/4}
\left( {\textstyle \frac{R_{2}}{12 \; {\rm{km}}}} \right)^{7/4}
\nn \\
&  \left( {\textstyle \frac{|v_{12}/c|}{2.08 \times 10^{-3}}} \right)^{1/4}
\left( {\textstyle \frac{\sigma_{\dot P_{b}}/\dot P_{b}}{1.7 \times 10^{-2}}} \right)^{1/8}
\nn \\
& \left[ 
 \left( {\textstyle \frac{m_{1}}{1.337 M_{\odot}}} \right)^{3} \left( {\textstyle \frac{R_{2}}{12 \; {\rm{km}}}} \right)^{7}
\right. 
\nn \\
& \left. 
{}-
\left( {\textstyle \frac{m_{2}}{1.250 M_{\odot}}} \right)^{3} \left( {\textstyle \frac{R_{1}}{11.5 \; {\rm{km}}}}  \right)^{7}
\right]^{-1/4}\,,
\label{eq:estimateNSNScons-TEDGB}
\end{align}
where we have assumed that $m_{1} > m_{2}$, and thus, $R_{2} \geq R_{1}$.
Observe that this constraint is comparable but slightly weaker than 
those obtained with J1738+0333. This is because in the
double NS binary pulsar case there is a natural suppression in the amount of scalar
energy flux emitted due to the comparable masses of the system, i.e.~when $m_{1} \sim m_{2}$, 
then $\mu_{1} \approx \mu_{2}$ and thus $\dot{E}^{(\vartheta)}$ is
suppressed.  Observe also that
this constraint depends on the radii of both NSs. Varying the radii of Eq.~\eqref{eq:estimateNSNScons-TEDGB}, 
we find that the estimated constraint varies between roughly $1.5$  
and $3.5$ km. 

Let us now estimate the magnitude of the constraints one could achieve on 
TEdGB gravity with a future GW observation. As in Sec.~\ref{sec:DDGB-cons}, 
we assume an aLIGO observation of the late inspiral of a NS/NS binary 
with an SNR of 20. Setting $\gamma = 1$ and $\vartheta_{\infty}=0$, we find the projected
constraint
\begin{align}
\left(\sqrt{|\alpha_{\TEDGB}|}\right)^{\rm{GW}}_{\rm{NSNS}} \lesssim{}& 7.9 \; {\rm{km}}
\left({\textstyle \frac{\sigma_{\beta_{\ppE}}}{5 \times 10^{-4}}}\right)^{1/8}
\left({\textstyle \frac{R_{1}}{11.5 \; {\rm{km}}}} \right)^{7/4}
\nn \\
& 
\left({\textstyle \frac{R_{2}}{12 \; {\rm{km}}}} \right)^{7/4}  \left({\textstyle \frac{0.249}{\eta}}\right)^{1/20}
\nn \\
& 
\left[ \left({\textstyle \frac{m_{1}}{1.6 M_{\odot}}}\right)^{3} \left({\textstyle \frac{R_{2}}{12 \; {\rm{km}}} }\right)^{7}
\right. 
\nn \\
& 
\left.
{}- 
\left({\textstyle \frac{m_{2}}{1.4 M_{\odot}}}\right)^{3} \left({\textstyle \frac{R_{1}}{11.5 \; {\rm{km}}} }\right)^{7}
\right]^{-1/4}
\,.
\end{align}
As before, these projected constraints also depend on the radii of both NSs. 
Varying these parameters, we find the constraint is always between $6$ and $14$ km.
Once more, we see that the projected constraints we could place with GWs are 
weaker than the estimated constraints with J1738+0333.
This is because in
TEdGB the scalar charge does not vanish, thus exciting dipole radiation which
enters at a pre-Newtonian order and could be better constrained by binary
pulsar observations, as discussed in Sec.~\ref{sec:dipole-rad-bin-gw}.

\subsubsection{Binaries with at least one black hole}

Let us first discuss future binary pulsar observations where one component of the binary
is a BH. The correction to the scalar energy flux 
is then dominated by the scalar charge of the BH, which we model through 
Eq.~\eqref{eq:BH-scalar-charge}.  This scaling from the \DDGBText{}
limit agrees with what is found in analytic calculations in
TEdGB~\cite{Kanti:1995vq}.
Given this, the projected constraint with a future BH-pulsar binary observation
is the same as that in Eq.~\eqref{eq:bound-EdGB-decouple} and Fig.~\ref{fig:constraints-BH-NS}. 
Such projected bounds are roughly an order of magnitude stronger 
than the estimated constraints that use binary systems without BHs. 

Let us now consider projected constraints with GW observations.
As in Sec.~\ref{sec:DDGB-cons}, let us assume an aLIGO observation 
with an SNR of 20 of the late inspiral of a compact binary. 
Setting $\gamma = 1$ and $\vartheta_{\infty}=0$, we find the same 
projected constraints as those in Eq.~\eqref{eq:GB-GW-BH-NS} 
for a NS/BH inspiral and Eq.~\eqref{eq:DDGB-BHBH-Const} 
for a BH/BH binary inspiral observation. In both cases, 
this is because both constraints use the same BH scalar charge 
as in \DDGBText{}. We see again that both constraints
improve by roughly an order of magnitude when including BHs
relative to those of Sec.~\ref{sec:TEDGB-no-BH}.
But again, these GW constraints would be weaker than
binary pulsar constraints, if the latter observed a BH/pulsar system. 

How do these future constraints compare to current constraints? 
Recall that current constraints on TEdGB gravity come from BH existence
considerations, which require that $\sqrt{|\alpha_{\TEDGB}|} <  1.4 \; {\rm{km}}$~\cite{Kanti:1995vq,Pani:2009wy}. 
We then see that projected constraints with BH/pulsar observations will improve this
bound by one order of magnitude, while projected GW constraints will be comparable
or slightly weaker than current constraints. 

\subsection{Non-Shift-Symmetric, Non-Topological Quadratic Gravity: Kretschmann Example}
\label{sec:non-shift-sym-non-topo}

Let us take Kretschmann gravity as an example of quadratic gravity theories with a non-shift-symmetric, non-derivative,
non-topological interaction density. As in the TEdGB case, NSs do source a scalar charge, and once more, 
BHs are not needed to activate corrections to the dynamics due to the scalar energy flux. We therefore expect
radio pulsar observations of NS/WD and NS/NS binaries to lead to strong estimated bounds.  
We again separate the discussion into observations that include binaries with at least one BH
and those without BHs. 

\subsubsection{Binaries without black holes}

Let us first focus on binary pulsar observations. 
As in the previous cases, the dominant correction
to the energy flux is given by Eq.~\eqref{eq:Edot-scalar}. 
In order to obtain estimated constraints on Kretschmann gravity,
we use the post-Minkowskian result for $\mu^{(0)}_\K$
with the Tolman VII model given in Eq.~\eqref{eq:mu-K}. 
For a WD/NS binary pulsar, Eq.~\eqref{eq:bound} 
leads to the estimated constraint 
\begin{align}
\label{eq:bound-Tolman}
\left(\sqrt{|\alpha_{\K}|}\right)^{\rm{Bin.Pul.}}_{\rm{NSWD}}  \lesssim{}& 0.076 \; \mrm{km} \; \left( {\textstyle \frac{m_{\NS}}{1.46 M_\odot}} \right)^{-1/2}
\left( {\textstyle \frac{R_\NS}{12 \mrm{km}}} \right)^{3/2}
\nn \\
& 
\left( {\textstyle \frac{|v_{12}/c|}{1.2 \times 10^{-3}}} \right)^{1/2}
\left( {\textstyle \frac{\sigma_{\dot P_b}/\dot P_{b}}{6.1 \times 10^{-2}}} \right)^{1/4}\,.
\end{align}
Observe that this is a rather strong estimated constraint, when compared to those found in \DDGBText{}
gravity and TEdGB gravity.

\begin{figure}[tb]
\begin{center}
\includegraphics[width=8.5cm,clip=true]{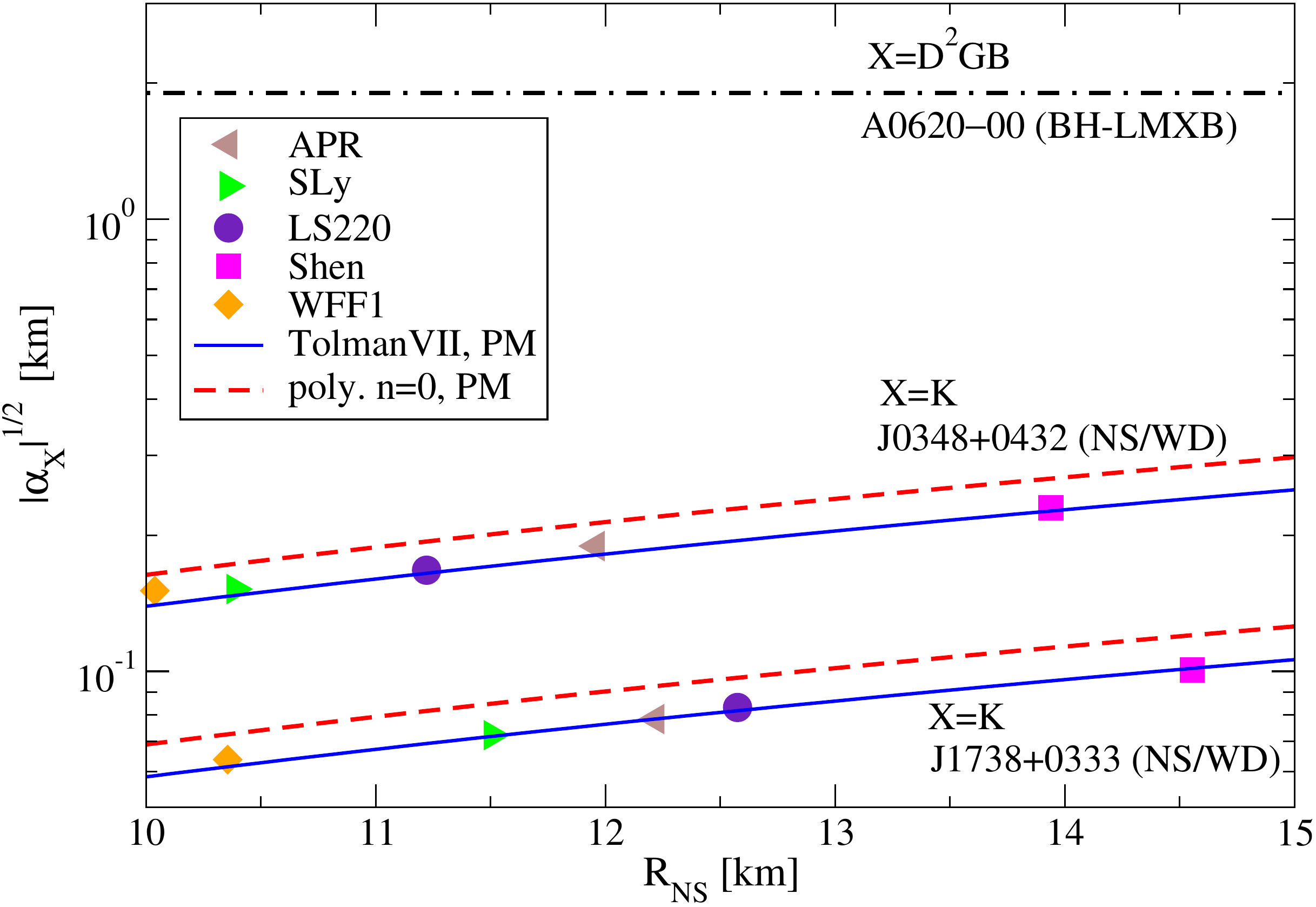}
\caption{\label{fig:constraints2} (Color online) Upper bound on
  $\sqrt{|\alpha_{\K}|}$ in Kretschmann gravity from observations
  of the orbital period decay of J1738+0333~\cite{Freire:2012mg} and
  J0348+0432~\cite{2.01NS} with various equations of state as a function
  of NS radius. Solid and
  dashed curves represent the bound using a Tolman VII model and an
  $n=0$ polytropic model respectively, and to leading-order in the
  post-Minkowskian approximation. For reference, we also present the
  estimated bound on \DDGBText{} gravity from observations of 
  A0620--00~\cite{kent-LMXB}.}
\end{center}
\end{figure}
As in the TEdGB gravity case, the estimated binary pulsar constraints 
depend on the NS radius. Figure~\ref{fig:constraints2} shows the upper bound on
$\sqrt{|\alpha_{\K}|}$ from the $\dot P_{b}$ measurement of
J1738+0333~\cite{Freire:2012mg} and J0348+0432~\cite{2.01NS} as a
function of NS radius for a set of tabulated equations of
state. The scalar charge of a non-rotating NS with 
such equations of state in Kretschmann gravity
is calculated numerically as explained in Sec.~\ref{sec:non-rot}.
The observation of J1738+0333 places a stronger constraint than
the observation of J0348+0432 because $\dot P_{b}$ for the former has
been measured more accurately.  For comparison, we also present
constraints using the leading-order, post-Minkowskian expression for
the scalar charge with an $n=0$ polytropic model and the
Tolman VII model. Observe that the constraints using realistic equations of state
lie very close to those obtained using the Tolman VII model.
We also show the constraint on \DDGBText{} gravity using
observations of A0620--00~\cite{kent-LMXB}.
Although one cannot directly compare estimated constraints on two
different theories and in different types of systems,
this difference suggests that the absence of NS scalar charge in
\DDGBText{} reduces our constraining power by an order of magnitude in
that theory.
Meanwhile, the projected constraint on Kretschmann gravity estimated here is
comparable to the projected constraint on \DDGBText{} gravity
estimated for a BH/pulsar binary presented in
Sec.~\ref{sec:DDGB-BHbinaries} and summarized in
Fig.~\ref{fig:constraints-BH-NS}.

We can repeat the above calculation for a binary pulsar composed of two NSs. 
Doing so, we find the estimated constraint
\begin{align}
\left(\sqrt{|\alpha_{\K}|}\right)^{\rm{Bin.Pul.}}_{\rm{NSNS}}  \lesssim{}& 0.19 \; \mrm{km} \;
\left( {\textstyle \frac{R_1}{11.5 \mrm{km}}} \right)^{3/2}
\left( {\textstyle \frac{R_2}{12 \mrm{km}}} \right)^{3/2}
\nn \\
&  
\left( {\textstyle \frac{|v_{12}/c|}{2.08 \times 10^{-3}}} \right)^{1/2}
\left( {\textstyle \frac{\sigma_{\dot P_b}/\dot P_{b}}{2.87 \times 10^{-2}}} \right)^{1/4}
\nn \\
& \left[
\left({\textstyle \frac{m_{1}}{1.337 M_{\odot}}}\right) \left({\textstyle \frac{R_{2}}{12 \; {\rm{km}}}}\right)^{3}
\right. 
\nn \\
&\left. 
{}- \left({\textstyle \frac{m_{2}}{1.250 M_{\odot}}}\right) \left({\textstyle \frac{R_{1}}{11.5 \; {\rm{km}}}}\right)^{3}
\right]^{-1/2}
\,.
\end{align}
Observe that this is weaker than that found for WD/NS
pulsar binary systems [Eq.~\eqref{eq:bound-Tolman}], in spite of the latter
typically being less relativistic and having less accurate $\dot P_{b}$ measurement.
Observe also that this estimated constraint depends
on both NS radii; varying both of these, we find that the estimated constraint
ranges between $0.15$ and $0.45$ km, and of course it becomes stronger as
the uncertainty in the measurement of $\dot{P}_{b}$ decreases.

Let us now discuss projected GW constraints on Kretschmann gravity
with signals emitted in the inspiral of NS binaries. Let us once more assume
an aLIGO detection with SNR 20. Then, using Eq.~\eqref{eq:bound-GW} and 
the post-Minkowskian result for $\mu^{(0)}_\K$ with the Tolman VII model given in Eq.~\eqref{eq:mu-K}, 
we find the estimated constraint 
\begin{align}
\left(\sqrt{|\alpha_{\K}|}\right)^{\rm{GW}}_{\rm{NSNS}} \lesssim{}& 4.1 \; {\rm{km}}
\left({\textstyle \frac{\sigma_{\beta_{\ppE}}}{5 \times 10^{-4}}}\right)^{1/4}
\left({\textstyle \frac{R_{1}}{11.5 \; {\rm{km}}}} \right)^{3/2}
\nn \\
& 
\left({\textstyle \frac{R_{2}}{12 \; {\rm{km}}}} \right)^{3/2}  \left({\textstyle \frac{0.249}{\eta}}\right)^{1/10}
\nn \\
& 
\left[ \left({\textstyle \frac{m_{1}}{1.6 M_{\odot}}}\right) \left({\textstyle \frac{R_{2}}{12 \; {\rm{km}}} }\right)^{3}
\right. 
\nn \\
& 
\left.
{}-
\left({\textstyle \frac{m_{2}}{1.4 M_{\odot}}}\right) \left({\textstyle \frac{R_{1}}{11.5 \; {\rm{km}}} }\right)^{3}
\right]^{-1/2}
\,.
\end{align}
As expected, this constraint depends on the radii of both NSs; varying these
quantities, we find the estimated constraint ranges between $3$ and $9$ km. 
Observe that this is roughly one order of magnitude weaker than binary pulsar constraints. 

\subsubsection{Binaries with at least one black hole}

Future radio observations of BH/NS pulsar binaries
could lead to complementary constraints. Since vacuum solutions in \DDGBText{} gravity
are also solutions in Kretschmann gravity with the identification 
$\alpha_\GB = \alpha_\K$, we can use some of the results from Sec.~\ref{sec:DDGB-BHbinaries}. 
Let us then again use Eq.~\eqref{eq:bound}, with the NS scalar charge
modeled through Eq.~\eqref{eq:mu-K} and the BH scalar charge through
Eq.~\eqref{eq:BH-scalar-charge}. We then find the projected constraint
\begin{align}
\left(\sqrt{|\alpha_{\K}|}\right)^{\rm{Bin.Pul.}}_{\rm{BHNS}}  
\lesssim{}& 0.049 \; \mrm{km} \;
\left( {\textstyle \frac{R_\NS}{12 \mrm{km}}} \right)^{3/2}
\left({\textstyle \frac{m_{\BH}}{10 M_{\odot}}} \right)
\left( {\textstyle \frac{|v_{12}/c|}{10^{-3}}} \right)^{1/2}
\nn \\
&  
\left( {\textstyle \frac{\sigma_{\dot P_b}/\dot P_{b}}{10^{-2}}} \right)^{1/4}
 \left[
7 \left({\textstyle \frac{R_{\NS}}{12 \; {\rm{km}}}}\right)^{3} \right. \nn \\
&  \left. {}+ 180 \left({\textstyle \frac{m_{\NS}}{1.4 M_{\odot}}}\right)
\left({\textstyle \frac{m_{\BH}}{10 M_{\odot}}}\right)^{2}
\right]^{-1/2}\,.
\end{align}
This projected constraint varies between $0.036$ and $0.074$~km as one
varies the NS radius. Notice that such a constraint is 2--3 times stronger 
than that on \DDGBText{} gravity for the same system observed 
[compare to Eq.~\eqref{eq:bound-EdGB-decouple}].
This is because the NS scalar charge in Kretschmann gravity
dominates the BH scalar charge, which leads to enhanced 
scalar dipole radiation compared to the \DDGBText{} case, where the 
stellar scalar charge vanishes.
The projected constraints in Fig.~\ref{fig:constraints-BH-NS} 
are also valid in Kretschmann gravity as an order of magnitude estimate,
with the contours representing 
an upper bound on $\sqrt{|\alpha_\K|}$.

Let us now discuss GW observations, assuming once more an aLIGO
detection with SNR 20. Using Eq.~\eqref{eq:bound-GW}, with the neutron scalar charge modeled through
Eq.~\eqref{eq:mu-K} and the BH charge through Eq.~\eqref{eq:BH-scalar-charge}, 
we find the projected constraint 
\begin{align}
\left(\sqrt{|\alpha_{\K}|}\right)^{\rm{GW}}_{\rm{BHNS}}  \lesssim{}& 2.5 \; \mrm{km} \;
\left( {\textstyle \frac{m_{\BH}}{10 M_\odot}} \right)
\left({\textstyle \frac{R_{\NS}}{11.5 \; {\rm{km}}}}\right)^{3/2}
\nn \\
& \left( {\textstyle \frac{0.1}{\eta}} \right)^{1/10}
\left( {\textstyle \frac{\sigma_{\beta_{\ppE}}}{5 \times 10^{-4}}} \right)^{1/4}
 \left[
7 \left({\textstyle \frac{R_{\NS}}{12 \; {\rm{km}}}}\right)^{3} \right. \nn \\
& \quad \left. {}+ 180 \left({\textstyle \frac{m_{\NS}}{1.4 M_{\odot}}}\right)
\left({\textstyle \frac{m_{\BH}}{10 M_{\odot}}}\right)^{2}
\right]^{-1/2}
\end{align}
for a BH/NS inspiral. Varying the radius of the NS, 
we find that the projected constraint varies between
$2.0$ and $4.1$ km. Similarly, for a BH/BH inspiral, we find 
projected constraints that are identical to Eq.~\eqref{eq:DDGB-BHBH-Const}, 
since both use the same BH scalar charge. 
As expected, the projected constraints are roughly two
orders of magnitude weaker than those that could be placed with BH/pulsar binaries. 

\section{Conclusions and Future Directions}
\label{sec:conclusions}




Many people may believe that gravity theories with a long-ranged
scalar field are tightly constrained by binary pulsar observations.
This is because many theories with a long-ranged scalar field have
been proposed, where NSs source a spherically symmetric, $1/r$ scalar
hair (a scalar charge).  An accelerating charge produces dipole
radiation, but the effects of dipole radiation have been strongly
constrained by binary pulsar observations.

In this paper we have addressed and tried to dispel this lore.
We first classified the
relevant theories in terms of whether each theory has (i) shift
symmetry, (ii) a coupling to a topological density, and (iii)
derivative or non-derivative interactions between the scalar field and
the metric.  We conjecture that theories with a shift-symmetric
scalar field sourced by a linear (non-derivative) coupling to a
topological density do not activate a scalar charge in NSs.

We proved this theorem specifically for \DDGBText{} gravity, which is
an example of this class.  Our rigorous proof is based on the generalized
Gauss-Bonnet-Chern theorem, and improves upon the ``physicist's''
proof given in~\cite{Yagi:2011xp}, which was only valid for static,
spherically symmetric stars.  We confirmed the absence of the NS
scalar charge in this theory by explicitly calculating the scalar
field around a slowly-rotating star to quadratic order in spin, both
analytically within the weak-field expansion and numerically in the
strong-field case.

Therefore, in order to place meaningful constraints on \DDGBText{}
gravity, one needs to use observations of compact binaries that
contain at least one BH. The current estimated constraint on the
theory using a BH-LMXB was derived in~\cite{Yagi:2012gp}.  We
here derived projected constraints from radio observations of
BH/pulsar binaries and GW observations of NS/BH and BH/BH binaries.
We found that radio observations could place constraints that are one
order of magnitude stronger than the current bound, while GW
constraints would be slightly weaker than the current one.

We also derived estimated and projected constraints on (i) TEdGB and (ii)
Kretschmann gravity as examples of (i) non-shift-symmetric but
topological, and (ii) non-shift-symmetric and non-topological theories
respectively.  In both cases, we found that ordinary stars acquire a
scalar charge, and hence, current binary pulsar observations
are sufficient in constraining these theories.  We found that such
binary pulsar constraints are comparable to the current bound from the
existence of BH solutions in TEdGB, while they are stronger than the
current bound from a BH-LMXB by roughly one order of magnitude in
Kretschmann gravity.  We also found that future BH/pulsar observations
could improve current binary pulsar constraints by roughly one order
of magnitude.

Let us comment on the relation between the scalar
charge, defined in this paper as the spherically symmetric, $1/r$
coefficient of the scalar field in a far-field expansion, and
the derivative of the ADM mass with respect to the scalar field at spatial infinity, as is 
routinely done in scalar-tensor theories. 
In such theories, one can calculate the variation of the ADM mass from the variation of the Lagrangian 
with respect to the metric and the scalar and matter fields~\cite{Damour:1992we}. 
Then, the authors in~\cite{Damour:1992we}
prove that the $1/r$ coefficient in the scalar field is equivalent to 
the derivative of the ADM mass with respect 
to the scalar field at infinity. In fact, this approach has already been applied to
Einstein-\AE{}ther and khronometric theories in~\cite{Yagi:2013ava}.
If one applies this approach to \DDGBText{} gravity, one finds that 
the variation of the stellar ADM mass with respect to the scalar field vanishes~\cite{Barausse:2015wia},
in agreement with the vanishing scalar charge found in this paper.

\subsection{Future work}
\label{sec:future-work}
One natural extension of this work would consider the Pontryagin
density, rather than the Gauss-Bonnet density, as the topological
invariant to which the scalar couples.  Whereas the integral of the
Gauss-Bonnet density is related to the manifold's Euler
characteristic, the integral of the Pontryagin density is related to
the first Pontryagin number.  The proof should be similar in
spirit to the one presented here.  However, it requires understanding
a different theorem on a pseudo-Riemannian manifold with boundary.

A second natural and very straightforward extension is to consider
BH spacetimes.  We may already outline what happens to our proof
in the BH case.  Firstly, the integration region would consist of a
spherical annulus crossed with a compactified time interval.  The
outer boundary is treated the same as here, but the inner boundary is
a null surface---the Killing horizon of the BH.  The Euler
characteristic is unchanged.  However, both sides of
the equality relating the boundary integrals
[Eq.~\eqref{eq:EOM-int-apply-GB-chi-kink}] need to include both the
inner and outer boundaries.  The inner boundary
will in fact contribute in this calculation.  We will find that the scalar
charge is proportional to a horizon integral, and from scaling
arguments we see that we will recover the known scaling $\mu\sim
1/m_\BH$~\cite{1992PhLB..285..199C,Mignemi:1992nt,
Yunes:2011we,Yagi:2011xp,Sotiriou:2014pfa}.  This approach can potentially connect to
the thermodynamics of Gauss-Bonnet BHs.

Yet a third natural extension is to derive a proposed bound on dynamical 
Chern-Simons gravity with BH/NS pulsar observations.
Corrections to some of the post-Keplerian parameters in this theory
have been derived in~\cite{Yagi:2013mbt}, which are absent if bodies are
non-rotating. One can extend the analysis in~\cite{Liu:2014uka} by including such 
corrections to the binary evolution and study how well one can constrain
the theory with a rotating BH/pulsar observations. If the observation
is accurate enough, one should be able to derive an upper bound on 
$\sqrt{|\alpha_\CS|}$ that is comparable to the size of a BH.
If this is the case, one would obtain a constraint comparable to that projected using
future GW observations derived in~\cite{kent-CSGW}, which is six orders of magnitude
stronger than the current bound.

Finally, a fourth extension is the same problem but in a cosmological
setting.  In this case, the spacetime is no longer asymptotically
flat.  Therefore, the calculations presented in this paper do not
apply, and NSs may acquire a scalar charge in \DDGBText{} gravity.
However, it is not clear how to define multiple moments of the metric and
scalar fields in a cosmological setting.

\acknowledgments
The authors thank
Enrico Barausse,
Mike Boyle,
Gilles Esposito-Far\`ese,
Paulo Freire,
Takahiro Tanaka,
and Norbert Wex
for useful discussions and comments.
We also thank Paolo Pani for providing us numerical data on the scalar
charge in TEdGB gravity to compare against.
L.C.S.~and N.Y.~would like to thank the Max Planck Institute for Radioastronomy
for their hospitality, while this work was started.
K.Y.~acknowledges support from JSPS Postdoctoral Fellowships for Research Abroad
and the NSF grant PHY-1305682.
L.C.S.~acknowledges that support for this work was provided by NASA
through Einstein Postdoctoral Fellowship Award Number PF2-130101
issued by the Chandra X-ray Observatory Center, which is operated by
the Smithsonian Astrophysical Observatory for and on behalf of the
NASA under contract NAS8-03060, and further acknowledges support from
the NSF grant PHY-1404569.
N.Y.~acknowledges support from NSF CAREER Grant PHY-1250636.
Some calculations used the computer algebra system MAPLE, in
combination with the GRTENSORII package~\cite{grtensor}.

\appendix

\section{Details of Calculating Neutron Star Scalar Charge
in Quadratic Gravity}

In this appendix, we describe details of how one calculate the scalar
monopole charge in quadratic gravity, both analytically 
 and numerically. We focus on the $l = 0$ mode
since the scalar charge enters through this mode.
We concentrate on the even-parity sector of the theory, namely, setting
$ f_4 = 0$ in the action in Eq.~\eqref{eq:q-action} 
(or $\alpha_4=0$ in Eq.~\eqref{eq:quad-action-we-focus-on}).

\subsection{Field Equations and Exterior Solutions}
\label{app:scalar-field-eq}

We first present the field equations and exterior solutions for NSs 
in quadratic gravity. We assume matter to be described by a perfect fluid and 
use the same metric ansatz as that proposed by Hartle and Thorne~\cite{hartle1967,Hartle:1968ht}.

\subsubsection{Non-rotating Neutron Stars}
\label{app:nonrot-NS}

Let us first consider non-rotating NSs with the coupling function in the action
given by Eq.~\eqref{eq:g-example}. The field equation is given by Eq.~\eqref{eq:scalar-field-eq-int} with
$S^{(0)}$ given by
\ba
\label{eq:source}
S^{(0)} &=&    \frac{8}{r^5 (r-2 M)} \gamma e^{-\gamma \vartheta^{(0)}} \left\{ 2 \alpha_3 M (3 M -8 \pi \rho r^3) \right. \nn \\
& & \left. {}+\pi^2  r^6  \left[ 49 \alpha_{\GB,1}   p^2 + 8   \alpha_{\GB,2} \rho^2 - 16 (3 \alpha_1 -  \alpha_3) p \rho  \right] \right\}\,. \nn \\
\ea
We decompose the scalar field $\vartheta^{(0)}$ in a manner similar to Eq.~\eqref{eq:scalar-decomp}: 
\begin{equation}
\label{eq:scalar-decomp-0}
\vartheta^{(0)} (r) = \vartheta_\infty +  \vartheta^{(0,1/2)} (r)
+   \vartheta^{(0,1)} (r) + \mathcal{O}(\zeta^{3/2})\,,
\end{equation}
where $\vartheta^{(0,A)}={\cal{O}}(\zeta^{A})$. 
Decomposing the field equation within the small coupling approximation using
Eq.~\eqref{eq:scalar-decomp-0} 
and  setting $p = 0 = \rho$ and $M=M_*$, one can solve for the exterior solutions and find
\bw
\allowdisplaybreaks
\begin{align}
 \label{eq:0th-Rot-ext-sol}
\vartheta_{\ext}^{(0,1/2)} ={}& -  \frac{2 }{M_* r} \left[ \alpha_3  \gamma e^{-\gamma \vartheta_\infty} \left(1 +\frac{M_*}{r} + \frac{4}{3} \frac{M_*^2}{r^2} \right)
  + \frac{1}{4} \left(\mu^{(0,1/2)} M_*^2 + 2 \alpha_3 \gamma e^{-\gamma \vartheta_\infty} \right) \frac{r}{M_*} \ln \left( 1 - \frac{2 M_*}{r} \right)  \right]\,, \\
 \label{eq:0th-Rot-ext-sol-2}
\vartheta_{\ext}^{(0,1)} ={}&  \frac {11}{6\,{{\it M_*}}^{4}{r}
} \left\{  \left[ \frac {147}{110} \,{\alpha_3}^{2}{\gamma}^{3}{{\rm e}^
{-2\,\gamma\,\vartheta_\infty}}  \left(1 - \frac{40}{49}\frac{M_*}{r} - \frac{40}{49}\frac{M_*^2}{r^2}  - \frac{160}{147}\frac{M_*^3}{r^3} \right) \right. \right. \nn \\
 & \left. \left.
 {}+\mu^{(0,1/2)} \alpha_3\, {\gamma}^{2} {{\rm e}^{-\gamma\,\vartheta_\infty}} M_*^2\left(1 - \frac{6}{11}\frac{M_*}{r} - \frac{6}{11}\frac{M_*^2}{r^2}  - \frac{8}{11}\frac{M_*^3}{r^3} \right) - \frac{3}{
11}\,\mu^{(0,1)}\,{{\it M_*}}^{4} \right] r
 \ln  \left(1 -  {\frac {2
\,{\it M_*}}{r}} \right)  \right. \nn \\
 & \left. {}+2 \alpha_3\, {\gamma}
^{2} {\it M_*}\, \left[ {\frac {147}{110} \,\alpha_3 
\,\gamma\,{{\rm e}^{-2\,\gamma\,\vartheta_\infty}}  \left(1+ \frac{9}{49} \frac{M_*}{r} - \frac{44}{147} \frac{M_*^2}{r^2} - \frac{146}{147} \frac{M_*^3}{r^3} - \frac{64}{105} \frac{M_*^4}{r^4}- \frac{160}{441} \frac{M_*^4}{r^4} \right) } \right. \right. \nn \\
 & \left. \left.
{}+\mu^{(0,1/2)}\, M_*^2 {{\rm e}^{-\gamma\,\vartheta_\infty}}  \left(1 - \frac{5}{11}\frac{M_*}{r} - \frac{8}{33}\frac{M_*^2}{r^2} \right)  \right]  \right\}\,,
\end{align}
\ew
where we decomposed $\mu^{(0)}$ within the small coupling approximation
in a similar manner to Eq.~\eqref{eq:scalar-decomp-0} as
\begin{equation}
 \mu^{(0)} = \mu^{(0,1/2)}
+ \mu^{(0,1)} + \mathcal{O}(\zeta^{3/2})\,.
\end{equation}
We have also used the exterior solutions for the metric in GR,
which is justified since the quadratic gravity correction to the metric only affects 
the scalar field exterior solution at $\mathcal{O}(\zeta^{3/2})$,
even when $\vartheta_\infty$ does not vanish. 
This is because the non-shift-symmetric contribution to the metric field equations vanish
in the exterior region, and hence, $\mathcal{O}(\zeta^{1/2})$ contribution sourced by $\vartheta_\infty$
is absent in the metric exterior solution.
Notice that $\mu^{(0,1/2)}$ and $\mu^{(0,1)}$ are the only integration constants, which by
construction coincide with the scalar charge, because we
have removed the other constants through a redefinition of the constant
$\vartheta_\infty$ in Eq.~\eqref{eq:scalar-decomp}.

\subsubsection{Slowly-rotating Neutron Stars}
\label{app:rot-NS}

Next, we derive the scalar field equations and exterior solutions for slowly-rotating NSs.
We assume the scalar coupling function is given by 
\begin{equation}
\label{eq:g-lin}
g(\vartheta) = \vartheta\,,
\end{equation}
which is a subclass of the non-linear function of Eq.~\eqref{eq:g-example} 
used in Sec.~\ref{app:nonrot-NS} by taking the limit $\gamma \to 0$ while replacing
$\alpha_i$ to $-\alpha_i/\gamma$ and keeping this new $\alpha_i$ constant. The purpose
of calculating the scalar charge with the linear coupling function is to demonstrate explicitly that
the charge vanishes in the \DDGBText{} limit.
Following Eq.~\eqref{eq:scalar-decomp-0}, we decompose $\vartheta^{(2)}_l$ as
\begin{equation}
\label{eq:scalar-decomp-2}
\vartheta^{(2)}_l (r) =   \vartheta^{(2,1/2)}_l (r)
+   \vartheta^{(2,1)}_l (r) + \mathcal{O}(\zeta^{3/2})\,, \quad (l = 0,2)\,,
\end{equation}
where $\vartheta^{(2,A)}_l={\cal{O}}[(M_{*} \Omega)^{2},\zeta^{A}]$ 
and $\vartheta^{(2)}_l (\infty)$ can be set to zero
without loss of generality. 
For a linear scalar field coupling function, 
the $\mathcal{O}(\zeta)$ contribution  vanishes ($\vartheta^{(2,1)}_0= 0$)
in the exterior region, and hence, 
one only needs to calculate $\vartheta^{(2,1/2)}_0$.
%
The scalar field equation and the exterior solutions for non-rotating NSs with
the linear coupling function can easily be derived from the results in Sec.~\ref{app:nonrot-NS}
with the mapping explained above as
\ba
\label{eq:source-lin}
S^{(0)} &=&  -  \frac{8}{r^5 (r-2 M)} \left\{ 2 \alpha_3 M (3 M -8 \pi \rho r^3) \right. \nn \\
& & \left. {}+\pi^2  r^6  \left[ 49 \alpha_{\GB,1}   p^2 + 8   \alpha_{\GB,2} \rho^2 - 16 (3 \alpha_1 -  \alpha_3) p \rho  \right] \right\}\,, \nn \\
\\
 \label{eq:0th-Rot-ext-sol-lin}
\vartheta_{\ext}^{(0,1/2)} &=& \frac{2 }{M_* r} \left[ \alpha_3 \left(1 +\frac{M_*}{r} + \frac{4}{3} \frac{M_*^2}{r^2} \right)
 \right. \nn \\
 & & \left. {}- \frac{1}{4} \left(\mu^{(0,1/2)} M_*^2 - 2 \alpha_3  \right) \frac{r}{M_*} \ln \left( 1 - \frac{2 M_*}{r} \right)  \right]\,. \nn \\
\ea

Regarding the scalar field equation at second order in spin,
$\vartheta^{(2)}_0$ obeys the same equation as Eq.~\eqref{eq:scalar-field-eq-int}, 
except $S^{(0)}$ needs to be replaced by $S^{(2)}$, which is given by
\bw
\allowdisplaybreaks
\begin{align}
\label{eq:source-2}
S^{(2)} ={}&  \frac {1}{6 {r}^{6} \left( r-2{\it M} \right) ^{3}} \left\{
1024{r}^{5} \left( r-2{\it M} \right) {
{\rm e}^{-\nu}}
 \left\{ -\frac{\left( r-2{\it M} \right) {{\it \omega_1'}}^{2}}{32} \left[ -\frac{1}{
32}{r}^{4} \left( r-2{\it M} \right) \vartheta^{(0)}{}'-2 \alpha_3 {{\it M}}^{
2}  \right. \right. \right. \nn \\
& \left. \left. \left. {}+ \alpha_3 \left( \pi \rho{r}^{2}+\pi p{r}^{2}+\frac{7}{2} \right) r
{\it M}
+ {r}^{2} \left( {\pi }^{2} \left( \alpha_{\GB,1}-\alpha_{\GB,2} \right) {r}^{4}  {\rho}^{2} -{\pi }^{2} \left( \alpha_{\GB,1}+\alpha_{\GB,2} \right) {r}^{4}p
  \rho-2{\pi }^{2}
\alpha_{\GB,1}{r}^{4}{p}^{2}- \frac{5}{4} \alpha_3 \right) 
 \right]   \right. \right. \nn \\
& \left. \left.
{}-\frac{ \pi }{4}
 \left( \rho+p \right) {r}^{2} \left( r-2{\it M} \right)  \omega_1  {\it \omega_1'
}
 \left\{ -
\alpha_3{\it M} + r \left[\pi \left( \alpha_{\GB,1}-
\alpha_{\GB,2} \right) {r}^{2}   \rho-2 \pi \alpha_{\GB,1}{r}^{2} 
p+\alpha_3 \right]  \right\}
\right. \right. \nn \\
& \left. \left.
{}+ \pi \left( \rho+p \right) r \omega_1^2 \left\{
\frac {1}{64} {r}^{4} \left( r-2{\it M
} \right) 
\vartheta^{(0)}{}' - \frac{3}{4} \alpha_3{{\it M}}^{2}+ \frac{3}{4} r \left[ 
 \left( \alpha_{\GB,1}-\alpha_{\GB,2}
 \right)  \pi {r}^{2} \rho-2 \pi  {r}^{2} \left( \alpha_{\GB,1}+\frac{2}{3}\alpha_3
 \right) p+\frac{2}{3}\alpha_3 \right] {\it M}  
 \right. \right. \right. \nn \\
& \left. \left. \left.
 {}+  \pi {r}^{4} \left[
 \left( \alpha_{\GB,1}-\alpha_{\GB,2} \right)  \left( \pi 
p{r}^{2}-\frac{1}{4} \right) \rho - 2 \left(\pi  \alpha_{\GB,1}{r}^{2}
p-\frac{1}{4}\alpha_{\GB,1}- \frac{1}{2} \alpha_3 \right) p \right]
 \right\}   \right\} 
 +6{r}^{6} \vartheta^{(0)}{}' \left( r-2{\it M}
 \right) ^{3}\xi_0'' \right.  \nn \\
 & \left. 
 {}+1536r \left\{ \frac{1}{128} \left( 2 \pi {r}^{3}
\rho-2 \pi p {r}^{3}- r+ {\it M} \right)  \left( r-2{\it M}
 \right) {r}^{4} \vartheta^{(0)}{}'+\frac{3}{4}\alpha_3{{\it M}}^{3}-\frac{\alpha_3}{4}
 \left( 7 \pi \rho{r}^{2}- \pi p{r}^{2}+\frac{3}{2} \right) r
{{\it M}}^{2}  
\right. \right.  \nn \\
 & \left. \left.
{}+\frac{\pi}{4} {\it M} {r}^{4} \left\{ \pi  \left( \alpha_{\GB,1}+3\alpha_{\GB,2} \right) {r}^{2} {\rho}^{2}+ \left[ -13 \pi
 \left( \alpha_{\GB,1}-{\frac {11}{13}}\alpha_{\GB,2}+{
\frac {12}{13}}\alpha_3 \right) {r}^{2} p+4\alpha_3 \right]
\rho 
\right. \right. \right.  \nn \\
 & \left. \left. \left.
 {}+10 \pi \left( \alpha_{\GB,1}+\frac{2}{5}\alpha_3 \right) {r}^{2}{p}^
{2} \right\} 
+  {\pi }^{2}  {r}^{7} 
\left\{ \pi 
 \left( \alpha_{\GB,1}-\alpha_{\GB,2} \right)  \left( {r}^{2} p-\frac{1}{2}
\alpha_{\GB,2} \right) {\rho}^{2}
\right. \right. \right.  \nn \\
 & \left. \left. \left.
{}- \left[ \pi  \left(
\alpha_{\GB,1}+\alpha_{\GB,2} \right) {r}^{2} p-\frac{3}{2} \alpha_{\GB,1}+\frac{3}{2}\alpha_{\GB,2}-2\alpha_3 \right] 
p \rho-2 \alpha_{\GB,1}
 \left( \pi p{r}^{2}+\frac{3}{4} \right) {p}^{2}
 \right\} 
 \right\}  
 \left( r-2{\it M} \right) \xi_0'
 \right.  \nn \\
 & \left.
 {}+192 \left\{ \frac{1}{8}{{\it M}}^{2}\xi_0-\frac{1}{2}r \left( \pi \rho{r}^
{2}-\pi p{r}^{2}+\frac{1}{4} \right) \xi_0{\it M}+{r}^{2} 
\left[ {\pi }^{2} {r}^{4
}{\rho}^{2}\xi_0-  \pi {r}^{2} \rho \left( \pi {r}^{2}p\xi_0 - \frac{m_0}{
4}+\frac{\xi_0}{8} \right)
\right. \right. \right.  \nn \\
 & \left. \left. \left.
 {}-\frac{\pi}{8}{r}^{2}  \left( \xi_0+
2m_0 \right) p+\frac{\xi_0}{16}- \frac{m_0}{16} \right]  \right\}
{r}^{
4} \left( r-2{\it M} \right) \vartheta^{(0)}{}'+5760\alpha_3 {{\it M}}^{4}\xi_0
 \right.  \nn \\
 & \left.
 {}-9984 \left(\pi {r}^{2} \rho \xi_0 - \frac{\pi}{13}{r}^{2}p\xi_0 +{\frac {3}{26}}m_0+{\frac {33}{52}}\xi_0
 \right) r\alpha_3{{\it M}}^{3}+768 
 \left\{ {\pi }^
{2} \left( \alpha_{\GB,1}-3\alpha_{\GB,2} \right) {r}^{4} {\rho} \xi_0^{2}
\right. \right.  \nn \\
 & \left. \left.
{}+5 \pi {r}^{2} \rho \xi_0 \left[ \pi  \left( \alpha_{\GB,1}-\frac{7}{5}
\alpha_{\GB,2} \right) {r}^{2}  p +{\frac {18}{5}} \alpha_3
 \right] 
 - 8{\pi }^{2} \left( \alpha_{\GB,1}+\alpha_3 \right) {r}^{4} {p}^{2}\xi_0 -\pi \alpha_3 {r}^{2}  p\xi_0+\frac{9}{4}
\alpha_3 \left( m_0+\xi_0 \right)  \right\}
{r}^{2}{{\it M}}
^{2}
\right.  \nn \\
 &  \left.
 {}-24576{r}^{3} \left\{ -\frac{{\pi }^{3}}{4}\alpha_{\GB,2}{r}^{6}{\rho}^{3}\xi_0 + \pi^2 {r}^{4} \rho^2
 \left[ \pi  \left( \alpha_{\GB,1}-\alpha_{\GB,2}+\frac{1}{2} \alpha_3 \right) {r}^{2} p \xi_0 + \left( \frac{1}{32}
\alpha_{\GB,1}-\frac{1}{16}\alpha_{\GB,2}+\frac{3}{16}\alpha_3
 \right) \xi_0
 \right. \right. \right.  \nn \\
 & \left. \left. \left.
 {}-\frac{1}{16} \alpha_{\GB,2} m_0 \right]
 - \left\{ {\pi }^{2} \left( \alpha_{\GB,1}+\frac{1}{4}\alpha_{\GB,2}+\frac{1}{2}\alpha_3 \right) {r}^{4}{p}^{2} \xi_0 -\frac{\pi}{16} r^2 p
 \left[  \left( \alpha_{\GB,1}-2\alpha_{\GB,2} \right) 
\xi_0
  \right. \right. \right. \right. \nn \\
 & \left. \left. \left. \left.
 {}+3m_0 \left( \alpha_{\GB,1}-\alpha_{\GB,2}+\frac{2}{3}\alpha_3 \right)  \right]
 -\frac{3}{16}\alpha_3 \left( 
\xi_0+\frac{m_0}{4} \right)  \right\} 
\pi {r}^{2} \rho-\frac{{\pi }^{3}}{2} \alpha_{\GB,1} {r}^{6}
{p}^{3} \xi_0 
\right. \right. \nn \\
 & \left. \left.
 {}-\frac{5}{32} {
\pi }^{2} {r}^{4}{p}^{2}
 \left[  \left( \alpha_{\GB,1}+\frac{2}{5}\alpha_3 \right) \xi_0+\frac{6}{5}
 \left( \alpha_{\GB,1}+\frac{2}{3}\alpha_3 \right) m_0 \right] 
-{\frac {\pi }{64}} \alpha_3 {r}^{2} p m_0
+{\frac {3}{128}}  \alpha_3m_0 \right\} 
{\it M} 
\right. \nn \\
 & \left.
 {}+12288\pi {r}^{6
} \left\{ {\pi }^{2} {r}^{4} {\rho}^{3}  \xi_0 \left[ \pi \left( \alpha_{\GB,1}-\alpha_{\GB,2} \right) {r}^{2}   p+\frac{1}{8}\alpha_{\GB,1}-\frac{3}{8}
\alpha_{\GB,2} \right] - \pi {r}^{2}   {\rho}^{2}
 \left\{ {\pi }^{2} \left( \alpha_{\GB,1}+\alpha_{\GB,2} \right) {r}^{4}  {p}^{2} \xi_0 
 \right. \right. \right. \nn \\
 & \left. \left. \left.
 {}- \frac{3}{4}{r}^{2}  \pi p
 \left[  \left( \alpha_{\GB,1}-\frac{4}{3}\alpha_{\GB,2}+\frac{4}{3}\alpha_3 \right) \xi_0+\frac{1}{3}
\left( \alpha_{\GB,1}-\alpha_{\GB,2} \right) m_0 
 \right]
 -\frac{1}{4} \alpha_3 \xi_0-\frac{1}{32} \left( \alpha_{\GB,1}-3\alpha_{\GB,2} \right) m_0  \right\}
 \right. \right. \nn \\
 & \left. \left.
 {}+
 \left\{ -2{r}^{6}\alpha_{\GB,1}\xi_0{\pi }^{3}{p}^{3}
 -
\frac {9}{8} {\pi }^{2} \left[  \left( \alpha_{\GB,1}+\frac{1}{9}
\alpha_{\GB,2} \right) {r}^{4} {p}^{2} \xi_0+\frac{2}{9} \left( \alpha_{\GB,1}+\alpha_{\GB,2} \right) m_0 \right] 
\right. \right. \right. \nn \\
 & \left. \left. \left.
 {}+\frac{5}{32}  \pi
\left( \alpha_{\GB,1}-\frac{7}{5}\alpha_{\GB,2}+\frac{8}{5}
\alpha_3 \right) {r}^{2}  p m_0 
+\frac{1}{16}\alpha_3m_0
 \right\} 
 \rho-\frac{\pi }{4} \alpha_{\GB,1}{r}^{2}{p}^{2} \left( \pi  {r}^{2} p \left( \xi_0+2
m_0 \right) +m_0 \right)  
 \right\}  \right\}\,,
\end{align}
\ew
where the primes represent derivatives with respect to $r$.
$\omega_1$ and $m_0$ are the $l=1$ mode of the $(t,\phi)$ component and
the $l=0$ mode of the $(r,r)$ component
of the metric perturbation, while $\xi_0$ is the $l=0$ mode of the perturbation to the radial coordinate
such that the perturbations to $p$ and $\rho$ vanish~\cite{hartle1967,Hartle:1968ht}.
Solving such an equation in the exterior region within the small coupling approximation, 
one finds the exterior solution for the scalar field:
\bw
\begin{align}
\label{eq:vartheta20-ext}
\vartheta^{(2,1/2)}_{0,\ext}(r) ={}& \frac {1}{30  M_{*} \left( r-2 {\it M_*} \right)
} \left\{ 15    \left[ - \left(  \mu^{(0,1/2)} M_*^2  - \alpha_3 \right) {\chi}^{2} +2{
\it \delta m}  \left(  \mu^{(0,1/2)} M_*^2 -4 \alpha_3 \right)  \right]  \right. \nn \\
&  \left.
 {}+15 \left[     \left( \mu^{(0,1/2)} M_*^2  - \alpha_3  \right) {\chi}^{2} - 2
 \left( \mu^{(0,1/2)} M_*^2 \xi_{0,*} - 4 \alpha_3 {\it \delta m} \right) \right] \frac {M_{*}}{r}
+ 10 \left[  \left( \mu^{(0,1/2)} M_*^2  - \alpha_3  \right) {\chi}^{2} + 8
  {\it \delta m} \right] \frac{M_{*}^{2}}{{r}^{2}} \right. \nn \\
   &  \left.
   {} +
10 \left[  \left( \mu^{(0,1/2)} M_*^2  - \alpha_3  \right) {\chi}^{2}-16
  {\it \delta m} \right] \frac { M_{*}^{3}}{{r}^{3}}
  +
12 \left( 40  \xi_{0,*} -{\chi}^{2}
 \right) \frac {M_{*}^{4}}{{r}^{4}}
 +224 {\chi}^{2}  \frac { M_{*}^{5}}{{r}^{5}}
-160 {\chi}^{2}  \frac {M_{*}^{6}}{{r}^{
6}} \right. \nn \\
&
\left. {}- \frac {  15 }{2} \frac{r}{M_*} \left( 1 - \frac{2M_*}{r} \right) \left[ 2 \mu^{(2,1)} M_*^2 - 2{\it
\delta m}  \left( \mu^{(0,1/2)} M_*^2  - 4 \alpha_3  \right) + {\chi}^{2} \left( \mu^{(0,1/2)} M_*^2  - \alpha_3 
 \right) \right] \ln  \left( 1-2 {\frac {{\it
M_*}}{r}} \right) \right\} \,,
\end{align}
\ew
where $\mu^{(2,1/2)}$ is the only integration constant
that corresponds to the dimensionless scalar charge
at second order in spin and to leading order in $\zeta$.
$\delta m$ is the fractional correction to the stellar mass at second order in rotation
while 
$\xi_{0,*}$ corresponds to $\xi_0/M$ at the surface.
We have set $\xi_0$ in the exterior region to a constant, namely
$\xi_{0,\ext} (r) = \xi_0(R_*)$.

\subsection{Numerical Scheme}
\label{app:num}

Let us next explain the numerical algorithm that we use to calculate
the scalar charge of slowly-rotating isotropic NSs 
with a variety of realistic equations of state in quadratic gravity.
We assume the linear coupling function of the scalar field in
Eq.~\eqref{eq:g-lin} and use the field equation and exterior solutions derived
in Sec.~\ref{app:rot-NS}.

The equations of state we employ are the following tabulated ones: 
APR~\cite{APR}, SLy~\cite{SLy}, LS220~\cite{LS},
Shen~\cite{Shen1,Shen2} and WFF1~\cite{Wiringa:1988tp}. 
These equations of state are found by
solving certain many-body quantum field theory equations for the
internal pressure and density at supra-nuclear densities. Due to the
difficulty of solving these equations and uncertainties about the
strength of certain interactions, different approximations are made
that lead to different equations of state.  
We also continue to consider an $n=0$ polytropic equation
of state and a Tolman VII model to allow for comparisons with the
analytical study section.

A numerical solution to the scalar field evolution equation requires
boundary conditions. We choose to specify these at the stellar center
through a local analysis of the differential equation, which leads to
\begin{align}
\vartheta^{(0,1/2)} ={}& \vartheta_{c}^{(0,1/2)} - \frac{32 \pi^2}{9 R_*^2} \Bigg[ (3 \alpha_{\GB,2} - 4 \alpha_3) \nn \\
&  - 6 (3 \alpha_1 - \alpha_3) \frac{p_c}{\rho_c} + 9 \alpha_{\GB,1} \left(\frac{p_c}{\rho_c}\right)^2 \Bigg] x^2 \nn \\
&  + \mathcal{O}(x^3)\,, \\
\vartheta_{0}^{(2,1/2)} ={}& \vartheta_{0,c}^{(2,1/2)} - \frac{16 \pi}{9R_*^2} \frac{\omega_{1,c}^2 e^{-\nu_c}}{\rho_c + 3 p_c} \Bigg[ (3 \alpha_{\GB,2} -4  \alpha_3) \nn \\
& - 6 (3 \alpha_1- \alpha_3) \frac{p_c}{\rho_c}+ 9  \alpha_{\GB,1} \left(\frac{p_c}{\rho_c}\right)^2 \Bigg] x^2 \nn \\
& + \mathcal{O}(x^3)\,,
\end{align}
where $\nu=\nu(r)$ and $\omega_{1}=\omega_{1}(r)$ are metric functions
at zeroth- and first-order in rotation. We have defined the
expansion parameter $x \equiv R_* r \rho_{c} \ll 1$ and where the
subscript $c$ stands for the value of the function at the stellar
center. We initiate our integrations at a core radius $r=r_{c}>0$,
whose value we choose ensuring the local analysis presented above is
valid, namely $R_* \rho_{c} r_{c} \ll 1$. For example, setting
$R_*=12 \; {\rm{km}}$ and
$\rho_{c} = 10^{15} \; {\rm{g}}/{\rm{cm}}^{3}$, which is the typical
NS central density, the constraint $x(r=r_{c}) \ll 1$
implies $r_{c} \ll 10^{7} \; {\rm{cm}}$; we choose here
$r_{c} = 10 \; {\rm{cm}}$ which satisfies this bound.

The interior solution to the scalar field evolution equation can be
found numerically as follows. First, we choose an arbitrary trial
value for the boundary value constants 
($\vartheta^{(0,1/2)}_{c}$, $\vartheta^{(2,1/2)}_{0,c}$), with
which we find homogeneous ($\vartheta^{(0,1/2)}_{\HH}$, 
$\vartheta^{(2,1/2)}_{0,\HH}$) and particular
($\vartheta^{(0,1/2)}_{\p}$, $\vartheta^{(2,1/2)}_{0,\p}$) solutions through a fourth-order
Runge-Kutta integrator. The solution that satisfies the proper
boundary conditions at the stellar surface must then be a linear
combination of these two solutions, namely
\begin{align}
\vartheta^{(0,1/2)} &= C^{(0)}_{\vartheta} \;   \vartheta^{(0,1/2)}_{\HH} +\vartheta^{(0,1/2)}_{\p}\,, \\
\vartheta^{(2,1/2)}_{0} &= C^{(2)}_{\vartheta} \;  \vartheta^{(2,1/2)}_{0,\HH} + \vartheta^{(2,1/2)}_{0,\p}\,,
\end{align}
where $C^{(A)}_{\vartheta} \ (A=0,2)$ is an integration constant to be
determined by matching at the stellar surface.

With the interior and exterior solutions at hand with the latter given by 
Eqs.~\eqref{eq:0th-Rot-ext-sol-lin} and~\eqref{eq:vartheta20-ext}, we match them at the
stellar surface at each order in $M_{*} \Omega$.  At zeroth-order in
rotation, the matching condition is simply
\begin{equation}
\label{eq:matching}
\vartheta^{(0,1/2)} (R_*) = \vartheta^{(0,1/2)}_{\ext} (R_*)\,, \ \vartheta^{(0,1/2)}{}'(R_*) = \vartheta^{(0,1/2)}_{\ext}{}'(R_*)\,,
\end{equation}
where the primes denote differentiation with respect to $r$. 
In this paper, we choose the radial deformation function
$\xi_{0}$ in the exterior region to be constant as
$\xi_{0,\ext} (r) = \xi_{0} (R_*)$. In such a case,
the matching condition at second-order in rotation is
\ba
\label{eq:matching201}
\vartheta^{(2,1/2)}_{0} (R_*) &=& \vartheta^{(2,1/2)}_{0,\ext} (R_*)\,, \\
\label{eq:matching202}
\vartheta^{(2,1/2)}_{0}{}'(R_*)  &=& \vartheta^{(2,1/2)}_{0,\ext}{}'(R_*) + \vartheta^{(0,1/2)}_{0,\ext}{}' (R_*) \xi^{(2)}_{0}{}'(R_*)\,. \nn \\
\ea
The second term in Eq.~\eqref{eq:matching202} is required to ensure
smoothness of the scalar field at the stellar surface and it is
induced by the non-smoothness of 
$\xi^{(2)}_{0}$ at the stellar surface.  These four conditions
determine $C^{(A)}_{\vartheta}$ and $\mu^{(A)}$ at each order in
rotation.

\subsection{Strongly Anisotropic Neutron Stars}
\label{app:ani}

Let us finally explain how one can show that the scalar charge vanishes
in \DDGBText{}  gravity for strongly anisotropic NSs without using the 
post-Minkowskian expansion. Such a star can be modeled through
a stress-energy tensor of the form~\cite{Doneva:2012rd,Silva:2014fca}
\begin{equation}
T_{\mu \nu} = \rho \; u_\mu u_\nu + p_{r} \; k_\mu k_\nu + q_{t} \; \Pi_{\mu \nu}\,,
\end{equation}
where $k^\mu$ is the unit radial four-vector orthogonal to the
timelike four-velocity vector $u^{\mu}$ and
\begin{equation}
\Pi_{\mu \nu} = g_{\mu \nu} + u_\mu u_\nu - k_\mu k_\nu
\end{equation}
is a projection operator onto a two-surface orthogonal to $u^\mu$ and
$k^\mu$.  The quantity $p_{r}=p_{r}(r,\theta)$ is the radial pressure
function, while $q_{t}=q_{t}(r,\theta)$ is the tangential pressure
function, so that
$\sigma = \sigma(r,\theta) \equiv p_{r}(r,\theta)-q_{t}(r,\theta)$ is
an anisotropy parameter function.  Clearly, the limit $\sigma \to 0$
corresponds to isotropic matter~\cite{Doneva:2012rd,Silva:2014fca}.

Are NSs expected to be anisotropic? Clearly, some degree of
anisotropy should be present, for example due to magnetic fields or
superfluidity. But such anisotropy is expected to be small, which
translates to the functional constraint $\sigma/p_{r} \ll
1$.
Precisely how much anisotropy is present in NSs and how this
anisotropy manifests itself mathematically is not clear. The framework
described above should thus be considered a toy model.  We study
it here because it turns out to allow for analytic, closed-form
solutions to the equations of structure in spherical symmetry for the
metric functions without requiring a subsequent post-Minkowskian
expansion.

Let us then work to zeroth order in rotation---in spherical symmetry~\cite{Doneva:2012rd,Silva:2014fca}.
The scalar field evolution equation in Eq.~\eqref{eq:scalar-field-eq-int} 
acquires an anisotropic term $S^{(0)}_{\sigma}$. With the linear scalar coupling in
Eq.~\eqref{eq:g-lin}, $S^{(0)}_{\sigma}$ is given by
\begin{align}
S^{(0)}_\sigma ={}&   \frac{128 \pi  \left[ 2\pi \alpha_{\GB,1} r^3  p
- \alpha_3 M - 2\pi  (\alpha_1 - \alpha_3) r^3 \rho \right]}{r^2 (r-2 M)}   \sigma \nn \\
& {}- \frac{64 \pi^2 (\alpha_{\GB,1}+\alpha_{\GB,2}) r}{(r-2M)}  \sigma^2\,.
\end{align}
Observe that $S^{(0)}_{\sigma}$ does not vanish in the Gauss-Bonnet limit.

Before we can proceed, we must now choose a particular model for
the anisotropy parameter function $\sigma$. We adopt here the model
proposed in~\cite{1974ApJ...188..657B}, namely
\begin{equation}
\sigma =  \frac{\lambda_\BL}{3} (\rho + 3p) (\rho + p) \left( 1-\frac{2M}{r} \right)^{-1} r^2\,,
\end{equation}
where $\lambda_\BL $ is a dimensionless parameter that quantifies the
amount of anisotropy. The isotropic pressure case corresponds to 
setting $\lambda_\BL=0$, while the lower limit on $\lambda_\BL$ is given by
$\lambda_\BL = -2\pi$, beyond which the maximum compactness becomes 
negative and unphysical~\cite{1974ApJ...188..657B}.
Such a model was specifically
constructed so that the equations of structure in spherical symmetry
can be solved in closed-form for an $n=0$ polytropic equation of
state, without requiring a post-Minkowskian
expansion~\cite{1974ApJ...188..657B}.

Let us now solve the scalar field evolution equation for anisotropic
stars.  We further specialize the scalar field equation to the
Gauss-Bonnet limit by setting
$(\alpha_1, \alpha_2, \alpha_3) = \alpha_\GB (1, -4,1)$, where $\alpha_\GB$ is
now the only free parameter; without this specialization, the scalar
field equation is too difficult to solve in closed-form without a
post-Minkowskian expansion. We also take the strongly anisotropic
limit of $\lambda_\BL \to -2\pi$, in which the radial pressure vanishes
throughout the constant density ($n=0$) star. The scalar field equation then simplifies
to
\begin{align}
\frac{d^2 \vartheta^{(0,1/2)}}{d r^2} ={}& - \frac{5 C r^2 - 2 R_*^2}{(2 C r^2 - R_*^2)r} \frac{d \vartheta^{(0,1/2)}}{dr}  \nn \\
&{}- \frac{48 \ell^{2} \alpha_\GB C^2 (C r^2-R_*^2)}{(2 C r^2-R_*^2)^2 R_*^2}\,,
\label{eq:anisotropic-evol}
\end{align}
where we used the small coupling approximation.
In the interior region, the solution to Eq.~\eqref{eq:anisotropic-evol} is
\begin{equation}
\vartheta^{(0,1/2)} =  - \frac{4 \alpha_\GB C}{R_*^2} \ln( R_*^2-2 C r^2 ) + \vartheta_{c}^{(0,1/2)}\,,
\end{equation}
where $\vartheta_{c}^{(0,1/2)}$ is the only integration constant because
we have imposed regularity at the center. 
Matching this interior solution with the exterior one in Eq.~\eqref{eq:0th-Rot-ext-sol-lin} and their
first-derivatives at the stellar surface then automatically forces
$\mu^{(0,1/2)} = 0$.

\section{Corrections to Neutron Star Binary Evolution in
Shift-Symmetric Topological Quadratic Gravity}
\label{app:shift-sym-TQG-corr}

In this appendix, we derive the quadratic gravity correction to a
NS binary within the shift-symmetric, topological class
of quadratic gravity theories. 
As an example, we consider \DDGBText{} gravity.
We look at the conservative and dissipative corrections due to the 
scalar field and metric deformation in turn.

\subsection{Corrections due to the Scalar Field}

Let us first consider corrections to the dynamics of binary systems 
without BHs
caused \emph{directly} by the scalar field.  To do this, we first need
to determine the leading-order asymptotic behavior of the scalar field
at spatial infinity.  When the scalar charge vanishes, 
$\vartheta^{(0,1/2)}$ and $\vartheta^{(2,1/2)}_{0}$ are completely 
sourced by the particular solutions to the scalar field evolution equation
and their asymptotic behavior at spatial infinity becomes
\ba
\vartheta^{(0,1/2)}(r) &=& - 4  \frac{\alpha_\GB}{M_*^2} \left(\frac{M_*}{r}\right)^{4} + \mathcal{O}\left(  \frac{M_*^5}{r^5} \right)\,, \\
\vartheta^{(2,1/2)}_{0}(r) &=& - 8  \, \delta m \frac{\alpha_\GB}{M_*^{2}} \left(\frac{M_*}{r}\right)^{4} + \mathcal{O}\left(  \frac{M_*^5}{r^5} \right)\,. \nn \\
\ea
On the other hand, the asymptotic behavior of $\vartheta^{(2,1/2)}_{2}$ at spatial infinity is
\begin{equation}
\label{eq:theta-22}
\vartheta^{(2,1/2)}_{2}(r) = \mu^{(2,1/2)}_{2}  \left( \frac{M_*}{r} \right)^{3} + \mathcal{O}\left(  \frac{M_*^4}{r^4} \right)\,,
\end{equation}
where $\mu^{(2,1/2)}_{2}$ corresponds to the dimensionless scalar quadrupole charge, 
as predicted in~\cite{Stein:2013wza}.
Clearly, the scalar quadrupole charge is dominant at spatial infinity if it is non-vanishing.

Let us now consider scalar field corrections to the conservative dynamics of binary systems,
i.e.~to the Hamiltonian or the binary's binding energy. The correction
to the latter due to a scalar quadrupole-quadrupole interaction is
given by~\cite{Stein:2013wza}:
\begin{equation}
E^\vartheta_{b} \propto \left( \frac{m}{r_{12}} \right)^{5}\,,
\end{equation}
where $m$ and $r_{12}$ are the total mass and binary separation respectively.
Comparing this to the Newtonian potential ($E_{b} = m/r_{12}$), it is clear that this deformation
enters at relative $4$PN order.

Let us now consider scalar field corrections to the dissipative
dynamics of binary systems, i.e.~to the energy 
fluxes carried away by a dynamical scalar field with quadrupole scalar
charge. Following~\cite{Yagi:2011xp}, one can show that 
the dominant contribution comes from e.g.~the scalar field coupled
to the metric perturbation in $\Box \vartheta$ as
\begin{equation}
\dot E^\vartheta \propto (v_{12})^{18}\,.
\label{eq:dotE-order}
\end{equation}
Comparing this to the leading-order energy flux in general
relativity ($\dot E_\GR \propto (v_{12})^{10}$), it is clear that this
deformation enters at 4PN
order.
\subsection{Corrections due to the Metric Deformation}

Let us now consider corrections to the dynamics of NS/NS binary systems
caused \emph{indirectly} by the scalar field, due to how this induces
a deformation in the metric.

Consider first the conservative sector of the dynamics of
binaries. The binding energy is constructed directly from the metric
tensor in the \emph{near-zone}, i.e.~at distances smaller
than the GW wavelength of the binary. Since the
quadrupole moment of each NS is modified due to a
non-vanishing scalar quadrupole hair, the near-zone metric will
acquire corrections proportional to $1/r_{12}^{3}$, which will then
propagate into the binding energy. This then implies
that the metric deformation causes a 2PN correction to the
energy~\cite{Yagi:2013mbt}.

Consider now the dissipative sector of the dynamics. The dominant effect
of the modification to the energy and angular momentum fluxes carried away by 
the dynamical part of the metric perturbation comes from e.g.~the effective
source term that gives the correct quadrupole moment deformation.
 Following the analysis of~\cite{Yagi:2011xp}, one can
show that such a dynamical metric perturbation will be proportional to
\begin{align}
|\delta h_{ij}| & \propto \frac{m}{r} (v_{12})^{8}\,.
\label{eq:hij}
\end{align}
Comparing this to the magnitude of GWs in general
relativity ($|h_{ij}| \propto (m/r) (v_{12})^{2}$), we conclude
that these deformations induce a modification in the dissipative
dynamics of 3PN relative order.

\bibliography{master}
\end{document}